%% file: main.tex
\RequirePackage{etex}
\documentclass[12pt]{article}
\usepackage[margin=.95in]{geometry} 
\usepackage[T1]{fontenc}    % use 8-bit T1 fonts
\usepackage{lmodern}
\usepackage{nicefrac}       % compact symbols for 1/2, etc.

\usepackage{amssymb}
\usepackage[LGR]{fontenc}
\usepackage[greek,english]{babel}

\usepackage{tcolorbox}

\usepackage{fullpage}
\usepackage{authblk}
\usepackage{complexity}
\usepackage{nicefrac}
\usepackage{gensymb}

\usepackage{enumitem}
\usepackage{microtype}
\usepackage{graphicx}
\usepackage{subfigure}
\usepackage{booktabs} % for professional tables
\usepackage{pifont}

\usepackage{tikz}
\usetikzlibrary{positioning}

\usepackage{bbm}

\usepackage[numbers]{natbib}

\renewcommand{\PPAD}{\textsf{PPAD}\xspace}

\newcommand{\BR}{\operatorname{BR}}

\newcommand{\calB}{\mathcal B}
\newcommand{\X}{\mathcal X}
\newcommand{\Y}{\mathcal Y}

\newcommand{\Z}{\mathbb Z}
\newcommand{\eps}{\varepsilon}

\newcommand{\one}{\mathbf{1}}
\newcommand{\rank}{\operatorname{rank}}

\usepackage{url}
\usepackage{amsmath}
\usepackage{amssymb}
\usepackage{amsthm}
\usepackage{bbm}
\usepackage{algorithm}
\usepackage[noend]{algorithmic}
\usepackage[algo2e,vlined,ruled,linesnumbered]{algorithm2e}
\usepackage{footmisc}
\usepackage[table]{xcolor}
\makeatletter
\newcommand\footnoteref[1]{\protected@xdef\@thefnmark{\ref{#1}}\@footnotemark}
\makeatother
\usepackage{arydshln}

\SetAlFnt{\small}
\SetAlCapFnt{\small}
\SetAlCapNameFnt{\small}
\SetAlCapHSkip{0pt}
\IncMargin{-\parindent}

\usepackage[hidelinks, colorlinks = true, linkcolor=red, citecolor=red]{hyperref}
\usepackage{cleveref}
\newcommand{\declarecolor}[2]{\definecolor{#1}{RGB}{#2}\expandafter\newcommand\csname #1\endcsname[1]{\textcolor{#1}{##1}}}
\declarecolor{White}{255, 255, 255}
\declarecolor{Black}{0, 0, 0}
\declarecolor{Maroon}{128, 0, 0}
\declarecolor{Coral}{255, 127, 80}
\declarecolor{Red}{182, 21, 21}
\declarecolor{LimeGreen}{50, 205, 50}
\declarecolor{DarkGreen}{0, 80, 0}
\declarecolor{Purple}{146, 42, 158}
\declarecolor{Navy}{0, 0, 128}
\declarecolor{LightBlue}{84, 101, 202}
\definecolor{mydarkblue}{rgb}{0,0.08,0.45}
\usepackage{tcolorbox}

\usepackage{xcolor}	% for color

\definecolor{CardinalRed}{HTML}{C41E3A}
\definecolor{Dartmouth}{HTML}{00693E}
\definecolor{SapphireBlue}{HTML}{0F52BA}

\colorlet{MyRed}{CardinalRed}
\colorlet{MyGreen}{Dartmouth}

\colorlet{MyLightRed}{MyRed!25}
\colorlet{MyLightGreen}{MyGreen!25}

\colorlet{AlertColor}{MyRed}	% for alerts
\colorlet{BadColor}{MyRed}	% for good
\colorlet{FocusColor}{MyRed}	% for alerts
\colorlet{GoodColor}{MyGreen}	% for good
\colorlet{MacroColor}{MyRed}	% for macros

\renewcommand{\R}{\mathbb{R}}	% for set of reals

\hypersetup{ %
    pdftitle={},
    pdfkeywords={},
    pdfborder=0 0 0,
    pdfpagemode=UseNone,
    colorlinks=true,
    linkcolor=MyRed,
    citecolor=MyGreen,
    filecolor=Purple,
    urlcolor=blue,
}

\usepackage[%
linewidth=2pt,
linecolor=gray,
middlelinecolor= black,
middlelinewidth=0.4pt,
roundcorner=1pt,
topline = false,
rightline = false,
bottomline = false,
rightmargin=0pt,
skipabove=0pt,
skipbelow=0pt,
leftmargin=0pt,
innerleftmargin=4pt,
innerrightmargin=0pt,
innertopmargin=0pt,
innerbottommargin=0pt,
]{mdframed}

\usepackage{tcolorbox}
\definecolor{TheoremBlueFrame}{RGB}{25,92,170}
\definecolor{LemmaGreenFrame}{RGB}{46,125,50}
\definecolor{TheoremBlueBack}{RGB}{245,249,255}
\definecolor{LemmaGreenBack}{RGB}{247,252,247}

\tcbuselibrary{breakable}

\newenvironment{definitionbox}
  {\begin{tcolorbox}[
    breakable,
    colframe=MyRed,
    colback=red!5,
    top=2pt,
    left=2pt,
    right=2pt,
    bottom=2pt
  ]}
  {\end{tcolorbox}}
\newenvironment{statementbox}
  {\begin{tcolorbox}[
    colframe=TheoremBlueFrame,
    colback=TheoremBlueBack,
    top=2pt,
    left=2pt,
    right=2pt,
    bottom=2pt
  ]}
  {\end{tcolorbox}}

\let\oldnl\nl% Store \nl in \oldnl
\newcommand{\nonl}{\renewcommand{\nl}{\let\nl\oldnl}}% Remove line number for one line

\newcounter{protocol}


\usepackage{array}
\usepackage[utf8]{inputenc} % allow utf-8 input
\usepackage[T1]{fontenc}    % use 8-bit T1 fonts
\usepackage{booktabs}       % professional-quality tables
\usepackage{amsfonts}       % blackboard math symbols
\usepackage{nicefrac}       % compact symbols for 1/2, etc.
\usepackage{microtype}      % microtypography
\usepackage{amsmath}
\usepackage{amssymb}
\usepackage{mathtools}
\usepackage{amsthm}
\usepackage{enumitem}

\usepackage{multirow}

\theoremstyle{plain}
\newtheorem*{theorem*}{Theorem}
\newtheorem{theorem}{Theorem}[section]
\newtheorem{proposition}[theorem]{Proposition}
\newtheorem{lemma}[theorem]{Lemma}

\newtheorem*{lemma*}{Lemma}
\newtheorem{corollary}[theorem]{Corollary}
\newtheorem*{corollary*}{Corollary}
\theoremstyle{definition}
\newtheorem{definition}[theorem]{Definition}
\newtheorem*{definition*}{Definition}

\newtheorem{remark}[theorem]{Remark}
\theoremstyle{observation}

\usepackage[textsize=tiny]{todonotes}

\usepackage{dsfont}

\allowdisplaybreaks

\title{\fontsize{18}{22}\selectfont The Complexity of Computing Nash Equilibria \\ in Colonel Blotto Games}

\author[1,2]{Vasilis Pollatos\textsuperscript{*}}
\author[2,3]{Andreas Kontogiannis\textsuperscript{*}}
\affil[1]{National and Kapodistrian University of Athens}
\affil[2]{Archimedes, Athena Research Center, Greece}
\affil[3]{National Technical University of Athens}

\begin{document}

\maketitle
\renewcommand{\thefootnote}{\fnsymbol{footnote}}
\footnotetext[1]{Equal contribution. Corresponding authors: vaspoll@math.uoa.gr, andreaskontogiannis@mail.ntua.gr.}

\setcounter{footnote}{0}
\renewcommand{\thefootnote}{\arabic{footnote}}

\input{abstract}

\newpage

\tableofcontents

\newpage

\input{intro}

\input{prelims}

\input{ppad-membership}

\input{hardness-2-battles}

\input{hardness-single-threshold}

\input{fixed_rank}

\input{wta_frontier_start_and_positive}
\input{wta_equal_splitting_hardness_only}

\input{conclusion}

\subsection*{Acknowledgements}
The authors acknowledge the financial support of the Archimedes Research Unit (Athena RC) during their PhD. 
The authors would also like to thank Ioannis Panageas, Gabriele Farina, and Panayotis Mertikopoulos for helpful conversations.
Initial versions of the proofs were developed with the assistance of ChatGPT 5.6 Sol and Astra. The authors subsequently streamlined and revised the proofs substantially.

\clearpage

\input{appendix}

\bibliographystyle{alpha}  
\bibliography{main}

\end{document}

%% file: abstract.tex
\begin{abstract}

We study the complexity of Nash equilibrium computation in discrete Colonel Blotto games. For two-player general-sum Colonel Blotto with monotonic piecewise-constant battlefield payoffs, we show that computing an inverse-polynomial approximate Nash equilibrium is \PPAD-hard, even with only two battlefields and a constant number of pieces per battlefield, thereby resolving an open question of \cite{kontogiannis2025}. Allowing more battlefields, the hardness persists when every battlefield payoff is $2$-piecewise constant with respect to either player's allocation. A general \PPAD-membership theorem for multiplayer Blotto with arbitrary local reward functions then implies \PPAD-completeness for both hardness results.

Motivated by fixed-rank bimatrix games, we then identify a tractable frontier within two-player general-sum Colonel Blotto. If the social payoff has rank-$r$ structure, then, for every fixed $r$, an $\varepsilon$-Nash equilibrium can be computed in time polynomial in $B_1,B_2,k$, and $1/\varepsilon$. Thus constant-rank two-player Colonel Blotto remains tractable despite its exponentially large pure strategy spaces.

Finally, we study multiplayer winner-takes-all Colonel Blotto with player-specific battlefield values and uncover a sharp tie-breaking frontier. If every highest bidder receives her full battlefield value, an exact pure Nash equilibrium can be computed in polynomial time for arbitrary binary-encoded budgets. In contrast, under ordinary equal splitting among highest bidders, we prove that computing an inverse-polynomial-accuracy Nash equilibrium is \PPAD-hard even with only five players. This resolves an open question posed in concurrent work by Bichler and Ghosh~\cite{bichlerghosh2026}, who established \PPAD-hardness in the same player-specific equal-splitting setting when the number of players grows with the instance and asked whether hardness persists for a constant number of players.

\end{abstract}

%% file: intro.tex
\section{Introduction}

Colonel Blotto is a classic two-player game introduced by Borel in 1921 (see the English translation \cite{borel1953theory}), and is widely used to model competitive resource allocation \cite{roberson2006colonel, hart2008discrete, kovenock2021generalizations, kroer}. 
In discrete Colonel Blotto games, each player distributes a budget of soldiers among battlefields, and the total payoff is obtained by aggregating the payoffs across battlefields.  
The classical model awards a positive battlefield payoff only to the player who allocates more resources on that battlefield and the total payoff (utility) of each player is the sum of her awarded battlefield payoffs.
Over the years, Colonel Blotto games have been used to model a variety of competitive resource-allocation settings, including electoral competition and security allocation \cite{myerson1993incentives,laslier2002distributive,schwartz2014heterogeneous}. The model has also been studied experimentally as a benchmark for strategic allocation behavior \cite{arad2012multi}.

Computing equilibria in Colonel Blotto games has received significant attention.
Already in the discrete setting, a player with budget $B$ and $k$ battlefields has $\binom{B+k-1}{k-1}$ pure allocation strategies, so the induced bimatrix game may be exponentially large.
Despite this combinatorial explosion, efficient Nash-equilibrium algorithms are known in the classical two-player zero-/constant-sum discrete model \cite{duels, behnezhad2017faster, vu2018simple}. The continuous Blotto literature likewise contains explicit equilibrium characterizations and computational methods for several classical, heterogeneous, and asymmetric variants \cite{roberson2006colonel,kovenock2021generalizations,perchet2022algorithmic}. Behnezhad et al.~\cite{behnezhad2019optimal} also prove complexity results for a substantially more general Blotto-like maximin problem; this is a different objective and model from the ordinary Nash-computation problem considered here. Recent positive results show that other structured extensions can remain tractable: Ka{\'z}mierowski~\cite{kazmierowski2025costs} reduces Colonel Blotto with resource and allocation costs to a zero-sum Colonel Blotto instance with one additional battlefield, while Afiouni et al.~\cite{kroer} study equilibrium computation when individual battlefields themselves host allocation-dependent zero-sum games. By contrast, general-sum bimatrix games are computationally hard \cite{daskalakis2009complexity, chen2009settling}, and the computational status of Nash equilibrium in natural \emph{discrete} general-sum Blotto classes has remained much less understood.

For weaker equilibrium notions, Beaglehole et al.~\cite{beaglehole2023sampling} give efficient no-regret learning for structured Blotto games, yielding approximate coarse correlated equilibria in general-sum and multiplayer settings, while Kontogiannis et al.~\cite{kontogiannis2025} develop efficient CCE learning in a broader class of polyhedral games that includes Colonel Blotto. Related online-learning formulations for Colonel Blotto are studied by \cite{vu2020path, kontogiannis2026efficient}. These results do not, however, resolve the complexity of Nash equilibrium.

In this paper, we first study a two-player competitively monotone specialization of the discrete general-sum framework of \cite{beaglehole2023sampling}: each battlefield payoff is nondecreasing in the player's own allocation, nonincreasing in the opponent's allocation, and piecewise constant in the player's own allocation once the opponent's allocation is fixed. Thus, we impose the additional competitive monotonicity in the opponent's allocation rather than attributing it to the full model of \cite{beaglehole2023sampling}.
This setting is of particular interest, as it can naturally model the classical ``\textit{winner-takes-all}'' competition \cite{kovenock2010conflicts}, by allowing every battlefield payoff to be $2$-piecewise constant, as well as competition scenarios where the magnitude of the utility depends not only on whether a player wins or loses a battlefield, but also on the extent of her allocation (dis)advantage, through a finite collection of \textit{allocation thresholds} and \textit{payoff levels}.
Crucially, Kontogiannis et al. \cite{kontogiannis2025} highlighted the following open question:

\begin{quote}
    \textbf{Question 1: }\textit{Can we efficiently compute $\eps$-Nash equilibria in discrete Colonel Blotto games with monotonic, piecewise constant battlefield payoffs?}
\end{quote}

\paragraph{Hardness results.} We answer the above question in the negative. We show that computing an inverse-polynomial $\varepsilon$-Nash equilibrium in discrete general-sum Colonel Blotto with monotonic, piecewise constant battlefield payoffs is \PPAD-hard. Notably, our hardness continues to hold even for two battlefields and $4$ pieces per battlefield payoff. 
We prove the following theorem:

\begin{statementbox}
\begin{theorem}[Abridged version of Theorem \ref{thm:two-main}]
It is \PPAD-hard to compute an inverse-polynomial $\eps$-Nash equilibrium in Colonel Blotto with exactly two battlefields and monotonic, $4$-piecewise constant battlefield payoffs.
\end{theorem}
\end{statementbox}

The general membership theorem of Section~\ref{sec:ppad-membership} places this approximation problem in \PPAD, so the result is \PPAD-complete. This should be contrasted with the canonical two-battlefield WTA model, for which Macdonell and Mastronardi~\cite{macdonell2015waging} give a complete equilibrium characterization: our hardness with only two battlefields is driven by the richer, though still monotone and constant-complexity, local payoff functions.

By allowing more battlefields, we strengthen the above result showing \PPAD-hardness for $2$-piecewise constant battlefield payoffs for both allocations, in the sense that a battlefield payoff is $2$-piecewise constant \textit{both} with respect to the player's own allocation and with respect to the other player's allocation. 
We prove the following theorem:

\begin{statementbox}
\begin{theorem}[Abridged version of Theorem \ref{thm:single-threshold-main}]
It is \PPAD-hard to compute an inverse-polynomial $\eps$-Nash equilibrium in Colonel Blotto with monotonic battlefield payoffs that are $2$-piecewise constant with respect to both players' allocations.
\end{theorem}
\end{statementbox}

Again, Theorem~\ref{thm:general-blotto-membership} supplies \PPAD membership, so this stronger hardness result is also a \PPAD-completeness result.

Notably, our hardness results stand in contrast to two-player zero-/constant-sum regimes, where efficient algorithms for computing Nash equilibria in discrete Blotto games are known \cite{duels, behnezhad2017faster,vu2018simple}. The hardness already arises under strongly structured local interactions: the battlefield utilities in our reductions are monotone and piecewise constant, and Theorem~\ref{thm:single-threshold-main} strengthens this to a single threshold in each allocation coordinate, i.e., $2$-piecewise constant payoffs with respect to both players' allocations.

\paragraph{Positive result under fixed-rank structure.} 

The above results leave open whether additional global structure can recover tractability inside the two-player general-sum regime. Motivated by positive results for low-rank bimatrix games \cite{kannan2010games, mehta_rank1, patris2024learning} and, more directly, for fixed-rank bilinear games over compact polytopes \cite{garg2011bilinear}, we are naturally led to the following question:

\begin{quote}
    \textbf{Question 2: }\textit{Can we efficiently compute $\eps$-Nash equilibria in Colonel Blotto games with low-rank structure, despite the exponentially large strategy sets?}
\end{quote}

We answer the above question affirmatively, establishing a tractability frontier within two-player general-sum Blotto:

\begin{statementbox}
\begin{theorem}[Abridged version of Theorem \ref{thm:rank-newmain}]
For every fixed $r\ge0$, an
$\eps$-Nash equilibrium of a normalized rank-$r$ Colonel Blotto
game with $k$ battlefields and budgets $B_1,B_2$ can be computed in time $\poly(k,B_1,B_2)\,
\poly_r(1/\eps)$.
Consequently, for every fixed $r$, the algorithm is polynomial in the game parameters and $1/\eps$.
\end{theorem}
\end{statementbox}

Theorem~\ref{thm:rank-newmain} places Colonel Blotto within, and sharpens for this setting, the fixed-rank bilinear paradigm of Garg, Jiang, and Mehta~\cite{garg2011bilinear}. Their FPTAS applies to compact polytopal games of fixed ambient rank $\rank(A+B)$. Our analysis identifies a finer structural parameter: the \emph{intrinsic social-interaction rank}, obtained after restricting the payoff interaction to feasible Blotto marginal directions and removing player-separable terms. We then couple this structural reduction with a normalized equilibrium-gap formulation whose objective range is at most $2$, yielding an absolute $\eps$-Nash guarantee in time $\poly(k,B_1,B_2)\poly_r(1/\eps)$ for fixed intrinsic rank, without the representation-dependent scale factor present in the generic bilinear guarantee. Thus tractability is governed by the effective strategic interaction carried by feasible Blotto marginals, rather than by the rank or numerical scaling of an arbitrary ambient bilinear encoding. Together with the preceding hardness results, this yields a sharp two-player frontier: simple high-rank interactions already support \PPAD-hardness, whereas bounded intrinsic interaction rank restores efficient approximate equilibrium computation.

\paragraph{Multiplayer Winner-Takes-All Colonel Blotto.}
We finally return to the classical winner-takes-all (WTA) structure itself. Classical two-player WTA lies in the tractable constant-sum regime. What happens when this classical WTA structure is extended to multiple players? This leads to our final question:
\begin{quote}
    \textbf{Question 3: }\textit{What is the computational complexity of computing Nash equilibria in multiplayer winner-takes-all Colonel Blotto?}
\end{quote}
The equal-splitting extension has a direct precedent in the multiplayer Blotto literature: Boix-Adser\`a, Edelman, and Jayanti~\cite{boix2021multiplayer} study continuous multiplayer Blotto with common battlefield values and split each battlefield value evenly among the players tied for the highest bid. We study the discrete setting with player-specific battlefield values and place this established equal-splitting convention alongside another natural extension, in which every highest bidder receives her full player-specific value.

Perhaps surprisingly, there is no single answer to the broader WTA question: the complexity depends sharply on how ties are resolved. The two rules lead to opposite computational regimes. We prove the following theorem.

\begin{statementbox}
\begin{theorem}[Abridged version of Theorem~\ref{thm:wta-frontier}]
Consider discrete multiplayer WTA Blotto with nonnegative player-specific battlefield values.
\begin{enumerate}[label=(\alph*)]
    \item If every highest bidder receives her full battlefield value, then an exact pure Nash equilibrium can be computed in time polynomial in the binary input length for arbitrary integer budgets.
    \item If instead each player's battlefield value is split equally among all highest bidders, with all-zero battlefields shared among all players, then computing an inverse-polynomial-accuracy Nash equilibrium is \PPAD-hard even with only five players.
\end{enumerate}
\end{theorem}
\end{statementbox}

The general membership result of Section~\ref{sec:ppad-membership} also applies to the equal-splitting instances, so part~(b) is \PPAD-complete. Thus multiplayer WTA Colonel Blotto exhibits a sharp computational frontier: full rewards to co-winners preserve tractability, whereas the natural equal-splitting rule makes approximate Nash computation \PPAD-complete already with five players.

\paragraph{Concurrent work.}
Concurrent work by Bichler and Ghosh~\cite{bichlerghosh2026} studies the same
multiplayer WTA Blotto setting with player-specific battlefield values and
standard equal splitting among highest bidders.  They prove \PPAD-hardness of
computing a constant-approximate well-supported Nash equilibrium even when
every player has only three resources, and derive \PPAD-hardness of computing
a $(c/n)$-approximate Nash equilibrium, for some constant $c>0$, where the
number of players $n$ grows with the instance.  In particular, their
construction uses one primary player for each variable of a
\textsc{Pure-Circuit} instance, together with several auxiliary players that
implement the circuit gates.  They explicitly leave open whether hardness with
player-specific battlefield values persists when the number of players is
constant. Our five-player result resolves this question affirmatively for inverse-polynomial-accuracy Nash equilibrium: computing such an equilibrium under equal splitting remains \(\PPAD\)-hard with only five players.

The two hardness proofs follow substantially different technical routes.
Bichler and Ghosh~\cite{bichlerghosh2026} reduce from
\textsc{Pure-Circuit}: circuit variables are represented by different primary
players, while additional local players and battlefields enforce the gate
constraints. Consequently, the number of players grows with the size of the
source instance. Our reduction instead starts from a two-player normal-form
game and keeps the number of Blotto players fixed. Source actions are represented
by structured allocations of the two primary players, while the remaining three
players enforce these encodings and regulate how often they are used. Thus the
two reductions place the complexity in different parts of the Blotto instance:
the concurrent work distributes the source computation across a growing
population of players, whereas our
reduction represents the strategic interaction of the source game within the
allocation spaces of only five players. At a high level, both proofs
exploit local WTA interactions and auxiliary reserve-type battlefields to control
unintended behavior, but the source problems and the mechanisms used to encode
them are different.

Bichler and Ghosh also obtain \PPAD-hardness for constant-approximate Nash
equilibrium under non-uniform monotone tie-breaking, even when every player
has a single resource and all players have identical battlefield values.  In
that construction the required payoff differences are generated by the
tie-breaking rule itself.  This result complements the computational frontier
with respect to tie-breaking developed here.  We show that standard equal
splitting is already hard with only five players, whereas under the
full-highest-bidder rule, in which every highest bidder receives the full
battlefield reward, equilibrium computation admits the positive result
described above.  Together, these results highlight the tie-breaking rule as
a central determinant of the computational complexity of multiplayer WTA
Blotto.

%% file: prelims.tex
\section{Preliminaries}
\label{sec:preliminaries}

\paragraph{Notation and Background.}

All numerical inputs are rational and encoded in binary. Polynomial running-time bounds suppress polynomial dependence on their encoding lengths.

For any $t\in \mathbb{Z}_{>0}$, we write $[t]_0=\{0,\ldots,t-1\}$.
For a finite set $S$, we denote by
$\Delta(S)$ the probability simplex over $S$.

% \begin{definition}[Finite normal-form game]
% \label{def:normal-form}
% A finite normal-form game is a tuple
% \[
% G=
% \left(
% N,\,(A_p)_{p\in N},\,(u_p)_{p\in N}
% \right),
% \]
% where $N$ is a finite set of players, $A_p$ is a finite nonempty
% action set for player $p$, and
% \[
% u_p:\prod_{q\in N}A_q\to\mathbb{Q}
% \]
% is player $p$'s payoff function.

% A mixed strategy of player $p$ is a point
% \[
% \sigma_p\in\Delta(A_p),
% \]
% and a mixed-strategy profile is a point in the product of simplices
% \[
% \sigma=(\sigma_p)_{p\in N}
% \in
% \prod_{p\in N}\Delta(A_p).
% \]
% The expected payoff $u_p(\sigma)$ is the multilinear extension of
% $u_p$ under the independent product distribution induced by
% $\sigma$.
% \end{definition}

\begin{definition}[Bimatrix game]
\label{def:bimatrix}
A two-player bimatrix game is specified by payoff matrices
$A,B\in\mathbb{Q}^{m\times n}$.
Player $1$ chooses a row $i\in[m]_0$, player $2$ chooses a column
$j\in[n]_0$, and their payoffs are $A_{ij}$ and $B_{ij}$,
respectively.
The mixed-strategy spaces are the simplices $\Delta([m]_0)$ and
$\Delta([n]_0)$.
Thus, for mixed strategies
$\sigma\in\Delta([m]_0)$ and $\tau\in\Delta([n]_0)$, the expected payoffs are
$u_1(\sigma,\tau)=\sigma^\top A \tau$ and $u_2(\sigma,\tau)=\sigma^\top B \tau$.
\end{definition}

\begin{definition}[$\varepsilon$-Nash equilibrium]
\label{def:approx-nash}
A mixed-strategy profile $(\sigma,\tau)$ is an $\varepsilon$-Nash equilibrium (NE) if, for every mixed strategy
$\sigma'\in\Delta([m]_0)$ and $\tau'\in\Delta([n]_0)$,
\[
u_1(\sigma,\tau)
\ge
u_1(\sigma',\tau)-\varepsilon, \quad u_2(\sigma,\tau)
\ge
u_2(\sigma,\tau')-\varepsilon.
\]
\end{definition}

\begin{theorem}[Positively normalized bimatrix source hardness \cite{chen2009settling}]
\label{thm:s-general-source}
The following search problem is \textsf{PPAD}-hard: given a rational
$d\times d$ bimatrix game $(A,B)$ with
$A,B\in[0,1]^{d\times d}$, compute a $1/d$-Nash equilibrium.
\end{theorem}

\begin{definition}[Win--lose and sparse bimatrix games]
\label{def:sparse-win-lose}
A bimatrix game $(A,B)$ is a win--lose game if $A,B\in\{0,1\}^{m\times n}$.
For an integer $d\ge 1$, it is $d$-sparse if every row and every
column of each of $A$ and $B$ has at most $d$ nonzero
entries.
\end{definition}

The following theorem demonstrates that computing an $\varepsilon$-NE in win-lose sparse bimatrix games is \PPAD-hard. 

\begin{theorem}[Sparse win--lose source hardness \cite{Batziou2026}]
\label{thm:s-source}
There exists a polynomial $f$ such that the following search problem is \textsf{PPAD}-hard: given a $d\times d$ $3$-sparse win--lose bimatrix game $(A,B)$, compute $\sigma,\tau\in\Delta([d]_0)$ satisfying
\[
\max_{i\in[d]_0}\bigl(e_i^\top A\tau-\sigma^\top A\tau\bigr)\le\frac{1}{f(d)}
\quad
\text{and}
\quad
\max_{j\in[d]_0}\bigl(\sigma^\top Be_j-\sigma^\top B\tau\bigr)\le\frac{1}{f(d)}.
\]
\end{theorem}

\paragraph{Colonel Blotto games.}

In this paper, we use a multiplayer definition of discrete Colonel Blotto and specialize it to two players whenever we study the general piecewise-constant model. This convention lets the two-player sections retain the notation $\mathcal X,\mathcal Y,x,y$ used throughout the paper.

\begin{definition}[Discrete multiplayer Colonel Blotto game]
\label{def:colonel-blotto}
A discrete Colonel Blotto game consists of a player set $N=[K]=\{1,\ldots,K\}$, $k\ge 1$ battlefields, and integer budgets $B_i\ge 0$ for every $i\in N$. Player $i$ has the allocation set
\[
\mathcal A_i
=
\left\{
a_i\in\mathbb Z_{\ge 0}^{k}:\sum_{f=1}^{k}a_{if}=B_i
\right\}.
\]
For each battlefield $f$, player $i$ has a rational local payoff function $u_f^i$ depending on the allocations on that battlefield, and total utility is additive across battlefields. No monotonicity or other structural restriction is imposed at the level of this general definition; the computational representation used for the membership theorem is specified in Section~\ref{sec:ppad-membership}. Thus
\[
U_i(a)=\sum_{f=1}^{k}u_f^i(a_{1f},\ldots,a_{Kf}),
\qquad
a=(a_1,\ldots,a_K)\in\prod_{i\in N}\mathcal A_i.
\]
Unless stated otherwise, utilities are nonnegative and normalized so that $0\le U_i(a)\le 1$ for every player and every pure profile. Mixed strategies are independent distributions $\sigma_i\in\Delta(\mathcal A_i)$, and expected utilities are given by multilinear extension.

When $K=2$, we write $\mathcal X=\mathcal A_1$, $\mathcal Y=\mathcal A_2$, $x=a_1$, and $y=a_2$. We also write the two-player local utilities as
\[
u_f^1(x_f,y_f),\qquad u_f^2(y_f,x_f),
\]
so that
\[
U_1(x,y)=\sum_{f=1}^{k}u_f^1(x_f,y_f),
\qquad
U_2(x,y)=\sum_{f=1}^{k}u_f^2(y_f,x_f).
\]
Thus all subsequent two-player definitions and reductions are literal special cases of the multiplayer model.
\end{definition}

\begin{definition}[$\varepsilon$-Nash equilibrium in multiplayer Blotto]
\label{def:multiplayer-approx-nash}
A mixed profile $\sigma=(\sigma_i)_{i\in N}$ is an $\varepsilon$-Nash equilibrium if, for every player $i$ and every mixed deviation $\sigma_i'\in\Delta(\mathcal A_i)$,
\[
U_i(\sigma)\ge U_i(\sigma_i',\sigma_{-i})-\varepsilon.
\]
For $K=2$, this coincides with Definition~\ref{def:approx-nash}.
\end{definition}

\begin{definition}[Winner-takes-all sharing rule]
\label{def:wta-sharing}
Fix player-specific nonnegative battlefield values $v_{if}$. For a pure allocation profile $a$, let
\[
m_f(a)=\max_{j\in N}a_{jf}.
\]
For positive maxima, the set of highest bidders is
\[
S_f(a)=\{j\in N:a_{jf}=m_f(a)\}.
\]
For an all-zero battlefield we allow either standard convention: all players are treated as co-winners, or the battlefield awards zero to everyone. Accordingly, let $W_f(a)=S_f(a)$ when $m_f(a)>0$, and let $W_f(a)$ be either $N$ or $\varnothing$ when $m_f(a)=0$.

A \emph{co-winner sharing rule} is a family of functions
\[
\tau_{if}:2^N\to\mathbb Q_{\ge0}
\]
used only on sets $S$ containing $i$. The induced winner-takes-all payoff is
\[
U_i^{\tau}(a)
=
\sum_{f=1}^{k}v_{if}\,\tau_{if}(W_f(a))\,\one\{i\in W_f(a)\}.
\]
We will focus on two natural rules:
\begin{align*}
\text{full co-winner reward:}\qquad &\tau_{if}^{\mathrm{full}}(S)=1,\\
\text{equal splitting:}\qquad &\tau_{if}^{\mathrm{split}}(S)=\frac{1}{|S|}.
\end{align*}
Thus full reward gives every highest bidder her entire battlefield value, whereas equal splitting divides each player's battlefield value equally among all highest bidders. Values $v_{if}$ may differ across players and battlefields.
\end{definition}

\paragraph{Relation to existing multiplayer Blotto.}
The equal-splitting rule in Definition~\ref{def:wta-sharing} is the natural tie convention used by Boix-Adser\`a, Edelman, and Jayanti~\cite{boix2021multiplayer} in their multiplayer Colonel Blotto model: each battlefield has a common value and that value is divided evenly among all players tied for the highest bid. Thus, when $v_{if}=v_f$ for every player $i$ and we adopt the all-zero co-winner convention, $\tau^{\mathrm{split}}$ gives exactly the same battlefield payoff rule. Their main model has continuous bids, whereas we work with discrete integer allocations and allow player-specific values $v_{if}$. The full co-winner rule $\tau^{\mathrm{full}}$ is the alternative multiplayer WTA extension that we contrast with equal splitting in Section~\ref{sec:multiplayer-wta}.

\begin{definition}[Competitive monotonicity]
\label{def:competitive-monotonicity}
A player-$1$ battlefield payoff $u_b^1(s,t)$ is competitively monotone
if
$
u_b^1(s+1,t)\ge u_b^1(s,t)
$
whenever $s+1\le B_1$, and
$
u_b^1(s,t+1)\le u_b^1(s,t)
$
whenever $t+1\le B_2$.
We define similarly a player-$2$ battlefield payoff $u_b^2(s,t)$ as competitively monotone.
% A player-$2$ battlefield utility $u_b^2(t,s)$ is competitively monotone
% if
% \[
% u_b^2(t+1,s)\ge u_b^2(t,s)
% \]
% whenever $t+1\le B_2$, and
% \[
% u_b^2(t,s+1)\le u_b^2(t,s)
% \]
% whenever $s+1\le B_1$.
Equivalently, each player's battlefield utility is weakly increasing in her own allocation and weakly decreasing in the opponent's
allocation.
\end{definition}

\begin{definition}[$q$-piecewise constant battlefield payoff]
\label{def:q-piecewise-constant}
Following \cite{beaglehole2023sampling}, we define $u_b^1$ as
\emph{$q$-piecewise constant} if, for every
$t\in\{0,\ldots,B_2\}$, 
% the sequence
% \[
% \bigl(u_b^1(0,t),u_b^1(1,t),\ldots,u_b^1(B_1,t)\bigr)
% \]
% can be partitioned into at most $q$ consecutive constant pieces.
% Equivalently,
\[
1+
\left|
\left\{
s\in\{1,\ldots,B_1\}:
 u_b^1(s,t)\ne u_b^1(s-1,t)
\right\}
\right|
\le q.
\]
The definition for player $2$ is symmetric. A Colonel Blotto game has $q$-piecewise constant battlefield payoffs if every battlefield payoff of both players satisfies this condition. 
% The endpoints of the constant pieces may depend on the fixed opponent allocation.
\end{definition}

% For all reductions, the target precision in the Colonel Blotto instance is inverse polynomial in $d$. Given an approximate equilibrium of the constructed game, the source profile is recovered in polynomial time. 

% \begin{definition}[Normalized Colonel Blotto game]
% \label{def:normalized-blotto}
% A Colonel Blotto game is normalized if
% \[
% 0\le U_p(x,y)\le 1
% \]
% for every player $p\in\{1,2\}$ and every pure allocation profile
% $(x,y)\in\mathcal{X}\times\mathcal{Y}$.
% By multilinearity, the same bounds then hold for all mixed-strategy
% profiles.
% \end{definition}

%% file: ppad-membership.tex
\section{PPAD Membership for General Multiplayer Blotto}
\label{sec:ppad-membership}

Before turning to hardness, we record a general membership result that will apply to all of the hard subclasses considered later.  The argument does not use monotonicity, piecewise-constant structure, winner-takes-all payoffs, or any restriction on the number of players.  Its only requirements are the standard computational ones needed to represent arbitrary local rewards and to evaluate their expectations under independent mixed play.

\paragraph{Polynomial local representation.}
For a battlefield $f$ and player $i$, let $q_{jf}\in\Delta(\{0,\ldots,B_j\})$ denote a distribution over player $j$'s bid on battlefield $f$.  The \emph{local multilinear extension} of $u_f^i$ is
\begin{equation}
\overline u_f^i(q_{1f},\ldots,q_{Kf})
=
\sum_{b_1=0}^{B_1}\cdots\sum_{b_K=0}^{B_K}
 u_f^i(b_1,\ldots,b_K)
 \prod_{j=1}^K q_{jf}(b_j).
\label{eq:local-multilinear-extension}
\end{equation}
We say that a Blotto instance has a \emph{polynomial local representation} if each player's layered allocation network, defined below, has polynomial size in the input, every function \eqref{eq:local-multilinear-extension}, together with its first partial derivatives, is represented by polynomial-size arithmetic circuits with rational constants, and a polynomial-bit rational Lipschitz bound for the resulting payoff circuits is available.  In particular, this condition holds for completely arbitrary local payoff functions given by explicit rational tables (with their encoded payoff magnitudes furnishing the required Lipschitz bound).  It also covers succinctly described local rules whenever their expectation under independent bid distributions and its first derivatives admit polynomial-size arithmetic circuits.  This qualification is purely computational: no monotonicity, winner-takes-all, or other shape restriction is imposed on the local reward functions.  Some representation condition is nevertheless necessary for succinct rewards, because an arbitrary succinct payoff circuit need not admit polynomial-time evaluation of its expectation under a product distribution.

\begin{statementbox}
\begin{theorem}[General PPAD membership]
\label{thm:general-blotto-membership}
Given a normalized discrete multiplayer Colonel Blotto game with a polynomial local representation and a rational accuracy parameter $\varepsilon>0$, computing an $\varepsilon$-Nash equilibrium belongs to \PPAD.
\end{theorem}
\end{statementbox}

\begin{proof}
We reduce the game to a polynomially represented concave game.  For each player $i$, construct a layered directed acyclic graph with nodes
\[
V_i=\{(f,s): f\in\{0,\ldots,k\},\ s\in\{0,\ldots,B_i\}\}.
\]
For every battlefield $f\in\{1,\ldots,k\}$, cumulative allocation $s$, and bid $b\ge0$ with $s+b\le B_i$, add an arc
\[
(f-1,s)\longrightarrow(f,s+b).
\]
Let $z^i_{f,s,b}$ denote the flow on this arc, and let $P_i$ be the polytope of nonnegative unit flows from $(0,0)$ to $(k,B_i)$.  A source-to-sink path chooses exactly one bid on every battlefield and is therefore in one-to-one correspondence with a pure Blotto allocation in $\mathcal A_i$.  Conversely, every point of $P_i$ is a convex combination of such path-incidence vectors by the standard flow-decomposition theorem.  Thus every mixed strategy over $\mathcal A_i$ induces a point of $P_i$, and every point of $P_i$ admits a path decomposition into pure allocations.  The flow polytope is therefore a compact marginal representation of mixed Blotto play; the map from mixed strategies to flows may be many-to-one, which is harmless because additive battlefield payoffs depend only on the induced per-battlefield marginals.

For $z_i\in P_i$, define the induced marginal bid distribution on battlefield $f$ by
\begin{equation}
q_{if}(b)=\sum_{s=0}^{B_i-b}z^i_{f,s,b},
\qquad b\in\{0,\ldots,B_i\}.
\label{eq:flow-bid-marginal}
\end{equation}
Because one unit of flow crosses every layer, $q_{if}$ is a probability distribution.  We now define a continuous game in which player $i$ chooses $z_i\in P_i$ and receives
\begin{equation}
\widehat U_i(z)
=
\sum_{f=1}^k
\overline u_f^i(q_{1f},\ldots,q_{Kf}).
\label{eq:flow-game-payoff}
\end{equation}
For fixed $z_{-i}$, each term in \eqref{eq:flow-game-payoff} is linear in $q_{if}$, because \eqref{eq:local-multilinear-extension} is multilinear, and $q_{if}$ is linear in $z_i$ by \eqref{eq:flow-bid-marginal}.  Hence $\widehat U_i(\cdot,z_{-i})$ is affine, and in particular concave, on the convex strategy set $P_i$.  After negating utilities, this is equivalently a convex-loss game.

The reduction preserves payoffs and unilateral regret exactly.  Indeed, given a mixed Blotto strategy $\sigma_i\in\Delta(\mathcal A_i)$, send every pure allocation in its support to the corresponding path and average the path-incidence vectors; this produces a flow $z_i\in P_i$.  Conversely, a flow decomposition of any $z_i\in P_i$ gives a mixed strategy over allocations.  Under a product profile, players randomize independently, so on each battlefield the joint bid distribution is the product of the marginal distributions in \eqref{eq:flow-bid-marginal}.  Therefore the expected Blotto payoff is exactly \eqref{eq:flow-game-payoff}.  The same correspondence applies to every unilateral deviation.  Consequently, a profile $z$ is an $\varepsilon$-Nash equilibrium of the flow game if and only if any path decompositions of its components form an $\varepsilon$-Nash equilibrium of the original Blotto game.

It remains to check computational size.  The polytope $P_i$ has an explicit rational flow description with $\mathcal{O}(kB_i^2)$ arc variables and constraints.  By the definition of polynomial local representation, these polytopes and the arithmetic circuits computing \eqref{eq:flow-game-payoff} have size polynomial in the input.  Taking the product $P=\prod_i P_i$ gives an explicit strong separation oracle.  A maximal independent family of flow-conservation equalities can be eliminated in polynomial time; equivalently, we work in the affine hull of $P$.  After deleting fixed coordinates and applying an affine rational rescaling if necessary, the resulting rational polytope is full-dimensional in its affine hull and has polynomial-bit well-boundedness parameters by standard rational-polytope bounds.  Together with the value-and-gradient circuits above, the resulting instance is a polynomially represented concave game with a strong separation oracle in the sense of Papadimitriou, Vlatakis-Gkaragkounis, and Zampetakis~\cite{papadimitriou2023concave}.  Their strong-separation PPAD-membership theorem therefore places the problem of computing an $\varepsilon$-Nash equilibrium of the flow game in \PPAD.  Finally, each rational unit flow admits a polynomial-time path decomposition with polynomial support, which gives the required Blotto mixed strategies while preserving the same regret bound.
\end{proof}

\paragraph{Applicability to the hardness results.}
Theorem~\ref{thm:general-blotto-membership} is deliberately stated at the level of arbitrary local reward functions.  For the two-player reductions in Sections~\ref{sec:hardness-2-battles} and~\ref{sec:single-threshold-hardness}, the budgets are polynomially bounded and every local payoff can be evaluated explicitly, so the corresponding local tables and the allocation-flow networks have polynomial size.

For the multiplayer WTA hardness result, we record the specialization that matches the model used in Section~\ref{sec:multiplayer-wta}.

\begin{statementbox}
\begin{corollary}[PPAD membership for equal-splitting WTA Blotto]
\label{cor:wta-split-membership}
Consider a normalized discrete multiplayer WTA Blotto game under the equal-splitting rule $\tau^{\mathrm{split}}$, with the all-zero co-winner convention, nonnegative rational player-specific battlefield values, and integer budgets that are polynomially bounded in the binary input length.  For every positive rational accuracy parameter $\varepsilon$ given in binary, computing an $\varepsilon$-Nash equilibrium belongs to \PPAD.
\end{corollary}
\end{statementbox}

\begin{proof}
It suffices to verify the polynomial local representation required by Theorem~\ref{thm:general-blotto-membership}.  Fix a battlefield $f$, a player $i$, and a bid $b$.  For every opponent $j\ne i$, let
\[
L_j(b)=\Pr[a_{jf}<b],
\qquad
E_j(b)=\Pr[a_{jf}=b].
\]
Conditional on player $i$ bidding $b$, she receives a positive share exactly when no opponent bids above $b$.  If exactly $r$ opponents also bid $b$, equal splitting gives her the fraction $1/(r+1)$.  Hence her expected share of battlefield $f$ is
\begin{equation}
\int_0^1\prod_{j\ne i}\bigl(L_j(b)+tE_j(b)\bigr)\,dt.
\label{eq:equal-split-integral}
\end{equation}
The same formula also covers $b=0$ under the all-zero co-winner convention: then $L_j(0)=0$, so the integral equals
\[
\frac1K\prod_{j\ne i}E_j(0),
\]
which is exactly player $i$'s expected $1/K$ share when all $K$ players bid zero.

The integral in \eqref{eq:equal-split-integral} has a polynomial-size arithmetic-circuit representation.  For a fixed bid $b$, initialize $c_0=1$ and $c_r=0$ for $r>0$.  When opponent $j$ is incorporated, update
\[
c'_r=L_j(b)c_r+E_j(b)c_{r-1},
\qquad c_{-1}:=0.
\]
After all opponents are processed, the integral equals
\[
\sum_{r=0}^{K-1}\frac{c_r}{r+1}.
\]
Thus the expected local payoff of player $i$ is obtained by multiplying this share by $v_{if}$ and averaging over her own bid distribution.  The same recurrence differentiated gate-by-gate gives polynomial-size arithmetic circuits for all first partial derivatives.  Therefore equal-splitting WTA payoffs have a polynomial local representation.

Finally, because every budget is polynomially bounded in the binary input length, player $i$'s layered allocation network has $\mathcal O(kB_i^2)$ arcs and hence polynomial size.  All hypotheses of Theorem~\ref{thm:general-blotto-membership} are satisfied, so computing an $\varepsilon$-Nash equilibrium belongs to \PPAD.
\end{proof}

Corollary~\ref{cor:wta-split-membership} applies directly to the five-player construction of Theorem~\ref{thm:wta-five-player-hardness}.  That construction uses the all-zero co-winner convention, nonnegative rational player-specific values, and polynomially bounded budgets, and its payoffs are normalized before the target equilibrium accuracy is imposed.  Hence its inverse-polynomial-accuracy equilibrium problem lies in \PPAD.  Combined with Theorem~\ref{thm:wta-five-player-hardness}, the five-player equal-splitting problem is therefore \PPAD-complete.

%% file: hardness-2-battles.tex
\section{Hardness: two battlefields and constant pieces}
\label{sec:hardness-2-battles}

We establish \PPAD-hardness through an approximation-preserving reduction from the problem of computing an inverse-polynomial approximate Nash equilibrium in a $3$-sparse win--lose bimatrix game, which has been shown to be \PPAD-hard \cite{Batziou2026}. Given an arbitrary source instance $(A,B)$, we construct in polynomial time a Colonel Blotto game with exactly two battlefields and nonnegative, competitively monotone, $4$-piecewise constant battlefield utilities. We then show that an approximate Nash equilibrium of the constructed Colonel Blotto game can be transformed into an approximate Nash equilibrium of $(A,B)$, with only a polynomial loss in the approximation parameter.
Our first hardness result is the following.

\begin{statementbox}
\begin{theorem}[Two-battlefield hardness]
\label{thm:two-main}
It is \PPAD-hard to compute a Nash equilibrium with inverse-polynomial accuracy in normalized two-player Colonel Blotto games with exactly two battlefields and nonnegative, competitively monotone battlefield payoffs that are $4$-piecewise constant in the sense of Definition~\ref{def:q-piecewise-constant}.
\end{theorem}
\end{statementbox}

By Theorem~\ref{thm:general-blotto-membership}, this approximation problem also belongs to \PPAD. Hence the hardness statement above yields \PPAD-completeness.

The remainder of this section presents the proof of Theorem \ref{thm:two-main}.

\medskip

Let $(A,B)$ be an $m \times n$ $3$-sparse win--lose game. We set $B_1=m-1$, $B_2=n-1$, and use two battlefields. Thus every feasible allocation is already a source action; i.e., $x^i=(i,B_1-i)$ and $y^j=(j,B_2-j)$.
The remaining task is to realize each source payoff matrix as the sum of two competitively monotone battlefield utilities with at most four constant pieces in every own-allocation slice.

\begin{lemma}\label{lem:decomp-main}
Let $a=(a_0,\ldots,a_{L-1})\in\{0,1\}^L$ contain at most three ones. Let $a_{-1}=0$, and define
\begin{align*}
p_i=\sum_{h=0}^{i}\max\{a_h-a_{h-1},0\},\qquad q_i=a_i-p_i.
\end{align*}
Then the following properties hold:
\begin{enumerate}[label=(\roman*)]
\item $a_i=p_i+q_i$ for every $i\in[L]_0$;
\item $p_0,\ldots,p_{L-1}$ is weakly nondecreasing and $q_0,\ldots,q_{L-1}$ is weakly nonincreasing;
\item $0\le p_i\le3$ and $-3\le q_i\le0$ for every $i\in[L]_0$;
\item each of the two sequences has at most four consecutive constant pieces.
\end{enumerate}
\end{lemma}

\begin{proof}
Property (i) is the definition of $q_i$. Set $p_{-1}=0$. For every $i\in[L]_0$,
\[
p_i-p_{i-1}=\max\{a_i-a_{i-1},0\}\ge0,
\]
and
\begin{align*}
q_i-q_{i-1}
&=(a_i-a_{i-1})-(p_i-p_{i-1})\\
&=(a_i-a_{i-1})-\max\{a_i-a_{i-1},0\}\\
&=\min\{a_i-a_{i-1},0\}\le0.
\end{align*}
This proves (ii).

The quantity $p_i$ counts the number of runs of ones that begin at an index at most $i$. Since $a$ contains at most three ones, it has at most three such runs, and therefore $0\le p_i\le3$. If $a_i=0$, then $q_i=-p_i\le0$. If $a_i=1$, the run containing $i$ has already begun, so $p_i\ge1$ and $q_i=1-p_i\le0$. Since $a_i\ge0$ and $p_i\le3$, also $q_i\ge-3$. This proves (iii).

The sequence $p$ changes only when a run of ones begins. The displayed formula for $q_i-q_{i-1}$ shows that $q$ changes only at a transition from one to zero. Each sequence therefore has at most three changes and at most four consecutive constant pieces, proving (iv).
\end{proof}

Lemma~\ref{lem:decomp-main} provides the one-dimensional decomposition used to construct the battlefield payoff. We next apply the lemma separately to the columns of $A$ and the rows of $B$.

More precisely, with $A_{-1,j}=0$, we apply Lemma~\ref{lem:decomp-main} to every column of $A$ and define
\[
P^A_{ij}=\sum_{h=0}^{i}\max\{A_{hj}-A_{h-1,j},0\},
\qquad Q^A_{ij}=A_{ij}-P^A_{ij}.
\]
With $B_{i,-1}=0$, we apply the same lemma rowwise to $B$ and define
\[
P^B_{ij}=\sum_{h=0}^{j}\max\{B_{ih}-B_{i,h-1},0\},
\qquad Q^B_{ij}=B_{ij}-P^B_{ij}.
\]
Thus $A=P^A+Q^A$ and $B=P^B+Q^B$, and the relevant columns or rows of $P$ and $Q$ satisfy Lemma~\ref{lem:decomp-main}. Define
\begin{align*}
F^A_{ij}&=P^A_{ij}-3j+3(n-1),&
G^A_{ij}&=Q^A_{ij}+3j+3,\\
F^B_{ij}&=P^B_{ij}-3i+3(m-1),&
G^B_{ij}&=Q^B_{ij}+3i+3.
\end{align*}
The four battlefield payoffs are
\begin{align}
u_1^1(r,s)&=F^A_{rs},&u_2^1(r,s)&=G^A_{B_1-r,B_2-s},\tag{i}\\
u_1^2(r,s)&=F^B_{s,r},&u_2^2(r,s)&=G^B_{B_1-s,B_2-r}.\tag{ii}
\end{align}

The following lemma verifies that these payoffs satisfy the required battlefield-level restrictions and that their aggregate utilities reproduce the source game up to player-specific additive constants.

\begin{lemma}[Battlefield structure and source-payoff embedding]\label{lem:local-embed-main}
All four payoffs in (i)--(ii) are nonnegative and competitively monotone, and every fixed-opponent own-allocation slice has at most four constant pieces. Moreover, for every source profile $(i,j)$,
\[
U_1(x^i,y^j)=A_{ij}+3n,
\qquad
U_2(x^i,y^j)=B_{ij}+3m.
\]
\end{lemma}
\begin{proof}
By Lemma~\ref{lem:decomp-main}, $P^A_{ij}\in[0,3]$ and $Q^A_{ij}\in[-3,0]$. Hence
\[
F^A_{ij}=P^A_{ij}+3(n-1-j)\ge0,
\qquad
G^A_{ij}=Q^A_{ij}+3j+3\ge0.
\]
For fixed $j$, $F^A_{ij}$ is nondecreasing in $i$ and $G^A_{ij}$ is nonincreasing in $i$. Across adjacent columns, we have that
\[
F^A_{i,j+1}-F^A_{ij}=P^A_{i,j+1}-P^A_{ij}-3\le0,
\]
\[
G^A_{i,j+1}-G^A_{ij}=Q^A_{i,j+1}-Q^A_{ij}+3\ge0,
\]
because the relevant $P^A$ values lie in $[0,3]$ and the $Q^A$ values in $[-3,0]$. Thus $u_1^1$ increases in its first argument and decreases in its second. In $u_2^1$, increasing the own allocation decreases the first index of the nonincreasing array $G^A$, while increasing the opponent allocation decreases the second index, in which $G^A$ is nondecreasing. Hence $u_2^1$ is also competitively monotone. Its own slice is a reversed column of $Q^A$ plus a constant; an own slice of $u_1^1$ is a column of $P^A$ plus a constant. Lemma~\ref{lem:decomp-main} proves the constant-piece bound.

The player-$2$ utilities satisfy the same properties after exchanging rows and columns and applying the decomposition rowwise to $B$. 
In particular, their fixed-opponent own-allocation slices are rows of $P^B$ or reversed rows of $Q^B$, up to additive constants, and hence have at most four constant pieces. Finally, we get
% For player 2, $P^B$ is row-wise nondecreasing and $Q^B$ row-wise nonincreasing. Their ranges give nonnegativity exactly as above. Moreover, we have that
% \[
% F^B_{i+1,j}-F^B_{ij}=P^B_{i+1,j}-P^B_{ij}-3\le0,
% \]
% \[
% G^B_{i+1,j}-G^B_{ij}=Q^B_{i+1,j}-Q^B_{ij}+3\ge0.
% \]
% Along each fixed row, $F^B$ is nondecreasing in $j$ and $G^B$ is nonincreasing in $j$. Hence $u_1^2(r,s)=F^B_{s,r}$ increases with the own argument $r$ and decreases with the opponent argument $s$. In $u_2^2$, both indices are reversed, giving the same conclusion from the monotonicities of $G^B$. Its fixed-opponent own slices are reversed rows of $Q^B$ plus constants, while those of $u_1^2$ are rows of $P^B$ plus constants. Thus they also have at most four constant pieces.
\[
F^A_{ij}+G^A_{ij}=A_{ij}+3n,
\qquad
F^B_{ij}+G^B_{ij}=B_{ij}+3m,
\]
and evaluating (i)--(ii) at $(x^i,y^j)$ proves the desired result.
\end{proof}

% Lemma~\ref{lem:local-embed-main} establishes the required battlefield structure and the affine embedding at pure profiles. 
It remains to verify that the same embedding preserves unilateral payoff differences under mixed strategies.

\begin{proposition}[Exact mixed-strategy and approximation transfer]\label{lem:transfer-two-main}
Under the identification of source mixtures $p,q$ with the corresponding mixtures over $x^i,y^j$, respectively,
\begin{align*}
U_1(x^i,q)-U_1(p,q)&=e_i^\top Aq-p^\top Aq,\\
U_2(p,y^j)-U_2(p,q)&=p^\top Be_j-p^\top Bq.
\end{align*}
After dividing player 1's utilities by $3n+1$ and player 2's by $3m+1$, a normalized $\eps$-equilibrium induces a source equilibrium with error at most $\eps\max\{3n+1,3m+1\}$.
\end{proposition}
\begin{proof}
Linearity and Lemma \ref{lem:local-embed-main} give, for every pair of mixtures,
\[
U_1(p,q)=p^\top Aq+3n,
\qquad
U_2(p,q)=p^\top Bq+3m.
\]
The same calculation with one pure deviation gives $U_1(x^i,q)=e_i^\top Aq+3n$ and $U_2(p,y^j)=p^\top Be_j+3m$. Subtracting proves the displayed identities and, after maximizing over $i$ or $j$, equality of the two players' unnormalized maximum unilateral gains with those in the source game. After normalization these gains are divided exactly by $3n+1$ and $3m+1$, respectively. Therefore, a normalized $\eps$-equilibrium induces source gains at most $(3n+1)\eps$ and $(3m+1)\eps$.
\end{proof}

The preceding lemma shows that normalization changes the two players' unilateral gains only by the factors $3n+1$ and $3m+1$. We can therefore complete the approximation-preserving reduction.

\begin{proof}[Proof of Theorem \ref{thm:two-main}]
Given a source game $(A,B)$, construct the battlefield utilities as discussed above. Player $1$'s pure-profile payoffs lie in $\{3n,3n+1\}$, and player $2$'s lie in $\{3m,3m+1\}$. Divide player $1$'s utilities by $3n+1$ and player $2$'s utilities by $3m+1$.

Apply the construction to a square source instance from Theorem~\ref{thm:s-source}, so $m=n=d$. Let $\eps_0=1/f(d)$ and set
\[
\eps=\frac{\eps_0}{3d+1}.
\]
An $\eps$-Nash equilibrium of the  Colonel Blotto instance has unnormalized unilateral gain at most $\eps_0$ for each player. Proposition~\ref{lem:transfer-two-main} therefore yields an $\eps_0$-Nash equilibrium of the source game. The construction has polynomial size, and Lemmas~\ref{lem:decomp-main} and~\ref{lem:local-embed-main} establish all stated structural properties.
    
\end{proof}

% \begin{proof}
% Given the source game, we construct (3)--(4). The action maps are bijections onto the complete unrestricted Colonel Blotto action sets, and the four tables have polynomial dimensions and entries. If the source tolerance is $\eps_0$, we set the target accuracy for a Nash equilibrium of the constructed Colonel Blotto game equal to
% \[
% \eps=\frac{\eps_0}{\max\{3n+1,3m+1\}}.
% \]
% Lemma~\ref{lem:transfer-two-main} maps any normalized $\eps$-equilibrium to an $\eps_0$-equilibrium of $(A,B)$, while Lemma~\ref{lem:local-embed-main} proves every advertised structural condition. The source problem is PPAD-hard, and the construction is an approximation-preserving polynomial reduction.
% \end{proof}

% The four-piece bound corresponds to the three possible runs of ones in a $3$-sparse source vector. 
Theorem \ref{thm:two-main} controls unilateral incentive complexity, not the number of changes as the opponent allocation varies. 
Our results in the next section remove that limitation.

%% file: hardness-single-threshold.tex
\section{Single-Threshold Hardness}
\label{sec:single-threshold-hardness}

In this section, we obtain a stronger hardness result for battlefield payoffs with fewer constant pieces by allowing more battlefields.

\begin{definition}[Single-threshold battlefield payoff]
\label{def:single-threshold}
A single-threshold battlefield payoff has the form
\[
 u(s,t)=\mathbf{1}\{s\ge c\}
 \left(a+z\mathbf{1}\{t<c\}\right),
\]
where $s$ denotes own allocation, $t$ denotes other player's allocation, $c$ is a nonnegative integer threshold and $a,z\ge0$. Such a battlefield payoff is competitively monotone and is $2$-piecewise constant in either coordinate.
\end{definition}

The following theorem states the resulting single-threshold hardness result.

\begin{statementbox}
\begin{theorem}[Single-threshold hardness]
\label{thm:single-threshold-main}
Computing an inverse-polynomial $\eps$-Nash equilibrium in two-player Colonel Blotto with single-threshold battlefield payoffs as in Definition~\ref{def:single-threshold} is \PPAD-hard.
\end{theorem}
\end{statementbox}

Theorem~\ref{thm:general-blotto-membership} applies to this class as well, so the corresponding inverse-polynomial approximation problem is in \PPAD. Consequently, Theorem~\ref{thm:single-threshold-main} is a \PPAD-completeness result.

\paragraph{Proof overview.}
We give a polynomial-time reduction from a square $d\times d$
bimatrix game $(A,B)$ with payoffs in $[0,1]$
\cite{chen2009settling}.
The construction separates two tasks: embedding the $d$ source
actions into the exponentially large Colonel Blotto action sets, and
excluding all allocations outside this embedding.
For each source row $i$ and column $j$, we define a battlefield $(i,j)$, as well as represented
allocations $x^i$ and $y^j$. The allocation $x^i$ reaches the
prescribed thresholds along row $i$, while $y^j$ reaches those
along column $j$; hence they meet at  battlefield
$(i,j)$.
The thresholds are chosen so that no other subset of pair-indexed
battlefields has the same total threshold as a complete row or
column. 
We define $2d$ additional  battlefields which make all represented allocations use the same budget. Consequently, the represented
allocations are exactly those attaining the maximum possible
reached-threshold total.
A large activation term rewards this total and pointwise dominates
every nonrepresented allocation, while smaller source-dependent
payoff terms reproduce $A$ and $B$ on the represented $d\times d$
subgame. Thus exact equilibria use only represented allocations;
in an approximate equilibrium, pointwise dominance bounds the
probability assigned to all other allocations. Conditioning on the
represented allocations then yields an approximate NE of
the source game. 
% We prove a polynomial-time reduction from $3$-sparse win--lose game \cite{Batziou2026}\footnote{In fact, the reduction in this section only need a bimatrix game with utilities in $[0,1]$.}.
% Let $(A,B)$ be a square $d\times d$ $3$-sparse win--lose game, which will be used as the source game of our reduction. 
% The constructed Colonel Blotto game has many more pure allocations than the source game, so the reduction cannot identify every Colonel Blotto allocation with a source action. 
% Instead, we leverage only $d$ allocations $x^0,\ldots,x^{d-1}$ for player $1$ and $d$ allocations $y^0,\ldots,y^{d-1}$ for player $2$. 
% These allocations represent the row and column actions of the source game.
% We define $d^2$ battlefields for each joint action profile of the source game, as well as $2d$ auxiliary battlefields depending only on either player. 
% The battlefield thresholds are chosen so that the represented allocations are exactly the allocations that reach a prescribed total threshold value. 
% A sufficiently large activation payoff then makes every nonrepresented allocation strategies strictly dominated by every represented allocation strategy. 
% Finally, smaller source-payoff coefficients reproduce $A$ and $B$ on the represented $d\times d$ subgame.
% Next, we present the main proof components.
% The full proof can be found in the supplementary material.

\paragraph{Battlefields and thresholds.}
We first construct thresholds that uniquely identify complete rows and columns.
The battlefield index set is
\[
\mathcal I=([d]_0\times[d]_0)
\mathbin{\dot\cup}\{\rho_i:i\in[d]_0\}
\mathbin{\dot\cup}\{\kappa_j:j\in[d]_0\}.
\]
The following constants separate cardinality, row-index, and column-index information in threshold sums. Let $K=d(d-1)+1$, $\gamma_i=i+Ki^2$ and $\gamma_{\max}=\max_{i\in[d]_0}\gamma_i$. We set $M=d\gamma_{\max}+1$ and $L=d^2+1$.

For each joint action profile of the source game, $(i,j)$, we define a battlefield with threshold:
\begin{equation}
c_{ij}=2\bigl(1+L(\gamma_i+M\gamma_j)\bigr).
\label{eq:pair-threshold-main}
\end{equation}
We also define the row and column threshold sums:
\begin{equation}
R_i=\sum_{j=0}^{d-1}c_{ij},
\qquad
C_j=\sum_{i=0}^{d-1}c_{ij}.
\label{eq:s-row-column-sums}
\end{equation}

\begin{lemma}[Row and column compatibility]
\label{lem:threshold-sums-main}
For every $i,j\in[d]_0$ and every subset $S\subseteq[d]_0\times[d]_0$,
\[\bullet
\sum_{(h,r)\in S}c_{hr}=R_i
\quad\Longrightarrow\quad
S=\{i\}\times[d]_0,
\]
% and
\[\bullet
\sum_{(h,r)\in S}c_{hr}=C_j
\quad\Longrightarrow\quad
S=[d]_0\times\{j\}.
\]
% All thresholds $c_{ij}$ are bounded by a polynomial in $d$.
\end{lemma}
\begin{proof}
We use the elementary implication
\[
a+Nb=0,\quad a,b\in\mathbb Z,\quad |a|<N
\quad\Longrightarrow\quad a=b=0.
\tag{E}
\]
Suppose that $\sum_{(h,r)\in S}c_{hr}=R_i$. Expanding \eqref{eq:pair-threshold-main} gives
\[
0=(|S|-d)+L(\Delta_1+M\Delta_2),
\]
where
$\Delta_1=\sum_{(h,r)\in S}\gamma_h-d\gamma_i$ and  $\Delta_2=\sum_{(h,r)\in S}\gamma_r-\sum_{j=0}^{d-1}\gamma_j$.
Since $||S|-d|\le d^2<L$, implication (E) yields $|S|=d$ and $\Delta_1+M\Delta_2=0$. Both $\sum_{(h,r)\in S}\gamma_h$ and $d\gamma_i$ lie in $[0,d\gamma_{\max}]$, so $|\Delta_1|\le d\gamma_{\max}<M$. A second application of (E) gives $\Delta_1=\Delta_2=0$. If $h_1,\ldots,h_d$ are the first coordinates of the elements of $S$, then $\sum_{a=1}^d\gamma_{h_a}=d\gamma_i$.
Substituting $\gamma_h=h+Kh^2$ gives
\[
\left(\sum_{a=1}^dh_a-di\right)
+K\left(\sum_{a=1}^dh_a^2-di^2\right)=0.
\]
The first parenthesis has absolute value at most $d(d-1)<K$, so (E) makes both parentheses zero. Therefore
$\sum_{a=1}^d(h_a-i)^2=0$,
and hence every $h_a=i$. Since $S$ contains $d$ distinct pairs, it is the complete $i$-th row. The column statement follows by the same argument applied to the column indices. 
Finally, $\gamma_{\max}=\mathcal{O}(d^4)$, $M=\mathcal{O}(d^5)$, and \eqref{eq:pair-threshold-main} gives $c_{ij}=\mathcal{O}(d^{11})$. Hence $R_i,C_j=\mathcal{O}(d^{12})$.
\end{proof}
Now, let
$N=\max\{R_i,C_j:i,j\in[d]_0\}$ and define the common soldier budget $\mathcal B=2N+1$. In particular, $\mathcal B=\mathcal{O}(d^{12})$, so every threshold and the budget have $\mathcal{O}(\log d)$-bit binary encodings.
For the auxiliary $2d$ battlefields, we assign thresholds
\begin{equation}
\bar c_i=\mathcal B-R_i
\quad\text{to }\rho_i,
\qquad
\widetilde c_j=\mathcal B-C_j
\quad\text{to }\kappa_j.
\label{eq:s-balancing-thresholds}
\end{equation}

\paragraph{Allocations representing source actions.}

For every source row $i$, we define the allocation strategy $x^i\in\X$ by
\begin{equation}
x^i_{(h,r)}=
\begin{cases}
c_{ir},&h=i\\
0,&h\ne i
\end{cases}
,\quad
x^i_{\rho_h}=
\begin{cases}
\bar c_i,&h=i\\
0,&h\ne i
\end{cases}
,\quad
x^i_{\kappa_r}=0.
\label{eq:s-row-allocation}
\end{equation}
For every source column $j$, define $y^j\in\Y$ by
\[
y^j_{(h,r)}=
\begin{cases}
c_{hj},&r=j\\
0,&r\ne j
\end{cases}
,\quad
y^j_{\kappa_r}=
\begin{cases}
\widetilde c_j,&r=j\\
0,&r\ne j
\end{cases}
,\quad
y^j_{\rho_h}=0.
\]
Thus $x^i$ allocates exactly the threshold amounts along the complete $i$-th row and on $\rho_i$, while $y^j$ does so along the complete $j$-th column and on $\kappa_j$. It is easy to check that each allocation uses exactly $\mathcal B$ soldiers. Next, we define the set of represented allocation strategies induced by the source game: 
\[
\X^\star=\{x^0,\ldots,x^{d-1}\},
\qquad
\Y^\star=\{y^0,\ldots,y^{d-1}\}.
\]
For an arbitrary allocation $x\in\X$, we define
\begin{equation}\label{eq:s-theta-one}
\Theta_1(x)=
\sum_{i,j}c_{ij}\one\{x_{(i,j)}\ge c_{ij}\}
+\sum_i\bar c_i\one\{x_{\rho_i}\ge\bar c_i\},
\end{equation}
and, similarly we define $\Theta_2(y)$ for $y\in\Y$.
% \[
% \Theta_2(y)=
% \sum_{i,j}c_{ij}\one\{y_{(i,j)}\ge c_{ij}\}
% +\sum_j\widetilde c_j\one\{y_{\kappa_j}\ge\widetilde c_j\}.
% \]
These quantities sum the fixed threshold values reached by the allocation: allocating more than a threshold still contributes only the threshold itself.
The next lemma shows that the represented allocations are exactly those attaining the maximum possible value of $\Theta_p$.
% Next, we prove an important lemma demonstrating that the represented allocations are dominating strategies with respect to $\Theta_1$ and $\Theta_2$.

\begin{lemma}[Characterization of represented allocations]
\label{lem:code-main}
For every $x\in\X$ and $y\in\Y$,
it holds that
$\Theta_1(x)\le\mathcal B$ and
$\Theta_2(y)\le\mathcal B.$
Moreover,
\[
\Theta_1(x)=\mathcal B
\Longleftrightarrow
x\in\X^\star,\quad
\Theta_2(y)=\mathcal B
\Longleftrightarrow
y\in\Y^\star.
\]
Therefore, every nonrepresented allocation has threshold total at most $\mathcal B-1$.
% The common budget satisfies $\mathcal B=\mathcal{O}(d^{12})$.
\end{lemma}
% \begin{proof}
% For every counted battlefield, its threshold multiplied by the indicator that it is reached is at most the actual allocation on that battlefield. Therefore $\Theta_1(x)\le\sum_bx_b=\mathcal B$, and symmetrically for player $2$.

% Every $c_{ij}$ is even, whereas $\mathcal B$, $\bar c_i$, and $\widetilde c_j$ are odd. Also, each $\bar c_i$ and $\widetilde c_j$ exceeds $\mathcal B/2$. If $\Theta_1(x)=\mathcal B$, parity implies that $x$ reaches at least one threshold $\bar c_i$, and the half-budget inequality implies that it reaches exactly one. The reached pair-indexed thresholds then sum to $R_i$, so Lemma~\ref{lem:threshold-sums-main} forces the complete $i$-th row. Equality in $\Theta_1(x)\le\sum_bx_b$ further implies that every reached coordinate equals its threshold and every other coordinate is zero. Hence $x=x^i$. The converse follows directly from the definition of $x^i$. The player-$2$ argument is symmetric. Integrality gives the $\mathcal B-1$ bound. The polynomial bound follows by summing $d$ thresholds of order $\mathcal{O}(d^{11})$.
% \end{proof}

\begin{proof}
We prove the statements for player $1$. For every $(i,j)$,
\[
c_{ij}\one\{x_{(i,j)}\ge c_{ij}\}\le x_{(i,j)},
\]
and, for every $i$,
\[
\bar c_i\one\{x_{\rho_i}\ge\bar c_i\}\le x_{\rho_i}.
\]
Summing these inequalities and then including the nonnegative coordinates $x_{\kappa_j}$ gives
\begin{equation}
\Theta_1(x)
\le
\sum_{i,j}x_{(i,j)}+\sum_i x_{\rho_i}
\le
\sum_{b\in\mathcal I}x_b
=\calB.
\label{eq:s-theta-budget-chain}
\end{equation}
This proves the first desired inequality.

Every $c_{ij}$ is even by \eqref{eq:pair-threshold-main}; hence each $R_i$ is even. Hence by definition $\calB$ and every $\bar c_i$ are odd. Furthermore,
\begin{equation}
\bar c_i=2N+1-R_i\ge N+1>\frac{2N+1}{2}=\frac{\calB}{2}.
\label{eq:s-balancing-half-budget}
\end{equation}
Assume that $\Theta_1(x)=\calB$. The pair-indexed terms in \eqref{eq:s-theta-one} are even, while $\calB$ is odd. 
Therefore $x$ reaches at least one threshold $\bar c_i$. 
Hence \eqref{eq:s-balancing-half-budget} and the budget constraint imply that $x$ reaches at most one such threshold. Thus, there is a unique $i\in[d]_0$ such that $x_{\rho_i}\ge\bar c_i$.

Let
\[
S(x)=\{(h,r)\in[d]_0\times[d]_0:x_{(h,r)}\ge c_{hr}\}.
\]
The pair-indexed contribution to $\Theta_1(x)$ is
\[
\sum_{(h,r)\in S(x)}c_{hr}
=\calB-\bar c_i
=R_i.
\]
Lemma~\ref{lem:threshold-sums-main} yields $S(x)=\{i\}\times[d]_0$.

Since both ends of \eqref{eq:s-theta-budget-chain} equal $\calB$, equality holds in every coordinate-wise inequality used to obtain it. Thus:
\begin{enumerate}[label=(\alph*)]
\item $x_{(i,j)}=c_{ij}$ for every $j\in[d]_0$;
\item $x_{(h,r)}=0$ whenever $h\ne i$;
\item $x_{\rho_i}=\bar c_i$ and $x_{\rho_h}=0$ whenever $h\ne i$;
\item $x_{\kappa_j}=0$ for every $j\in[d]_0$.
\end{enumerate}
These are exactly the coordinates in \eqref{eq:s-row-allocation}, so $x=x^i\in\X^\star$.

Conversely, if $x=x^i$, then equations~\eqref{eq:s-row-allocation}, \eqref{eq:s-row-column-sums}, and~\eqref{eq:s-balancing-thresholds} give
\[
\Theta_1(x^i)=\sum_{j=0}^{d-1}c_{ij}+\bar c_i=R_i+\bar c_i=\calB.
\]
The proof for player $2$ is identical after exchanging $\rho_i,\bar c_i,R_i$ with $\kappa_j,\widetilde c_j,C_j$ and applying the column statement of Lemma~\ref{lem:threshold-sums-main}. Since all thresholds are integers, a threshold total strictly smaller than $\calB$ is at most $\calB-1$.
\end{proof}

\paragraph{Battlefield payoffs.}
It remains to encode the source payoffs without weakening the preceding allocation-forcing property. For player $1$, define
$w^1_{ij}=1-A_{ij}$,
$W_i^1=\sum_jw^1_{ij}$,
$L_1=\max_iW_i^1$, and 
$\lambda_i=L_1-W_i^1$. 

Similarly for player 2:
$w^2_{ij}=1-B_{ij}$,
$W_j^2=\sum_iw^2_{ij}$,
$L_2=\max_jW_j^2$, and 
$\mu_j=L_2-W_j^2$.

The coefficients $\lambda_i$ and $\mu_j$ equalize these sums across all represented rows and columns. Define
\[
Q_1=\sum_{i,j}w^1_{ij}+\sum_i\lambda_i=dL_1,
\]
and define $Q_2=dL_2$ analogously for player $2$. Finally, set
\[
H=1+\max\{Q_1,Q_2\}.
\]

For the single-threshold construction (Definition \ref{def:single-threshold}), player $1$'s nonzero battlefield utilities are
\begin{align}
 u_{ij}^1(s,t)
 &=\one\{s\ge c_{ij}\}
 \bigl(Hc_{ij}+w^1_{ij}\one\{t<c_{ij}\}\bigr),\notag\\
 u_{\rho_i}^1(s,t)
 &=\one\{s\ge\bar c_i\}
 \bigl(H\bar c_i+\lambda_i\one\{t<\bar c_i\}\bigr),
\label{eq:threshold-utilities-main}
\end{align}
with zero utility on every $\kappa_j$; player $2$ is defined symmetrically using $w^2_{ij},\widetilde c_j,\mu_j$.

\begin{remark}[Properties of the constructed battlefield payoffs]
\label{lem:shape-main}
\textit{Every battlefield payoff in \eqref{eq:threshold-utilities-main} is nonnegative, competitively monotone, and $2$-piecewise constant in either coordinate.}
\end{remark}
% \begin{proof}
% \end{proof}

The following lemma combines the two global properties needed for the reduction: pointwise exclusion of nonrepresented allocations and affine reproduction of the source game on represented profiles.

\begin{lemma}[Dominance and total utilities]
\label{lem:dom-main}
The following statements hold:
\begin{enumerate}
\item For every $i\in[d]_0$, every $x\in\X\setminus\X^\star$, and every $y\in\Y$,
\[
U_1(x^i,y)\ge U_1(x,y)+1.
\]
For every $j\in[d]_0$, every $y\in\Y\setminus\Y^\star$, and every $x\in\X$,
\[
U_2(x,y^j)\ge U_2(x,y)+1.
\]
\item For every $i,j\in[d]_0$,
\[
U_1(x^i,y^j)=H\mathcal B+L_1-1+A_{ij},
\]
\[
U_2(x^i,y^j)=H\mathcal B+L_2-1+B_{ij}.
\]
\end{enumerate}
\end{lemma}
\begin{proof}
Player $p$'s activation component equals $H\Theta_p$, while the remaining nonnegative source-payoff component is at most $Q_p$. By Lemma~\ref{lem:code-main}, a represented allocation has activation component $H\mathcal B$, whereas a nonrepresented allocation has total payoff at most $H(\mathcal B-1)+Q_p\le H\mathcal B-1$. This proves pointwise dominance.

At $(x^i,y^j)$, both players allocate exactly $c_{ij}$ to the unique intersection battlefield $(i,j)$, so neither source-payoff coefficient $w^p_{ij}$ is received there. Player $1$ receives all coefficients on the remaining battlefields of row $i$ and the coefficient $\lambda_i$ on $\rho_i$. Their sum is
\[
\sum_{r\ne j}w^1_{ir}+\lambda_i
=W_i^1-w^1_{ij}+L_1-W_i^1
=L_1-1+A_{ij}.
\]
Adding the activation component proves the player-$1$ identity. The player-$2$ identity is symmetric.
\end{proof}

The next key step transfers approximate NE. 
Pointwise dominance controls the probability assigned to nonrepresented allocations, while conditioning on the represented allocations changes the payoff of any represented strategy by at most $Q_p$ times the mass assigned to nonrepresented allocations.

% Finally, we prove the following key proposition after inverse-polynomial rescaling of the accuracy.
\begin{definitionbox}
\begin{proposition}
\label{lem:conditioning-main}
Let $Z_p=H\mathcal B+Q_p$ for $p\in\{1,2\}$, and normalize player $p$'s total utility by dividing it by $Z_p$. Let $(\sigma,\tau)$ be an $\eps$-Nash equilibrium of the normalized Colonel Blotto game, where
$\eps\le\frac{1}{2\max\{Z_1,Z_2\}}$.
Let $\delta_1=\sigma(\X\setminus\X^\star)$ and
$\delta_2=\tau(\Y\setminus\Y^\star)$.
Then
\[
\delta_1\le\eps Z_1\le 1/2,
\qquad
\delta_2\le\eps Z_2\le 1/2.
\]
Also, we define the distributions $p,q\in\Delta([d]_0)$ as
\[
p_i=\frac{\sigma(x^i)}{1-\delta_1},
\qquad
q_j=\frac{\tau(y^j)}{1-\delta_2}.
\] 
Then, for every $i,j\in[d]_0$,
\[
e_i^\top Aq-p^\top Aq
\le\eps Z_1+2Q_1\eps Z_2, 
\]
\[
p^\top Be_j-p^\top Bq
\le\eps Z_2+2Q_2\eps Z_1.
\]
\end{proposition}
\end{definitionbox}

\begin{proof}
% The detailed proof appears in the supplementary material. The pointwise inequalities in Lemma~\ref{lem:dom-main} first imply $\delta_p\le\eps Z_p$. The chosen accuracy therefore ensures $\delta_p\le1/2$, so conditioning is always defined. Replacing each player's full strategy by its conditional distribution on the represented allocations cannot reduce that player's payoff. Replacing the opponent's strategy by its conditional distribution changes the payoff of a represented allocation by at most $Q_p\delta_{-p}$. Applying this bound to both sides of each approximate best-response inequality and then using the affine identities of Lemma~\ref{lem:dom-main} gives the displayed source-game inequalities.
Because $(\sigma,\tau)$ is an $\eps$-Nash equilibrium for the normalized payoffs, for every $\sigma'\in\Delta(\X)$ and $\tau'\in\Delta(\Y)$,
\begin{equation}
U_1(\sigma,\tau)\ge U_1(\sigma',\tau)-\eps Z_1,
\label{eq:s-unnormalized-row-ne}
\end{equation}
% and
\[
U_2(\sigma,\tau)\ge U_2(\sigma,\tau')-\eps Z_2.
\]

Fix $x^0\in\X^\star$. Define $\sigma'\in\Delta(\X)$ by moving all probability that $\sigma$ assigns to $\X\setminus\X^\star$ onto $x^0$ and leaving the probabilities of the represented allocations unchanged. Formally,
\[
\sigma'(x^0)=\sigma(x^0)+\delta_1,
\qquad
\sigma'(x^i)=\sigma(x^i)\quad(i\ne0),
\]
and $\sigma'(x)=0$ for every $x\notin\X^\star$. Lemma \ref{lem:dom-main} implies
\[
U_1(\sigma',\tau)-U_1(\sigma,\tau)\ge\delta_1.
\]
Equation \eqref{eq:s-unnormalized-row-ne} gives the reverse upper bound $\eps Z_1$, so $\delta_1\le\eps Z_1 \le 1/2$. The symmetric construction for player $2$ also gives $\delta_2\le\eps Z_2 \le 1/2$. 
It is easy to see that $p_i$ and $q_j$ are valid distributions over $[d]_0$.
% , because
% \[
% \sum_i p_i=\frac{\sum_i\sigma(x^i)}{1-\delta_1}=1,
% \qquad
% \sum_j q_j=\frac{\sum_j\tau(y^j)}{1-\delta_2}=1.
% \]

Let $\widehat\sigma\in\Delta(\X^\star)$ and $\widehat\tau\in\Delta(\Y^\star)$ be the represented Colonel Blotto strategies induced by $p$ and $q$; i.e., 
$\widehat\sigma(x^i)=p_i$ and
$\widehat\tau(y^j)=q_j.$

If $\delta_1=0$, then $\sigma=\widehat\sigma$ and
$U_1(\widehat\sigma,\tau)-U_1(\sigma,\tau)=0$.

If $\delta_1>0$, let $\sigma^{\mathrm{out}}$ be $\sigma$ conditioned on $\X\setminus\X^\star$. Then
\[
\sigma=(1-\delta_1)\widehat\sigma+\delta_1\sigma^{\mathrm{out}}.
\]
Lemma \ref{lem:dom-main}, by averaging over $x^i\sim\widehat\sigma$, $x\sim\sigma^{\mathrm{out}}$, and $y\sim\tau$, gives
\[
U_1(\widehat\sigma,\tau)\ge U_1(\sigma^{\mathrm{out}},\tau)+1,
\]
and hence
\[
U_1(\widehat\sigma,\tau)-U_1(\sigma,\tau)
=\delta_1\bigl(U_1(\widehat\sigma,\tau)-U_1(\sigma^{\mathrm{out}},\tau)\bigr)
\ge\delta_1.
\]
Thus, in both cases,
\begin{equation}
U_1(\widehat\sigma,\tau)\ge U_1(\sigma,\tau).
\label{eq:s-conditioning-own-strategy}
\end{equation}
Combining \eqref{eq:s-conditioning-own-strategy} with \eqref{eq:s-unnormalized-row-ne}, using the pure deviation $x^i$, yields
\begin{equation}
U_1(\widehat\sigma,\tau)
\ge U_1(x^i,\tau)-\eps Z_1
\quad\text{for every }i\in[d]_0.
\label{eq:s-represented-row-ne}
\end{equation}

For every strategy $z$ supported on $\X^\star$ and every opponent strategy $r\in\Delta(\Y)$, the represented allocation has $\Theta_1=\calB$ and its nonactivation component lies in $[0,Q_1]$. Hence
\begin{equation}
H\calB\le U_1(z,r)\le H\calB+Q_1.
\label{eq:s-represented-payoff-range}
\end{equation}
% If $\delta_2=0$, then $\tau=\widehat\tau$ and the left-hand side of equation~\eqref{eq:s-opponent-conditioning-loss} below is zero. 
Let $\tau^{\mathrm{out}}$ be $\tau$ conditioned on $\Y\setminus\Y^\star$, so
\[
\tau=(1-\delta_2)\widehat\tau+\delta_2\tau^{\mathrm{out}}.
\]
For every $z$ supported on $\X^\star$, \eqref{eq:s-represented-payoff-range} gives
\begin{equation}
\bigl|U_1(z,\tau)-U_1(z,\widehat\tau)\bigr|
\le\delta_2Q_1.
\label{eq:s-opponent-conditioning-loss}
\end{equation}
Indeed, this is immediate when $\delta_2=0$; when $\delta_2>0$, it follows from
\[
U_1(z,\tau)-U_1(z,\widehat\tau)
=\delta_2\bigl(U_1(z,\tau^{\mathrm{out}})-U_1(z,\widehat\tau)\bigr).
\]
We apply \eqref{eq:s-opponent-conditioning-loss} to $z=\widehat\sigma$ and $z=x^i$ in \eqref{eq:s-represented-row-ne}. Then
\begin{align}
U_1(\widehat\sigma,\widehat\tau)
&\ge U_1(\widehat\sigma,\tau)-\delta_2Q_1\nonumber \\
&\ge U_1(x^i,\tau)-\eps Z_1-\delta_2Q_1 \nonumber\\
&\ge U_1(x^i,\widehat\tau)-\eps Z_1-2\delta_2Q_1.\label{ineq:last}
\end{align}
Using Lemma \ref{lem:dom-main} and defining $K_1=H\mathcal B+L_1-1$, we get
\[
U_1(\widehat\sigma,\widehat\tau)=K_1+p^\top Aq,
\qquad
U_1(x^i,\widehat\tau)=K_1+e_i^\top Aq.
\]
Using \eqref{ineq:last}, cancelling $K_1$, and using that $\delta_2\le\eps Z_2$, we get the player-1 desired NE condition for the source game.
The player-$2$ proof is symmetric. 
\end{proof}

\begin{proof}[Proof of Theorem \ref{thm:single-threshold-main}]
Let $\eta=1/d$, the source accuracy in Theorem~\ref{thm:s-general-source}, and set
\[
D:=\max\{Z_1+2Q_1Z_2,\;Z_2+2Q_2Z_1\}.
\]
Request a normalized Colonel Blotto equilibrium with accuracy
\[
\varepsilon
=
\min\left\{
\frac{1}{2\max\{Z_1,Z_2\}},
\frac{\eta}{D}
\right\}.
\]
The construction has polynomial encoding length. Indeed,
$\gamma_{\max}=\mathcal{O}(d^4)$, $M=\mathcal{O}(d^5)$, $c_{ij}=\mathcal{O}(d^{11})$, and
$\mathcal B=\mathcal{O}(d^{12})$. Moreover, because the source payoffs lie in
$[0,1]$, we have $L_1,L_2\le d$, $Q_1,Q_2\le d^2$, and
$H\le d^2+1$. Consequently $Z_1,Z_2,D$ are polynomially bounded in
$d$, and the requested $\varepsilon$ is inverse polynomial with
polynomial bit complexity. Proposition~\ref{lem:conditioning-main}
then returns source strategies whose unilateral gains are at most
$\varepsilon D\le\eta=1/d$. This solves the \PPAD-hard source problem
of Theorem~\ref{thm:s-general-source}, completing the reduction.
\end{proof}

% \begin{remark}
% \textit{The quantities defined in the reduction satisfy
% $\gamma_{\max}=\mathcal{O}(d^4)$,
% $M=\mathcal{O}(d^5)$,
% $c_{ij}=\mathcal{O}(d^{11})$, and
% $\calB=\mathcal{O}(d^{12})$.
% Consequently, every threshold and budget has $\mathcal{O}(\log d)$ binary encoding length. Hence, the construction is polynomial time.}
% \end{remark}

% The final step of the proof requests a Colonel Blotto Nash equilibrium with $\eps=\min\left\{\frac{1}{2\max\{Z_1,Z_2\}},
% \frac{\eps_0}{D}
% \right\}$.
% \]

% \begin{proof}
% Let $\eps_0$ be the inverse-polynomial source tolerance from the sparse win--lose hardness theorem, and put
% \[
% D=\max\{Z_1+2Q_1Z_2,\ Z_2+2Q_2Z_1\}.
% \]
% We request a Colonel Blotto equilibrium with
% \[
% \eps=\min\left\{
% \frac{1}{2\max\{Z_1,Z_2\}},
% \frac{\eps_0}{D}
% \right\}.
% \]
% All thresholds, budgets, payoff values, normalization factors, and the reciprocal of the requested accuracy have polynomial encoding length. Lemma~\ref{lem:conditioning-main} produces well-defined source strategies $p,q$ whose two unilateral gains are at most $\eps D\le\eps_0$. Lemmas~\ref{lem:shape-main} and~\ref{lem:dom-main} establish the claimed utility classes and payoff embedding. Hence both constructions are polynomial approximation-preserving reductions.
% \end{proof}

% \paragraph{Comparison of the hardness statements.}

%% file: fixed_rank.tex
\section{Efficient Computation under Low Rank}
\label{sec:low-rank}

Unlike the hardness sections, no monotonicity or piecewise-constancy restriction is imposed in this section.

\subsection{Representing Colonel Blotto as a Layered Graph}

\paragraph{Layered Graph Representation.}

We use the layered allocation-flow representation of \cite{duels,behnezhad2017faster}.
Without loss of generality, assume that $B_1=B_2=B$.  For each player, define the reduced allocation DAG $G(k,B)=(V,E)$ by
\[
V=\{(0,0),(k,B)\}\cup\{(b,h):1\le b<k,\ 0\le h\le B\}.
\]
For each $b\in\{1,\ldots,k\}$, let an edge
\[
e:(b-1,h)\longrightarrow(b,h+t)
% \tag{C.1}
\]
be labeled $\lambda(e)=t$ whenever both endpoints occur in $V$. 
Thus an edge entering the sink must end at cumulative allocation $B$. Every edge lies on a source--sink path. For $k\ge2$, it holds that $|E|=\mathcal{O}(kB^2)$.

We orient the node--edge incidence matrix so an edge contributes $+1$ at its tail and $-1$ at its head. Delete the sink row and call the resulting full-row-rank matrix $D$. Let $d \in \mathbb{R}^{|V|}$ have value one at the source and zero on every retained non-source node. The unit-flow polytope is
\begin{equation}
F(k,B)=\{f\in\R_{\ge0}^{E}:Df=d\}.
\label{C.2}
\end{equation}
For player $p$, write $G_p,D_p,d_p,E_p,F_p$ for the objects obtained with budget $B_p$.

\begin{proposition}[Path--allocation correspondence]\label{prop:s-path}
Source--sink paths in $G(k,B)$ are in bijection with pure allocations $z\in\Z_{\ge0}^k$ satisfying $\sum_bz_b=B$. Every probability distribution over allocations induces a flow in $F(k,B)$, every flow is induced by a distribution over paths, and every rational flow has a rational decomposition using at most $|E|$ paths.
\end{proposition}

\begin{proof}
    See Appendix \ref{app: 5.1}.
\end{proof}

\paragraph{Marginals and bilinear payoffs.}

Define the linear marginal operator
\begin{align*}
T_p:\R^{E_p}\to\R^{k(B_p+1)},\qquad
(T_pf)_{b,s}=\sum_{\substack{e\in E_p:\ e\text{ crosses layer }b\\\lambda(e)=s}}f_e.
\end{align*}

Equivalently, we can think of
$T_p$ as a matrix in $\mathbb{R}^{\,k(B_p+1)\times |E_p|}$, where each row corresponds to a pair \((b,s)\), and each column corresponds to an edge \(e\). 
The entries are defined by
\[
(T_p)_{(b,s),e}
=
\begin{cases}
1, & \text{if \(e\) crosses layer \(b\) and \(\lambda(e)=s\),}\\
0, & \text{otherwise.}
\end{cases}
\]

Therefore, if a path is sampled from a distribution inducing $f$, then $(T_pf)_{b,s}$ is the probability that its allocation on battlefield $b$ equals $s$.

For each battlefield \(b\), define the matrices
\[
S_b^1
:=
\bigl[u_b^1(s,t)\bigr]_
{\substack{s=0,\ldots,B_1\\ t=0,\ldots,B_2}},
\qquad
S_b^2
:=
\bigl[u_b^2(t,s)\bigr]_
{\substack{s=0,\ldots,B_1\\ t=0,\ldots,B_2}}.
\]
In both matrices, rows are indexed by player~1's allocation \(s\) and
columns by player~2's allocation \(t\).  Define the block-diagonal
matrices
\[
S^1:=\operatorname{diag}(S_1^1,\ldots,S_k^1),
\qquad
S^2:=\operatorname{diag}(S_1^2,\ldots,S_k^2).
\]
Thus $S^1,S^2
\in
\mathbb{R}^{\,k(B_1+1)\times k(B_2+1)}$.
The block-diagonal structure reflects the additivity of the Colonel
Blotto payoff across battlefields: only the two players' allocation
marginals on the same battlefield interact.

Since \(T_1f\) and \(T_2g\) are the vectors of battlefield-allocation
marginals induced by flows \(f\) and \(g\), respectively, player~1's
expected payoff can be written as function of $f$ and $g$ as follows
\[
U_1(f,g)=
(T_1f)^\top S^1(T_2g)
=
\sum_{b=1}^k
\sum_{s=0}^{B_1}
\sum_{t=0}^{B_2}
(T_1f)_{b,s}(T_2g)_{b,t}u_b^1(s,t).
\]
and similarly for $U_2(f,g)$.
Accordingly, define
\[
A_F:=T_1^\top S^1T_2,
\qquad
B_F:=T_1^\top S^2T_2.
\]
Then, we have that
\begin{equation}
U_1(f,g)=f^\top A_Fg,
\qquad
U_2(f,g)=f^\top B_Fg.
\label{eq: utility_rank}
\end{equation}
Hence a flow pair is an $\eps$-Nash equilibrium if and only if any inducing independent path distributions form an $\eps$-Nash equilibrium of the Colonel Blotto game.

\paragraph{Exact best responses.}

For fixed $g\in F_2$, define player $1$'s expected battlefield payoff
\begin{align*}
R_b^1(s;g)=\sum_{t=0}^{B_2}(T_2g)_{b,t}u_b^1(s,t).
\end{align*}

For every pure allocation $x$, $U_1(x,g)=\sum_bR_b^1(x_b;g)$. We define
\begin{align*} 
DP_1(0,0)=0,\qquad DP_1(0,h)=-\infty\quad(h>0),
\end{align*}

\begin{align*} 
DP_1(b,h)=\max_{0\le s\le h}
\{DP_1(b-1,h-s)+R_b^1(s;g)\}.
\end{align*}
The quantity \(DP_1(b,h)\) denotes the maximum expected payoff that player~1 can obtain by allocating exactly \(h\) soldiers among the first \(b\) battlefields.

We define player $2$'s recurrence symmetrically using
\[
R_b^2(t;f)=\sum_{s=0}^{B_1}(T_1f)_{b,s}u_b^2(t,s).
\]

\begin{proposition}[Exact deviation oracle]\label{prop:s-br}
The values $DP_1(k,B_1)$ and $DP_2(k,B_2)$ are the two best-response values, and optimizing allocations are obtained by backtracking. Forming both payoff arrays and running both dynamic programs uses
$\mathcal{O}(kB_1B_2+kB_1^2+kB_2^2)$ rational arithmetic operations.
\end{proposition}

\begin{proof}
    See Appendix \ref{app: 5.2}.
\end{proof}

\subsection{Intrinsic social interaction rank}\label{app:rank}

Let
$S=S^1+S^2$
be the block-diagonal social-payoff matrix. For a pure allocation strategy $x$ of player $1$, let $\chi_1(x)\in\{0,1\}^{k(B_1+1)}$ have entry one at
$(b,x_b)$ for every battlefield and otherwise zero. For a pure allocation strategy $y$ of player $2$, we define $\chi_2(y)$ analogously. Then, the social-payoff is defined as
\begin{align}
C_{xy}:=U_1(x,y)+U_2(x,y)
=\chi_1(x)^\top S\chi_2(y).
\label{D.1}
\end{align}

\paragraph{Feasible marginal directions.}

Let $\mathcal M_p=T_p(F_p)$ be player $p$'s marginal polytope. So, $\mathcal{M}_p$ is simply the set of all marginal vectors that can actually arise from feasible mixed Blotto strategies.
Let
$$ \mathcal{M}_p - \mathcal{M}_p := \{m-m' : m,m' \in \mathcal{M}_p \}. $$
Let
\[
\mathcal{L}_p=\operatorname{lin}(\mathcal{M}_p-\mathcal{M}_p)
:=
\left\{
\sum_{i=1}^r \alpha_i (m_i-m_i')
:
r\ge 1,\;
\alpha_i\in\mathbb{R},\;
m_i,m_i'\in\mathcal{M}_p
\right\}.
\]
be the vector space of all directions generated by moving between feasible marginal vectors.

The following lemma ensures that a rational basis of $\mathcal{L}_p$ can be computed efficiently from $D_p$ and $T_p$.

\begin{lemma}[Feasible flow and marginal directions]\label{lem: feasible_dir}
For every player \(p\),
\[
\operatorname{lin}(F_p-F_p)=\ker D_p,
\]
and consequently
\[
\mathcal L_p
=
\operatorname{lin}(\mathcal M_p-\mathcal M_p)
=
T_p(\ker D_p).
\]
In particular, a rational basis of \(\mathcal L_p\) can be computed in
polynomial time from \(D_p\) and \(T_p\).
\end{lemma}

\begin{proof}
    See Appendix \ref{app: 5.3}.
\end{proof}

% \begin{proof}
% Fix a player \(p\), and write \(G_p=(V_p,E_p)\) for the corresponding
% reduced allocation DAG, with source \(s_p\) and sink \(t_p\).
% Recall that
% \[
% F_p=\{f\in\mathbb{R}_{\ge 0}^{E_p}:D_pf=d_p\},
% \]
% where \(D_p\) is the node--edge incidence matrix with the sink row
% deleted.

% We first show that \(D_p\) has full row rank.  Suppose that
% \(a^\top D_p=0\), where \(a\in\mathbb{R}^{V_p\setminus\{t_p\}}\).
% Extend \(a\) to all vertices by setting \(a_{t_p}:=0\).
% For every directed edge \(e=(u,v)\in E_p\), the column of the full
% incidence matrix corresponding to \(e\) has a \(+1\) in row \(u\) and
% a \(-1\) in row \(v\).  Hence \(a^\top D_p=0\) implies
% \[
% a_u-a_v=0
% \]
% for every edge \(e=(u,v)\), where \(a_{t_p}=0\) is used when
% \(v=t_p\).  Every vertex of the reduced allocation DAG lies on a
% source--sink path and therefore has a directed path to \(t_p\).
% Applying the preceding equality successively along such a path gives
% \[
% a_v=a_{t_p}=0
% \]
% for every retained vertex \(v\).  Thus \(a=0\), so the rows of \(D_p\)
% are linearly independent.  Hence \(D_p\) has full row rank.

By Lemma \ref{lem: feasible_dir}, we have that $\mathcal L_p=T_p(\ker D_p)$.
and we can compute a rational basis matrix \(Z_p\) for \(\ker D_p\).  Then the
columns of \(T_pZ_p\) span \(\mathcal L_p\).  By rational Gaussian
elimination, extract a full-column-rank rational matrix \(L_p\) whose
columns form a basis of \(\mathcal L_p\).

% For a pure allocation \(x\in X\), define its marginal-incidence vector
% \[
% \chi_1(x)\in\{0,1\}^{k(B_1+1)}
% \]
% by
% \[
% \bigl(\chi_1(x)\bigr)_{b,s}
% :=
% \begin{cases}
% 1, & s=x_b,\\
% 0, & s\neq x_b,
% \end{cases}
% \qquad
% b\in[k],\quad s\in\{0,\ldots,B_1\}.
% \]
% Similarly, for \(y\in Y\), define
% \[
% \bigl(\chi_2(y)\bigr)_{b,t}
% :=
% \begin{cases}
% 1, & t=y_b,\\
% 0, & t\neq y_b.
% \end{cases}
% \qquad
% b\in[k],\quad t\in\{0,\ldots,B_2\}.
% \]
% Equivalently, if \(f^x\in F_1\) and \(g^y\in F_2\) are the unit path
% flows corresponding to \(x\) and \(y\), then
% \[
% \chi_1(x)=T_1f^x,
% \qquad
% \chi_2(y)=T_2g^y.
% \]
Choose arbitrary reference pure allocations \(x^0\in X\) and
\(y^0\in Y\), and define
\[
m_1^0:=\chi_1(x^0)\in\mathcal M_1,
\qquad
m_2^0:=\chi_2(y^0)\in\mathcal M_2.
\]
For \(p\in\{1,2\}\), let \(z\) denote a pure allocation of player \(p\),
and let \(m_p^0\) denote the corresponding reference marginal defined
above. Since
$\chi_p(z),\,m_p^0\in\mathcal M_p$,
we have
\[
\chi_p(z)-m_p^0
\in
\mathcal M_p-\mathcal M_p
\subseteq
\operatorname{lin}(\mathcal M_p-\mathcal M_p)
=
\mathcal L_p.
\]
Because the columns of \(L_p\) form a basis of \(\mathcal L_p\), there
exists a unique coordinate vector \(\theta_z^{(p)}\) such that
\[
\chi_p(z)-m_p^0=L_p\theta_z^{(p)}.
\]
Accordingly, writing
\[
\theta_x^{(1)}:=\xi_x,
\qquad
\theta_y^{(2)}:=\eta_y,
\]
every pair of pure allocation marginals admits the unique coordinate
representations
\begin{equation}
\chi_1(x)=m_1^0+L_1\xi_x,
\qquad
\chi_2(y)=m_2^0+L_2\eta_y.
\label{D.44}
\end{equation}

\paragraph{Rank-$r$ Colonel Blotto.}
The following definition extends the classical notion of rank-$r$ bimatrix games to Colonel Blotto by requiring the social payoff to admit a decomposition into battlefield-additive terms together with a bounded number of separable interactions between the players' aggregate battlefield features.

\begin{definitionbox}
\begin{definition}[Rank-$r$ Colonel Blotto]\label{def:s-rank-blotto}
A Colonel Blotto game has \emph{intrinsic social interaction rank}, or simpler is \textit{rank-$r$},
$r$ if $r$ is the least integer for which there exist functions
$\rho_{pb},\alpha_{\ell b},\beta_{\ell b}$ such that, for every
$(x,y)\in\X\times\Y$,
\begin{equation}
C_{xy}
=\sum_b\bigl(\rho_{1b}(x_b)+\rho_{2b}(y_b)\bigr)+\sum_{\ell=1}^r
\left(\sum_b\alpha_{\ell b}(x_b)\right)
\left(\sum_b\beta_{\ell b}(y_b)\right).
\label{D.55}
\end{equation}

\end{definition}
\end{definitionbox}

\paragraph{Normal-form characterization.}

\medskip
\noindent
Next, we define the matrix
\begin{equation}
K=L_1^\top SL_2 \in \mathbb{R}^{r_1 \times r_2},
\label{D.66}
\end{equation}
where $r_p = \rank(L_p)$ for $p\in \{1,2\}$. 
Thus, $K$ denotes the matrix representing the social-payoff interaction \(S\) restricted to feasible marginal directions.
Notably, changing bases replaces $K$ by $M_1^\top KM_2$ for invertible
matrices $M_1,M_2$, so $\rank(K)$ is basis independent.

For the chosen reference allocations, define the normal-form social-payoff matrix by
\begin{align}
\overline C_{xy}
=
C_{xy}-C_{xy^0}-C_{x^0y}+C_{x^0y^0}.
\label{D.77}
\end{align}
and let
\begin{equation}
\overline C
:=
\bigl[\overline C_{xy}\bigr]_{x\in X,\;y\in Y}
\in\mathbb{R}^{|X|\times |Y|}
\label{eq: C_stacked}
\end{equation}

The following lemma shows an interesting property regarding the rank of the above matrices.

\begin{lemma}\label{lem:s-double}
For every reference pair $(x^0,y^0)$,
\[
\rank(\overline C)=\rank(K).
\]
\end{lemma}

\begin{proof}
From \eqref{D.1}, we have that
\[
C_{xy}
=
\chi_1(x)^\top S\chi_2(y).
\]
By \eqref{D.44},
\[
\chi_1(x)=m_1^0+L_1\xi_x,
\qquad
\chi_2(y)=m_2^0+L_2\eta_y.
\]
Since \(m_1^0=\chi_1(x^0)\) and \(m_2^0=\chi_2(y^0)\), the uniqueness
of the representations in \eqref{D.44} and the full column rank of
\(L_1,L_2\) imply
\[
\xi_{x^0}=0,
\qquad
\eta_{y^0}=0.
\]
Therefore, by \eqref{D.77} we have that
\begin{align}
\overline C_{xy}
&=
(\chi_1(x)-m_1^0)^\top
S
(\chi_2(y)-m_2^0) \nonumber\\
&=
(L_1\xi_x)^\top S(L_2\eta_y) \nonumber \\
&=
\xi_x^\top L_1^\top S L_2\eta_y \nonumber \\
&=
\xi_x^\top K\eta_y.
\label{D.88} 
\end{align}

% Let \(X_\xi\) be the matrix whose row indexed by \(x\in X\) is
% \(\xi_x^\top\), and let \(Y_\eta\) be the matrix whose row indexed by
% \(y\in Y\) is \(\eta_y^\top\).  Equation (D.88) gives
% \[
% \overline C_{xy}=X_\xi K Y_\eta^\top.
% \tag{D.99}
% \]

\noindent
Define
\[
X_\xi\in\mathbb R^{|X|\times r_1},
\qquad
Y_\eta\in\mathbb R^{|Y|\times r_2},
\]
by
\[
(X_\xi)_{x,i}:=(\xi_x)_i
\qquad
(x\in X,\ i\in[r_1]),
\]
and
\[
(Y_\eta)_{y,j}:=(\eta_y)_j
\qquad
(y\in Y,\ j\in[r_2]).
\]
Equivalently, the row of \(X_\xi\) indexed by \(x\) is
\(\xi_x^\top\), while the row of \(Y_\eta\) indexed by \(y\) is
\(\eta_y^\top\).

Hence, for every \(x\in X\) and \(y\in Y\),
\begin{equation}
\bigl(X_\xi K Y_\eta^\top\bigr)_{x,y}
=
\xi_x^\top K\eta_y
=
\overline C_{xy}. \label{D.99}
\end{equation}
% Therefore,
% \[
% \overline C
% =
% X_\xi K Y_\eta^\top.
% \tag{D.99}
% \]

We next show that \(X_\xi\) and \(Y_\eta\) have full column rank.
Because every feasible flow is a convex combination of pure
source--sink path flows and \(T_p\) is linear, the marginal polytope
satisfies
\[
\mathcal M_1
=
\operatorname{conv}\{\chi_1(x):x\in X\}.
\]
Hence
\begin{equation}
\mathcal L_1
=
\operatorname{lin}(\mathcal M_1-\mathcal M_1)
=
\operatorname{span}\{
\chi_1(x)-m_1^0:x\in X
\}.
\label{D.100}
\end{equation}
Using that 
$\chi_1(x)-m_1^0=L_1\xi_x$
and the fact that the columns of \(L_1\) form a basis of
\(\mathcal L_1\), \eqref{D.100} implies
\[
\operatorname{span}\{\xi_x:x\in X\}
=
\mathbb R^{r_1}.
\]
Thus \(X_\xi\) has full column rank.  The same argument gives
\[
\operatorname{span}\{\eta_y:y\in Y\}
=
\mathbb R^{r_2},
\]
and therefore \(Y_\eta\) also has full column rank.

Since \(X_\xi\) has full column rank, there exists a left inverse
\(A\) such that
\[
AX_\xi=I.
\]
Since \(Y_\eta\) has full column rank, \(Y_\eta^\top\) has full row
rank, and therefore there exists a right inverse \(B\) such that
\[
Y_\eta^\top B=I.
\]
Applying these identities to \eqref{D.99} yields
\[
A\overline C B
=
A X_\xi K Y_\eta^\top B
=
K.
\]
Consequently,
\[
\rank(K)
=
\rank(A\overline C B)
\le
\rank(\overline C).
\]
On the other hand, \eqref{D.99} and the standard rank inequality for a matrix product give
\[
\rank(\overline C)
=
\rank(X_\xi K Y_\eta^\top)
\le
\rank(K).
\]
Combining the two inequalities proves
\[
\rank(\overline C)=\rank(K).
\]
\end{proof}

The next theorem demonstrates that we can efficiently compute the rank and the underlying decomposition of a $\rank$-$r$ Colonel Blotto game without enumerating over the exponentially large strategy spaces of the game.

\begin{definitionbox}
\begin{theorem}[Efficient rank and decomposition computation]
\label{thm:s-certificate}
Let $r=\rank(K)$, where $K$ is given in \eqref{D.66}. Then $r$ is the minimum number of product terms in
any identity of the form
\begin{align}
C_{xy}
=\sum_b\bigl(\rho_{1b}(x_b)+\rho_{2b}(y_b)\bigr)+\sum_{\ell=1}^q
\left(\sum_b\alpha_{\ell b}(x_b)\right)
\left(\sum_b\beta_{\ell b}(y_b)\right)
\label{D.8}
\end{align}
valid for all pure allocations $x \in \mathcal{X}$ and $y \in \mathcal{Y}$. Moreover, rational functions
realizing \eqref{D.8} with $q=r$ can be computed in polynomial time.
\end{theorem}
\end{definitionbox}

\begin{proof}

We first prove that every representation of the form \eqref{D.8}
requires at least \(r\) product terms.  Suppose that \eqref{D.8}
holds with \(q\) product terms.  For \(\ell\in[q]\), define
\[
A_\ell(x):=\sum_b \alpha_{\ell b}(x_b),
\qquad
B_\ell(y):=\sum_b \beta_{\ell b}(y_b).
\]
Also write
\[
R_1(x):=\sum_b \rho_{1b}(x_b),
\qquad
R_2(y):=\sum_b \rho_{2b}(y_b).
\]
Thus
\[
C_{xy}
=
R_1(x)+R_2(y)
+
\sum_{\ell=1}^q A_\ell(x)B_\ell(y).
\]

Fix the reference pair \((x^0,y^0)\).  By \eqref{D.77}, we have that
\[
\overline C_{xy}
=
C_{xy}-C_{xy^0}-C_{x^0y}+C_{x^0y^0}.
\]
In the above quantity, using \eqref{D.8} the one-player terms cancel identically.  For every
\(\ell\in[q]\), the contribution of the \(\ell\)-th product term is
\[
\begin{aligned}
&A_\ell(x)B_\ell(y)
-A_\ell(x)B_\ell(y^0)
-A_\ell(x^0)B_\ell(y)
+A_\ell(x^0)B_\ell(y^0)
\\
&\qquad =
\bigl(A_\ell(x)-A_\ell(x^0)\bigr)
\bigl(B_\ell(y)-B_\ell(y^0)\bigr).
\end{aligned}
\]
Define vectors \(a^\ell\in\mathbb{R}^{|X|}\) and
\(b^\ell\in\mathbb{R}^{|Y|}\) by
\[
a^\ell_x:=A_\ell(x)-A_\ell(x^0),
\qquad
b^\ell_y:=B_\ell(y)-B_\ell(y^0).
\]
Then
\[
\overline C
=
\sum_{\ell=1}^q a^\ell (b^\ell)^\top .
\]
Each outer product \(a^\ell (b^\ell)^\top\) has rank at most one.
Hence, by subadditivity of matrix rank,
\[
\rank(\overline C)
\le
\sum_{\ell=1}^q
\rank\!\left(a^\ell(b^\ell)^\top\right)
\le q.
\]
By Lemma~\ref{lem:s-double},
\[
\rank(\overline C)=\rank(K)=r,
\]
and therefore
$q\ge r$.

\bigskip

We now prove that a representation with exactly \(r\) product terms
exists and can be computed in polynomial time.  Let
$n_p:=k(B_p+1)$.
% so that
% \[
% L_1\in\mathbb{Q}^{n_1\times r_1},
% \qquad
% L_2\in\mathbb{Q}^{n_2\times r_2},
% \qquad
% K=L_1^\top S L_2
% \in\mathbb{Q}^{r_1\times r_2}.
% \]
Since \(K\) is rational and has rank \(r\), exact rational Gaussian
elimination yields a rank factorization
\[
K=PQ^\top,
\qquad
P\in\mathbb{Q}^{r_1\times r},
\qquad
Q\in\mathbb{Q}^{r_2\times r}.
\tag{D.9}
\]
Writing \(p_\ell\) and \(q_\ell\) for the \(\ell\)-th columns of
\(P\) and \(Q\), respectively, this is equivalently
\[
K=\sum_{\ell=1}^r p_\ell q_\ell^\top.
\]

Because \(L_1\) has full column rank \(r_1\), the matrix
\(L_1^\top\) has full row rank \(r_1\).  Consequently the linear map
\[
L_1^\top:\mathbb{Q}^{n_1}\to\mathbb{Q}^{r_1}
\]
is surjective.  Hence, for every \(\ell\in[r]\), there exists
\(a_\ell\in\mathbb{Q}^{n_1}\) satisfying
\[
L_1^\top a_\ell=p_\ell.
\]
Likewise, since \(L_2^\top\) has full row rank, there exists
\(b_\ell\in\mathbb{Q}^{n_2}\) satisfying
\begin{equation}
L_2^\top b_\ell=q_\ell. \nonumber
% \label{D.10}
\end{equation}
These solutions can be computed by exact rational Gaussian
elimination.

Define
\begin{equation}
R
:=
S-\sum_{\ell=1}^r a_\ell b_\ell^\top.
\label{D.11}
\end{equation}
Then
\begin{align}
L_1^\top R L_2
&=
L_1^\top S L_2
-
\sum_{\ell=1}^r
(L_1^\top a_\ell)
(L_2^\top b_\ell)^\top \nonumber
\\
&=
K-\sum_{\ell=1}^r p_\ell q_\ell^\top
=0.
\label{D.12}
\end{align}

We next convert the rank-one matrices
\(a_\ell b_\ell^\top\) into the desired battlefield product terms.  Index the coordinates of \(a_\ell\) by pairs \((b,s)\), where
\(b\) is a battlefield and \(s\in\{0,\ldots,B_1\}\), and index the
coordinates of \(b_\ell\) analogously by \((b,t)\).  Define
\begin{equation}
\alpha_{\ell b}(s):=(a_\ell)_{b,s},
\qquad
\beta_{\ell b}(t):=(b_\ell)_{b,t}.
\label{D.13}
\end{equation}
Since the pure marginal vector \(\chi_1(x)\) has exactly one nonzero
entry in each battlefield block, namely at \((b,x_b)\), we have
\[
\chi_1(x)^\top a_\ell
=
\sum_b \alpha_{\ell b}(x_b).
\]
Similarly,
\[
b_\ell^\top\chi_2(y)
=
\sum_b \beta_{\ell b}(y_b).
\]
Consequently,
\begin{equation}
\chi_1(x)^\top a_\ell b_\ell^\top\chi_2(y)
=
\left(\sum_b\alpha_{\ell b}(x_b)\right)
\left(\sum_b\beta_{\ell b}(y_b)\right).
\label{D.14}
\end{equation}

It remains to show that the residual bilinear form induced by \(R\) does not contain any two-player interaction on feasible marginals.
Let
\[
u:=Rm_2^0\in\mathbb{Q}^{n_1},
\qquad
v:=R^\top m_1^0\in\mathbb{Q}^{n_2},
\qquad
c:=(m_1^0)^\top Rm_2^0.
\]
For fixed $s\in [B_1]$, define
\[
\rho_{1b}(s):=u_{b,s}.
\]
Choose an arbitrary fixed battlefield \(b^\star\) and for any fixed $t\in[B_2]$  define
\begin{equation}
\rho_{2b}(t)
:=
\begin{cases}
v_{b,t}, & b\neq b^\star,\\
v_{b^\star,t}-c, & b=b^\star.
\end{cases}
\label{D.15}
\end{equation}
Because \(\chi_1(x)\) contains exactly one unit entry in every
battlefield block,
\begin{equation}
\sum_b\rho_{1b}(x_b)
=
\chi_1(x)^\top Rm_2^0.
\label{D.16}
\end{equation}
Likewise, \(\chi_2(y)\) contains exactly one unit entry in the
\(b^\star\)-block, so subtracting \(c\) from every value in that block
subtracts exactly \(c\), independently of \(y\).  Hence
\begin{equation}
\sum_b\rho_{2b}(y_b)
=
(m_1^0)^\top R\chi_2(y)-c.
\label{D.17}
\end{equation}
Using the coordinate representations
\[
\chi_1(x)=m_1^0+L_1\xi_x,
\qquad
\chi_2(y)=m_2^0+L_2\eta_y,
\]
we obtain
\[
\begin{aligned}
\chi_1(x)^\top R\chi_2(y)
={}&
\chi_1(x)^\top Rm_2^0
+
(m_1^0)^\top R\chi_2(y)
-
(m_1^0)^\top Rm_2^0
+
\xi_x^\top L_1^\top R L_2\eta_y.
\end{aligned}
\]
By \eqref{D.12}, the final term vanishes.  
Combining this with
\eqref{D.16}--\eqref{D.17} yields
\begin{equation}
\chi_1(x)^\top R\chi_2(y)
=
\sum_b\rho_{1b}(x_b)
+
\sum_b\rho_{2b}(y_b).
\label{D.18}
\end{equation}
Finally, by \eqref{D.11},
\[
S
=
R+\sum_{\ell=1}^r a_\ell b_\ell^\top.
\]
Therefore, using
\[
C_{xy}=\chi_1(x)^\top S\chi_2(y),
\]
together with \eqref{D.14} and \eqref{D.18}, we obtain
\[
C_{xy}
=
\sum_b\bigl(\rho_{1b}(x_b)+\rho_{2b}(y_b)\bigr)
+
\sum_{\ell=1}^r
\left(\sum_b\alpha_{\ell b}(x_b)\right)
\left(\sum_b\beta_{\ell b}(y_b)\right).
\]
Thus \eqref{D.8} always holds with \(q=r\).

All quantities constructed above are rational.  The matrices
\(D_p,T_p,L_p,S\) have dimensions polynomial in
\(k,B_1,B_2\).  Rational nullspace computation, extraction of image
bases, exact rank factorization, and 
\eqref{D.100} are all computable by exact rational Gaussian
elimination in polynomial time, with polynomially bounded encoding
length.  Hence a minimum-size rational representation with exactly
\(r\) product terms can be computed in polynomial time.
\end{proof}

\subsection{Relation to fixed-rank bilinear games}
\label{subsec:bilinear-relation}

The compact flow game is a bilinear game in the sense of Garg, Jiang, and Mehta~\cite{garg2011bilinear}, but their parameter is the ambient social-payoff rank $\rank(A+B)$ of a chosen representation. Theorem~\ref{thm:s-certificate} instead computes the minimum interaction rank on the feasible Blotto marginal sets, after removing terms that depend on only one player's action. If
\[
U_1(f,g)+U_2(f,g)
=
\rho_1(f)+\rho_2(g)+\sum_{\ell=1}^r \alpha_\ell(f)\beta_\ell(g),
\]
then subtracting $\rho_2(g)$ from player~1 and $\rho_1(f)$ from player~2 preserves every unilateral payoff difference and produces a strategically equivalent bilinear game of social rank $r$. Thus the intrinsic rank is exactly the rank relevant to equilibrium incentives, rather than to an arbitrary ambient flow encoding.

There is a second, algorithmically important distinction. The Vavasis-based FPTAS of Garg et al. controls the sum of unilateral gains at the scale
\[
\eta\,x_{\max}\lVert A+B\rVert_\infty y_{\max}.
\]
Obtaining an ordinary absolute $\eps$-Nash guarantee by direct substitution therefore introduces the representation-dependent factor $x_{\max}\lVert A+B\rVert_\infty y_{\max}$ into the required accuracy, and hence into the running time. In our intrinsic reduction, normalized feasible payoffs do not by themselves control the coefficient norm of a particular rational low-rank representation, so a black-box invocation would require an additional conditioning bound.

Our equilibrium-gap formulation avoids this issue. It is built directly from the normalized Blotto payoffs and exact best-response certificates, and Lemma~\ref{lem:gap-newmain} proves $\Phi_{\min}=0$ and $\Phi_{\max}-\Phi_{\min}\le 2$. Applying Vavasis to $\Phi$ with accuracy $\eps/2$ therefore yields the desired absolute $\eps$-Nash equilibrium with running time polynomial in the compact input size and $1/\eps$ for every fixed intrinsic rank, with no multiplicative dependence on the scale of a chosen bilinear factorization. This is the reason we use a dedicated gap program rather than invoke the generic bilinear FPTAS as a black box.

\subsection{Computing an approximate equilibrium from the rank decomposition}
\label{subsec:rank-equilibrium}

The preceding subsections already provide the compact flow representation, exact best-response oracle, and a minimum-rank decomposition of the social payoff. We now use that decomposition to formulate equilibrium computation as a fixed-rank quadratic optimization problem. This is the part of the argument not covered by the preceding rank analysis.

Fix the minimum-rank decomposition from Theorem~\ref{thm:s-certificate}. We extend each battlefield-additive feature linearly from pure allocations to flows. For $f\in F_1$ and $g\in F_2$, define
\begin{align*}
\rho_1(f)&:=\sum_{b=1}^k\sum_{s=0}^{B_1}(T_1f)_{b,s}\rho_{1b}(s),
&\alpha_\ell(f)&:=\sum_{b=1}^k\sum_{s=0}^{B_1}(T_1f)_{b,s}\alpha_{\ell b}(s),\\
\rho_2(g)&:=\sum_{b=1}^k\sum_{t=0}^{B_2}(T_2g)_{b,t}\rho_{2b}(t),
&\beta_\ell(g)&:=\sum_{b=1}^k\sum_{t=0}^{B_2}(T_2g)_{b,t}\beta_{\ell b}(t).
\end{align*}
Since $T_pf$ gives the battlefield-allocation marginals of a path distribution inducing $f$, linearity of expectation extends the pure-strategy identity of Theorem~\ref{thm:s-certificate} to every flow pair:
\begin{equation}
U_1(f,g)+U_2(f,g)
=
\rho_1(f)+\rho_2(g)+\sum_{\ell=1}^{r}\alpha_\ell(f)\beta_\ell(g).
\label{eq:rank-flow-social-decomp}
\end{equation}

\paragraph{Bellman certificates for best responses.}
For a player-$1$ edge $e=(u,v)$ crossing battlefield $b$ with label $s$, and a player-$2$ edge $e=(u,v)$ crossing battlefield $b$ with label $t$, define
\begin{equation}
r_e^1(g)=\sum_{\tau=0}^{B_2}(T_2g)_{b,\tau}u_b^1(s,\tau),
\qquad
r_e^2(f)=\sum_{\sigma=0}^{B_1}(T_1f)_{b,\sigma}u_b^2(t,\sigma).
\label{eq:rank-edge-rewards}
\end{equation}
Thus the reward of a source--sink path is exactly the corresponding pure allocation's expected payoff against the opponent flow.

For each player $p$, introduce a potential $p_v^p$ at every node of her allocation DAG, set
\begin{equation}
p_{(0,0)}^p=0,
\qquad
w_p=p_{(k,B_p)}^p,
\label{eq:rank-pot-endpoints}
\end{equation}
and impose
\[
p_v^1-p_u^1\ge r_e^1(g)
\quad\text{for every player-$1$ edge }e=(u,v),
\]
\begin{equation}
p_v^2-p_u^2\ge r_e^2(f)
\quad\text{for every player-$2$ edge }e=(u,v).
\label{eq:rank-pot-p2}
\end{equation}

\begin{lemma}[Potential representation of best responses]
\label{lem:rank-pot}
Every feasible potential vector satisfies
\[
w_1\ge \BR_1(g),
\qquad
w_2\ge \BR_2(f).
\]
For every $(f,g)\in F_1\times F_2$, longest-path potentials satisfy \eqref{eq:rank-pot-endpoints}--\eqref{eq:rank-pot-p2} with
\[
w_1=\BR_1(g),
\qquad
w_2=\BR_2(f).
\]
\end{lemma}

\begin{proof}
Summing the relevant edge inequalities along any source--sink path telescopes, so $w_p$ is at least the reward of that path. Maximizing over paths gives $w_p\ge\BR_p$. Conversely, assign to each node the maximum reward of a path from the source to that node. Because the allocation graph is acyclic, these longest-path values are well defined and satisfy every Bellman inequality; at the sink the potential equals the maximum source--sink path reward, which is exactly the best-response value by Proposition~\ref{prop:s-br}.
\end{proof}

Let
\[
H_p:=\sum_{b=1}^k\max_{s,t}|u_b^p(s,t)|.
\]
Every prefix of a source--sink path collects at most $H_p$ in absolute reward. Hence the exact longest-path potentials in Lemma~\ref{lem:rank-pot} satisfy $|p_v^p|\le H_p$ at every node. Imposing these bounds therefore makes the potential region bounded without excluding exact best-response certificates.

\paragraph{Maximum social payoff.}
Write
\[
c_b(s,t):=u_b^1(s,t)+u_b^2(t,s).
\]
The maximum pure-profile social payoff can be computed directly from the compact battlefield representation. Initialize
\[
M^+(0,0,0)=0
\]
and $M^+(0,h_1,h_2)=-\infty$ for every other pair $(h_1,h_2)$. For $b\ge1$, set
\[
M^+(b,h_1,h_2)=
\max_{\substack{0\le s\le h_1\\0\le t\le h_2}}
\left\{M^+(b-1,h_1-s,h_2-t)+c_b(s,t)\right\}.
\]
Then
\[
C_{\max}:=M^+(k,B_1,B_2)
=
\max_{(x,y)\in\X\times\Y}\{U_1(x,y)+U_2(x,y)\}.
\]
The recurrence uses $O(kB_1^2B_2^2)$ arithmetic operations. Since the game is normalized, $0\le C_{\max}\le2$.

\paragraph{A bounded equilibrium-gap program.}
Let $\mathcal P$ be the set of tuples
\[
(f,g,p^1,p^2,w_1,w_2)
\]
satisfying $f\in F_1$, $g\in F_2$, equations \eqref{eq:rank-pot-endpoints}--\eqref{eq:rank-pot-p2}, the bounds $|p_v^p|\le H_p$ for every player and node, and
\begin{equation}
w_1+w_2\le C_{\max}.
\label{eq:rank-social-cap}
\end{equation}
The flow polytopes are bounded, and all potential variables are explicitly bounded, so $\mathcal P$ is a bounded rational polytope. It is also nonempty. Indeed, take an exact Nash equilibrium of the finite allocation game, represent its mixed strategies by the induced flows, and equip them with exact longest-path potentials. At equilibrium,
\[
w_1+w_2
=U_1(f,g)+U_2(f,g)
\le C_{\max},
\]
where the last inequality follows because an expected social payoff cannot exceed the maximum social payoff of a pure profile.

Over $\mathcal P$, minimize
\begin{align}
\Phi(f,g,p^1,p^2,w_1,w_2)
&:=w_1+w_2-\rho_1(f)-\rho_2(g)
-\sum_{\ell=1}^{r}\alpha_\ell(f)\beta_\ell(g).
\label{eq:rank-objective-main}
\end{align}
By \eqref{eq:rank-flow-social-decomp}, this objective has the equivalent and useful form
\begin{equation}
\Phi
=w_1+w_2-U_1(f,g)-U_2(f,g).
\label{eq:rank-objective-gap}
\end{equation}

\begin{lemma}[Equilibrium-gap objective]
\label{lem:gap-newmain}
The minimum of $\Phi$ over $\mathcal P$ is zero, and every feasible point satisfies
\begin{equation}
0
\le
\BR_1(g)-U_1(f,g)+\BR_2(f)-U_2(f,g)
\le
\Phi
\le
2.
\label{eq:rank-gap-main}
\end{equation}
Moreover, $\nabla^2\Phi$ has rank at most $2r$ and at most $r$ negative eigenvalues.
\end{lemma}

\begin{proof}
Lemma~\ref{lem:rank-pot} gives $w_1\ge\BR_1(g)$ and $w_2\ge\BR_2(f)$. Using \eqref{eq:rank-objective-gap},
\begin{align*}
\Phi
&=(w_1-U_1(f,g))+(w_2-U_2(f,g))\\
&\ge
\BR_1(g)-U_1(f,g)+\BR_2(f)-U_2(f,g)
\ge0.
\end{align*}
At an exact Nash equilibrium, both unilateral gains vanish. Exact Bellman potentials therefore yield $w_p=\BR_p=U_p$ for both players and hence $\Phi=0$. As noted above, this tuple belongs to $\mathcal P$, so $\Phi_{\min}=0$.

For every feasible point, \eqref{eq:rank-social-cap} gives $w_1+w_2\le C_{\max}\le2$, whereas normalization gives $U_1(f,g)+U_2(f,g)\ge0$. Equation~\eqref{eq:rank-objective-gap} therefore gives $\Phi\le2$.

It remains to bound the quadratic rank. Let $A_\alpha$ and $B_\beta$ have as columns the flow coefficient vectors of $\alpha_\ell$ and $\beta_\ell$. The only quadratic part of $\Phi$ is
\[
-f^\top A_\alpha B_\beta^\top g.
\]
Up to the conventional factor of two, the symmetric Hessian on the variables $(f,g)$ is
\[
Q=
\begin{pmatrix}
0&-A_\alpha B_\beta^\top\\
-B_\beta A_\alpha^\top&0
\end{pmatrix}.
\]
The cross matrix $A_\alpha B_\beta^\top$ has rank at most $r$, so $Q$ has rank at most $2r$. Its nonzero eigenvalues occur in positive--negative pairs whose magnitudes are the singular values of $A_\alpha B_\beta^\top$. Hence $Q$, and therefore $\nabla^2\Phi$, has at most $r$ negative eigenvalues.
\end{proof}

We now invoke the fixed-rank quadratic-optimization approximation theorem of Vavasis, in the same range-relative form used in fixed-rank bilinear and bimatrix games~\cite{garg2011bilinear,kannan2010games}. Applied to our normalized gap objective, the theorem immediately converts the intrinsic-rank structure into an absolute approximation guarantee because $\Phi$ has range at most $2$; no representation-dependent bilinear scale enters the bound.

\begin{theorem}[Vavasis~\cite{Vavasis1992}]
\label{thm:rank-vavasis}
Fix $R$. Given a nonempty bounded rational polytope $P$ and a rational quadratic function $q$ whose Hessian has rank at most $R$, one can, for every rational $\delta>0$, compute $z^\delta\in P$ satisfying
\[
q(z^\delta)-q_{\min}
\le
\delta\bigl(q_{\max}-q_{\min}\bigr)
\]
in time polynomial in the binary input length and $1/\delta$, where the polynomial degree may depend on $R$.
\end{theorem}

The range-relative form above is exactly what is needed here: Lemma~\ref{lem:gap-newmain} supplies a dimension-independent range bound $\Phi_{\max}-\Phi_{\min}\le2$.

\begin{statementbox}
\begin{theorem}[Efficient computation for fixed intrinsic rank]
\label{thm:rank-newmain}
For every fixed $r\ge0$ and every $\eps\in(0,1)$, an $\eps$-Nash equilibrium of a normalized rank-$r$ Colonel Blotto game with $k$ battlefields and budgets $B_1,B_2$ can be computed in time
\[
\poly(k,B_1,B_2)\,
\poly_r\!\left(1/\eps\right).
\]
Consequently, for every fixed $r$, the algorithm is polynomial in the game parameters and $1/\eps$.
\end{theorem}
\end{statementbox}

\begin{proof}
The preceding subsections compute, in time polynomial in $k,B_1,B_2$ and the suppressed binary encoding lengths of the rational payoff data, the compact flow representation and a minimum-rank rational decomposition with exactly $r$ product terms. Equations~\eqref{eq:rank-edge-rewards}--\eqref{eq:rank-social-cap} construct the bounded rational polytope $\mathcal P$, and \eqref{eq:rank-objective-main} constructs its quadratic objective $\Phi$. By Lemma~\ref{lem:gap-newmain}, the Hessian rank is at most $2r$ and
\[
\Phi_{\min}=0,
\qquad
\Phi_{\max}-\Phi_{\min}\le2.
\]

Apply Theorem~\ref{thm:rank-vavasis} with $R=2r$ and $\delta=\eps/2$. For fixed $r$, the quadratic-optimization stage is polynomial in its binary input length and $1/\eps$, with a polynomial degree that may depend on $r$. It returns a feasible point $z^\delta$ satisfying
\[
\Phi(z^\delta)
=
\Phi(z^\delta)-\Phi_{\min}
\le
\frac{\eps}{2}\bigl(\Phi_{\max}-\Phi_{\min}\bigr)
\le\eps.
\]
Equation~\eqref{eq:rank-gap-main} implies that the sum of the two unilateral gains is at most $\eps$. Since each gain is nonnegative, both are at most $\eps$. Hence the returned flows form an $\eps$-Nash equilibrium of the flow game. Finally, the payoff-preservation statement preceding Proposition~\ref{prop:s-br} transfers the conclusion to the original allocation game, and Proposition~\ref{prop:s-path} decomposes the two rational flows into explicit polynomial-support mixed allocations.
\end{proof}

%% file: wta_frontier_start_and_positive.tex
\section{Multiplayer WTA Blotto: A Tie-Breaking Frontier}
\label{sec:multiplayer-wta}

The preceding sections established a complexity frontier for two-player Blotto, combining hardness for monotone piecewise-constant general-sum payoffs with tractability under fixed intrinsic rank. We now turn to the multiplayer setting and answer Question~3 for literal winner-takes-all (WTA) competition. As noted after Definition~\ref{def:wta-sharing}, equal splitting is the tie convention used in the multiplayer Blotto model of Boix-Adser\`a, Edelman, and Jayanti~\cite{boix2021multiplayer} for common battlefield values. We ask how equilibrium computation behaves in the discrete, player-specific-value setting, and compare this established convention with the full co-winner rule. The two extensions lead to sharply different computational regimes. If every co-winner receives her full battlefield value, an exact pure Nash equilibrium is computable in polynomial time. If the battlefield value is instead split equally among the co-winners, computing an inverse-polynomial approximate Nash equilibrium becomes \PPAD-hard even with only five players.

Our answer is summarized formally below.

\begin{statementbox}
\begin{theorem}[Tie-breaking frontier for multiplayer WTA Blotto]
\label{thm:wta-frontier}
Consider discrete multiplayer WTA Blotto with unrestricted access to all battlefields and nonnegative player-specific values.
\begin{enumerate}[label=(\alph*)]
\item Under the full co-winner rule $\tau^{\mathrm{full}}$, an exact pure Nash equilibrium can be computed in time polynomial in the binary input length, for arbitrary integer budgets. This holds under either convention for all-zero battlefields.
\item Under equal splitting $\tau^{\mathrm{split}}$, with the all-zero co-winner convention, computing an inverse-polynomial-accuracy Nash equilibrium is \PPAD-hard even with only five players. The constructed instances have unrestricted integer allocations, nonnegative player-specific rational values, polynomially bounded numerical budgets, and polynomial encoding length.
\end{enumerate}
\end{theorem}
\end{statementbox}

By Corollary~\ref{cor:wta-split-membership}, the approximation problem in part~(b) belongs to \PPAD. Thus part~(b) is, in fact, \PPAD-complete.

The contrast is structural. Under full rewards, joining a tie never reduces the incumbent winner's payoff. This allows a maximum-budget player to pin down the bid required to obtain each battlefield. Once these thresholds are fixed, the remaining players no longer interact strategically: each of them simply chooses which thresholds to meet subject to her budget. Under equal splitting, this reduction is no longer possible, because an additional co-winner reduces the payoff of every player already tied at the top. The five-player hardness construction developed below exploits this dilution effect together with two stabilizers and an inspector to enforce a discrete encoding of a bimatrix game.

\subsection{Full co-winner rewards admit an exact pure equilibrium}
\label{subsec:full-reward-easy}

We first prove the positive side of the frontier. Under the full co-winner rule, every player whose bid is maximal and positive on a battlefield receives her entire player-specific value for that battlefield, independently of how many other players submit the same maximal bid. We treat separately the two possible conventions for battlefields on which every player bids zero.

\begin{theorem}[Exact equilibrium under full co-winner rewards]
\label{thm:full-reward-easy}
Under $\tau^{\mathrm{full}}$, an exact pure Nash equilibrium can be computed using $\mathcal{O}(Kk\log(k+1))$ comparisons and $\mathcal{O}(Kk)$ additional arithmetic operations. The bit complexity is polynomial in the complete binary input length, with no restriction on the magnitudes of the budgets.
\end{theorem}

The proof relies on a simple threshold-pinning argument. When ties give every co-winner her full value, matching another player's maximal bid is sufficient to obtain a battlefield and does not harm the incumbent. We exploit this property by letting a maximum-budget player place a positive bid on every battlefield. Her allocation then determines a vector of winning thresholds. Every other player's best-response problem becomes a selection problem in which battlefield $f$ can be obtained at cost equal to its threshold. With the particular threshold vector used below, all but one battlefield have cost one, so this selection problem can be solved by sorting.

\begin{proof}
Let
\[
B^*:=\max_{i\in[K]} B_i
\]
denote the largest budget.

\medskip
\noindent
\emph{All-zero battlefields reward every player.}
We fix an arbitrary battlefield, say battlefield~$1$, and let every player allocate her entire budget there:
\[
a_i=(B_i,0,\ldots,0).
\]
Consider first a player $i$ with $B_i=B^*$. She wins or co-wins battlefield~$1$. Every other battlefield receives only zero bids, so under the present convention she also receives her full value on all of them. Hence her payoff equals
\[
\sum_{f=1}^k v_{if},
\]
which is the largest payoff she could obtain under any allocation.

Now consider a player $i$ with $B_i<B^*$. She cannot win or co-win battlefield~$1$ under any deviation, because doing so would require a bid of at least $B^*$ while her entire budget is smaller than $B^*$. On every other battlefield she already receives her full value from the all-zero convention. Consequently no deviation can increase her payoff. The displayed profile is therefore a pure Nash equilibrium.

If $B^*=0$, every player has zero budget and there is a unique feasible allocation, so the claim is immediate.

\medskip
\noindent
\emph{All-zero battlefields reward nobody.}
We first consider the case $B^*<k$. For each player $i$, let $S_i$ consist of her $B_i$ highest-valued battlefields, breaking ties arbitrarily, and set
\[
a_{if}=
\begin{cases}
1, & f\in S_i,\\
0, & f\notin S_i.
\end{cases}
\]
Since $B_i\le B^*<k$, this allocation uses exactly $B_i$ units and is feasible.

Every positive bid in the resulting profile equals one. Hence, on every battlefield $f\in S_i$, player $i$ is a winner or co-winner and receives the entire value $v_{if}$. Her payoff is therefore
\[
\sum_{f\in S_i}v_{if},
\]
the sum of her $B_i$ largest battlefield values.

To see that this is a best response, observe that every positive bid is an integer of at least one. Thus any feasible allocation of player $i$ can place positive bids on at most $B_i$ battlefields. Since a battlefield contributes at most its full value $v_{if}$ to player $i$'s payoff, no deviation can give her more than the sum of her $B_i$ largest values. The constructed allocation attains this upper bound. Hence every player is best responding, and the profile is a pure Nash equilibrium.

It remains to treat the case $B^*\ge k$. We fix a player $\ell$ with budget $B^*$. This player will place a positive bid on every battlefield and thereby determine the bid that another player must match in order to obtain it.

We reserve one unit for each of battlefields $2,\ldots,k$ and place the remaining budget on battlefield~$1$. Define
\[
q:=B^*-(k-1)=B^*-k+1\ge1
\]
and
\[
p:=(q,1,\ldots,1).
\]
The vector $p$ uses exactly
\[
q+(k-1)=B^*
\]
units and is therefore feasible for every maximum-budget player.

The next lemma formalizes the role of $p$. Once one maximum-budget player uses this allocation, the coordinates of $p$ become independent winning thresholds for all remaining players.

\begin{lemma}[A maximum-budget player fixes all winning thresholds]
\label{lem:full-reward-thresholds}
Suppose player $\ell$ uses $p$ and every other player $i\ne\ell$ uses an allocation satisfying $a_{if}\le p_f$ for every battlefield $f$. Then player $\ell$ receives $\sum_f v_{\ell f}$ and cannot improve. Moreover, for every other player $i$, the maximum payoff attainable against the resulting opponent profile is
\begin{equation}
\max\left\{
\sum_{f\in S}v_{if}:
S\subseteq[k],\ 
\sum_{f\in S}p_f\le B_i
\right\}.
\label{eq:full-reward-knapsack}
\end{equation}
Every maximizing set in \eqref{eq:full-reward-knapsack} can be realized by a full-budget allocation that is componentwise at most $p$.
\end{lemma}

The lemma isolates the essential effect of full co-winner rewards. Player $\ell$ makes $p_f$ the minimum bid needed to obtain battlefield $f$, while matching $p_f$ is already sufficient. Thus the other players face a deterministic collection of battlefield ``prices.''

\begin{proof}
Every coordinate of $p$ is positive. Since every other player bids at most $p_f$ on battlefield $f$, player $\ell$ wins or co-wins every battlefield. Under the full co-winner rule, ties do not reduce her reward, so she receives
\[
\sum_{f=1}^k v_{\ell f}.
\]
This is the largest payoff she can possibly obtain, and therefore player $\ell$ is best responding.

We next fix another player $i\ne\ell$. Because player $\ell$ bids exactly $p_f$ on battlefield $f$, player $i$ obtains positive value from $f$ only if she bids at least $p_f$. Conversely, bidding exactly $p_f$ is sufficient to win or co-win the battlefield and receive the full value $v_{if}$.

Suppose a deviation of player $i$ obtains a collection $S$ of battlefields. The deviation must spend at least $p_f$ on every $f\in S$, and hence
\[
\sum_{f\in S}p_f\le B_i.
\]
Its payoff is at most $\sum_{f\in S}v_{if}$. Maximizing over all possible collections $S$ gives the upper bound in \eqref{eq:full-reward-knapsack}.

It remains to show that this upper bound is attainable without violating the coordinatewise condition $a_i\le p$. Let $S$ be a maximizing set in \eqref{eq:full-reward-knapsack}. We first assign
\[
a_{if}=p_f
\qquad\text{for every } f\in S.
\]
The budget that remains unallocated is
\[
B_i-\sum_{f\in S}p_f.
\]
On the other hand, the total unused coordinate capacity below $p$ is
\[
\sum_{f\notin S}p_f
=
B^*-\sum_{f\in S}p_f,
\]
because $\sum_f p_f=B^*$. Since $B_i\le B^*$, we have
\[
B_i-\sum_{f\in S}p_f
\le
B^*-\sum_{f\in S}p_f.
\]
Thus all remaining budget can be distributed among the coordinates outside $S$ without making any bid exceed the corresponding threshold $p_f$.

This filling operation may cause player $i$ to meet the threshold on some additional battlefields. Let $S'\supseteq S$ denote the final set of battlefields on which $a_{if}=p_f$. Since the final allocation uses budget $B_i$,
\[
\sum_{f\in S'}p_f\le B_i,
\]
so $S'$ is itself feasible for \eqref{eq:full-reward-knapsack}. Because $S$ is optimal and all values are nonnegative,
\[
\sum_{f\in S}v_{if}
\le
\sum_{f\in S'}v_{if}
\le
\max_{\substack{T\subseteq[k]\\ \sum_{f\in T}p_f\le B_i}}
\sum_{f\in T}v_{if}
=
\sum_{f\in S}v_{if}.
\]
Hence equality holds throughout. The completed allocation therefore attains the optimum in \eqref{eq:full-reward-knapsack}, uses the entire budget, and remains componentwise bounded by $p$.
\end{proof}

It remains to solve \eqref{eq:full-reward-knapsack}. Although it is written as a $0$--$1$ knapsack problem, its cost structure is particularly simple. Battlefield~$1$ has threshold cost $q$, whereas every battlefield in $\{2,\ldots,k\}$ has threshold cost one. Thus there are only two possibilities to consider: an optimal set either skips battlefield~$1$ or contains it.

For each player $i$, we sort battlefields $2,\ldots,k$ in nonincreasing order of $v_{if}$. Let $H_i(t)$ denote the sum of the $t$ largest values among these $k-1$ battlefields, with $H_i(0)=0$.

If $B_i=B^*$, we assign player $i$ the allocation $p$. She then co-wins every battlefield and receives her maximum possible payoff. We therefore consider a player with $B_i<B^*$.

Suppose first that battlefield~$1$ is skipped. Since every remaining battlefield has cost one, the player should select as many of the most valuable fields among $2,\ldots,k$ as her budget permits. The resulting value is
\begin{equation}
V_i^{\mathrm{skip}}
=
H_i\bigl(\min\{B_i,k-1\}\bigr).
\label{eq:full-skip}
\end{equation}
This value is attained by placing one unit on the corresponding highest-valued fields. If $B_i\le k-1$, these bids use the entire budget. If $B_i>k-1$, all $k-1$ unit-cost fields are selected and the remaining
\[
B_i-(k-1)
\]
units are placed on battlefield~$1$. Since
\[
B_i<B^*=q+k-1,
\]
the leftover satisfies
\[
B_i-(k-1)<q.
\]
Thus the additional bid on battlefield~$1$ remains strictly below its threshold and does not alter the set of battlefields obtained.

We next consider solutions that contain battlefield~$1$. Such a solution is possible only when $B_i\ge q$. Meeting the threshold on battlefield~$1$ consumes $q$ units, leaving budget $B_i-q$. Every remaining battlefield costs one, so the optimal completion consists of the $B_i-q$ highest-valued battlefields among $2,\ldots,k$. Hence
\begin{equation}
V_i^{\mathrm{take}}
=
v_{i1}+H_i(B_i-q).
\label{eq:full-take}
\end{equation}
The assumption $B_i<B^*$ ensures that
\[
B_i-q<k-1,
\]
so the argument of $H_i$ is always valid. The corresponding allocation places $q$ units on battlefield~$1$ and one unit on each of the selected $B_i-q$ unit-cost battlefields, using the entire budget exactly.

These two cases exhaust every feasible winning set in \eqref{eq:full-reward-knapsack}. Indeed, a set either excludes battlefield~$1$, in which case all selected items have unit cost, or includes battlefield~$1$, in which case the remaining budget is spent on unit-cost items. Therefore the optimal value for player $i$ is
\[
\begin{cases}
V_i^{\mathrm{skip}},
& B_i<q,\\[2mm]
\max\{V_i^{\mathrm{skip}},V_i^{\mathrm{take}}\},
& B_i\ge q.
\end{cases}
\]
We assign player $i$ the allocation corresponding to a maximizing case, breaking ties arbitrarily. In either case the resulting allocation is componentwise bounded by $p$ and attains the optimum in \eqref{eq:full-reward-knapsack}. Lemma~\ref{lem:full-reward-thresholds} therefore implies that player $i$ is best responding.

Applying this construction independently to every player other than $\ell$ yields a profile in which all allocations are componentwise bounded by $p$. The same lemma then shows simultaneously that $\ell$ is best responding and that every other player is best responding. The resulting profile is therefore an exact pure Nash equilibrium.

It remains to verify the running time. For each player, the only nontrivial operation is sorting the $k-1$ values associated with battlefields $2,\ldots,k$. This requires
\[
\mathcal{O}(k\log(k+1))
\]
comparisons per player. Once the values are sorted, all prefix sums $H_i(t)$ are computed in $\mathcal{O}(k)$ arithmetic operations, after which \eqref{eq:full-skip} and, when applicable, \eqref{eq:full-take} are evaluated in constant time. Over all $K$ players, this gives
\[
\mathcal{O}(Kk\log(k+1))
\]
comparisons and $\mathcal{O}(Kk)$ additional arithmetic operations.

Finally, the construction never iterates once per unit of budget. The quantities $B^*$, $q$, and the coordinates of the output allocations are manipulated directly in binary. Their bit lengths are polynomial in the bit lengths of the input budgets, and the same holds for all arithmetic involving the input values and their prefix sums. Hence the overall bit complexity is polynomial in the complete binary input length, even when the numerical budgets are exponentially large.
\end{proof}

The proof can be viewed as a particularly simple instance of a threshold- or price-pinning argument. A maximum-budget player establishes the effective cost $p_f$ of obtaining each battlefield. Full co-winner rewards are essential here: another player may match one of these thresholds without reducing the payoff of the player who created it. Once the thresholds are fixed, strategic interaction disappears from the remaining best-response problems.

In the present construction, this reduction is especially strong because the threshold vector has only two distinct coordinates. Battlefield~$1$ has cost $q$, while every other battlefield has cost one. The induced $0$--$1$ knapsack therefore reduces to comparing the best solution that skips the exceptional battlefield with the best solution that includes it, both of which are obtained from a single sorting of the unit-cost battlefields.

Theorem~\ref{thm:full-reward-easy} proves part~(a) of Theorem~\ref{thm:wta-frontier}. The proof also identifies exactly what fails under equal splitting. If a second player matches the threshold set by a maximum-budget player, the incumbent no longer retains her entire value; instead, her payoff is divided among the co-winners. Thus the threshold setter's incentives depend on which other players match her bids, and the independent-selection interpretation breaks down. We next show that this seemingly small change in the treatment of ties is sufficient to reverse the computational picture.

%% file: wta_equal_splitting_hardness_only.tex
\subsection{Equal splitting is \PPAD-hard with five players}
\label{subsec:equal-splitting-hardness}

% Local notation for the five-player reduction.
\newcommand{\wtaFiveR}{\mathcal R}
\newcommand{\wtaFiveC}{\mathcal C}
\newcommand{\wtaFiveReg}{\operatorname{Reg}}
\newcommand{\wtaFiveOut}{\mathrm{out}}
\newcommand{\wtaFiveBalR}[1]{\mathrm{bal}^{R}_{#1}}
\newcommand{\wtaFiveBalC}[1]{\mathrm{bal}^{C}_{#1}}
\newcommand{\wtaFivePartR}{\mathrm{part}_{R}}
\newcommand{\wtaFivePartC}{\mathrm{part}_{C}}
\newcommand{\wtaFiveCodeR}[1]{\mathrm{code}_{R}(#1)}
\newcommand{\wtaFiveCodeC}[1]{\mathrm{code}_{C}(#1)}
\newcommand{\wtaFiveLead}[1]{\par\medskip\noindent\textbf{#1}\par\smallskip\noindent}
\definecolor{wtafivepipeone}{RGB}{232,234,249}
\definecolor{wtafivepipetwo}{RGB}{237,231,248}
\definecolor{wtafivepipethree}{RGB}{238,237,246}
\definecolor{wtafivepipefour}{RGB}{231,239,249}

We next prove the hardness side of Theorem~\ref{thm:wta-frontier}.  Throughout this subsection we use the
multiplayer Blotto conventions from Section~\ref{sec:preliminaries} and the equal-splitting rule
$\tau^{\mathrm{split}}$ from Definition~\ref{def:wta-sharing}.  For the reduction we use the all-zero
co-winner convention, so an all-zero battlefield is split equally among all five players.

By Theorem~\ref{thm:s-general-source}, it is \PPAD-hard to compute a $1/q$-Nash equilibrium of a
normalized $q\times q$ bimatrix game.  Let
\[
(A,B)\in[0,1]^{q\times q}\times[0,1]^{q\times q}
\]
be such a source instance, with row and column action sets $[q]_0$.  We restrict to the hard family with
$q\ge2$, which is without loss for the asymptotic hardness statement.

The reduction uses five Blotto players
\[
R,\ C,\ S,\ T,\ J.
\]
Players $R$ and $C$, the two \emph{primaries}, carry the row and column strategies of the source game.
Players $S$ and $T$, the \emph{stabilizers}, create the local competition that makes incomplete encodings
unprofitable.  Player $J$, the \emph{inspector}, regulates how often the primaries use their encodings.
Their precise actions, budgets, and battlefield values are introduced with the construction below.

Let $n_{\rm inst}$ denote the binary encoding length of the constructed Blotto instance.

\begin{theorem}[Five-player hardness under equal splitting]\label{thm:wta-five-player-hardness}
Computing an inverse-polynomial-accuracy Nash equilibrium of discrete equal-splitting WTA Blotto with
unrestricted integer allocations, nonnegative player-specific rational values, polynomially bounded numerical
budgets, and the all-zero co-winner convention is \PPAD-hard, even with only five players.
\end{theorem}

\subsubsection{Roadmap and the logic of the reduction}\label{subsec:wta5-roadmap}

\paragraph{Encoding structure.}
For every $(a,b)\in[q]_0^2$ there is a \emph{pair field} $f_{ab}$.
For every row label $a$ there is a row \emph{balancing field} $\wtaFiveBalR{a}$, and for every column label $b$
there is a column balancing field $\wtaFiveBalC{b}$.  The row-side and column-side encoding families are
\[
\mathcal F_R=\{f_{ab}:a,b\in[q]_0\}\cup\{\wtaFiveBalR{a}:a\in[q]_0\},
\qquad
\mathcal F_C=\{f_{ab}:a,b\in[q]_0\}\cup\{\wtaFiveBalC{b}:b\in[q]_0\},
\]
and $\mathcal F_{\rm enc}=\mathcal F_R\cup\mathcal F_C$ is the set of all \emph{encoding fields}.
The pair fields are shared by the two sides; the balancing fields are side-specific.

Each encoding field $f$ will be assigned a positive integer \emph{completion bid} $c_f$.  The support of
row action $a$ and column action $b$ is, respectively,
\[
\wtaFiveR_a=\{f_{ab}:b\in[q]_0\}\cup\{\wtaFiveBalR{a}\},
\qquad
\wtaFiveC_b=\{f_{ab}:a\in[q]_0\}\cup\{\wtaFiveBalC{b}\}.
\]

The completion bids are chosen so that
\[
\sum_{f\in\wtaFiveR_a}c_f=\sum_{f\in\wtaFiveC_b}c_f=D
\]
for every $a,b$; $D$ is therefore the common \emph{encoding cost} of one source action.

An \emph{endpoint bid} on $f$ is a bid in $\{0,c_f\}$.  An \emph{endpoint pattern} is an encoding allocation
in which every encoding-field bid is an endpoint bid.  Its \emph{completed-field set} is the set of encoding
fields on which the primary bids $c_f$.  A \emph{complete row codeword} is an endpoint pattern whose
completed-field set is $\wtaFiveR_a$ for some $a$; a complete column codeword is defined analogously.  A
\emph{partial row codeword for $a$} is a nonzero encoding allocation supported inside $\wtaFiveR_a$, with every
bid on a field $f\in\wtaFiveR_a$ lying between $0$ and $c_f$, but which is not the complete row codeword; partial
column codewords are defined symmetrically.  Thus a partial codeword may omit a field or use an intermediate
bid strictly between $0$ and $c_f$.  More generally, whenever a primary places troops on encoding fields
without reproducing one complete codeword at the prescribed completion bids, we call that nonzero encoding
allocation \emph{malformed}.  Thus partial codewords, mixtures of fields from different codes, and overbids are all
malformed.  The zero encoding vector is not malformed when it occurs as part of the neutral action.  To avoid a
terminological gap later, we call a \emph{noncanonical nonneutral allocation} any nonneutral primary allocation
that is not one of the complete code allocations.  This includes nonneutral allocations whose encoding part is
zero.  The canonicalization argument rules out all noncanonical nonneutral allocations, not only those whose
encoding part is nonzero.

The common primary budget will be $B_*=D+1$.  Primary $i\in\{R,C\}$ has a \emph{participation field}
$\mathrm{part}_i$.  Its \emph{neutral action} $N_i$ puts all $B_*$ troops on $\mathrm{part}_i$.  Its canonical active allocation
$\wtaFiveCodeR{a}$ (respectively $\wtaFiveCodeC{b}$) puts one troop on its participation field and the remaining $D$ troops
at the completion bids of $\wtaFiveR_a$ (respectively $\wtaFiveC_b$).  For a mixed profile $\mu$, let $\mu_i$ denote
player $i$'s mixed strategy.  We say that primary $i$ is \emph{active}, or \emph{activated}, whenever it does not
play $N_i$, and write
\[
x_i=1-\mu_i(N_i)
\]
for its activation probability.  The symbol $h$ denotes the small target activation probability; its
explicit inverse-polynomial value is fixed later in \eqref{wta5:eq:scales}.

A \emph{buffer field} is an auxiliary battlefield whose purpose is only to park unused troops at a controlled
marginal payoff.  We denote the private buffer banks of the stabilizers by $\mathcal B_S$ and $\mathcal B_T$,
and the shared low-value primary-buffer bank of $R$ and $C$ by $\mathcal B_P$.  The inspector also has its
own private buffer bank.  The symbols $I_R$ and $I_C$ denote the inspector's two
\emph{full-inspection actions}, which place its whole budget on $\wtaFivePartR$ and $\wtaFivePartC$, respectively;
$I_{\wtaFiveOut}$ denotes its canonical outside action in its private buffer bank.  The exact buffer values and
the precise outside allocation are specified in Section~\ref{subsec:wta5-instance}.

Finally, \emph{canonicalization} means replacing negligible malformed probability mass so that every primary
action in support is either $N_i$ or one of its complete codewords.

The proof proceeds through four logical layers, each supplying the structural fact required by the next.
Figure~\ref{wta5:fig:pipeline} summarizes this progression.

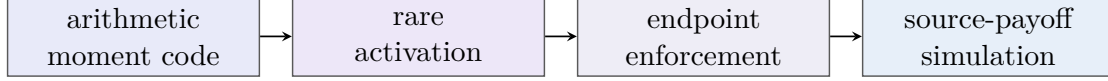
\begin{figure}[ht]
\centering
\begin{tikzpicture}[
  >=stealth,
  node distance=0.42cm,
  stage/.style={draw=black!65, line width=0.55pt, minimum height=1.05cm,
                text width=3.05cm, align=center, inner sep=4pt, font=\small}
]
\node[stage,fill=wtafivepipeone] (code) {arithmetic\\moment code};
\node[stage,fill=wtafivepipetwo,right=of code] (rare) {rare\\activation};
\node[stage,fill=wtafivepipethree,right=of rare] (canon) {endpoint\\enforcement};
\node[stage,fill=wtafivepipefour,right=of canon] (sim) {source-payoff\\simulation};
\draw[->,line width=0.65pt] (code) -- (rare);
\draw[->,line width=0.65pt] (rare) -- (canon);
\draw[->,line width=0.65pt] (canon) -- (sim);
\end{tikzpicture}
\caption{Logical pipeline of the reduction.}
\label{wta5:fig:pipeline}
\end{figure}

\wtaFiveLead{Layer I: the arithmetic code embeds source actions and makes them uniquely decodable.}
Each source action is assigned one canonical Blotto allocation.  Conversely, once an arbitrary Blotto allocation
has been rounded to endpoint bids, the arithmetic fingerprint identifies at most one source action.  This converse
decoding property is formalized in Section~\ref{subsec:wta5-construction}, in particular Lemma~\ref{wta5:lem:code}, and is the
part of the arithmetic gadget used in the soundness argument.

The completion costs are chosen so that every complete codeword has total encoding cost $D$, whereas every other endpoint
pattern has total completion cost separated from $D$ by a large gap.  This is the decoding property:
after arbitrary bids have been rounded to endpoints, the budget identifies one and only one source action.
The second structural property is
\[
\wtaFiveR_a\cap\wtaFiveC_b=\{f_{ab}\}.
\]
Thus a row codeword and a column codeword intersect in exactly one pair field, namely $f_{ab}$.  This unique
common field is what later carries the source payoff entry $(A_{ab},B_{ab})$ in Lemma~\ref{wta5:lem:simulation}.

\wtaFiveLead{The code uses a deterministic moment fingerprint.}
The arithmetic code is best understood as a tiny ``method of moments'' argument.  It is not statistical:
there is no sampling or estimation.  We simply encode three algebraic moments of the row labels into
separate base-$Q$ digits, where the base will later be fixed as $Q=q^4+1$.  If a candidate subset of pair fields has row labels $a_1,\ldots,a_t$, define
\[
M_0=t,\qquad M_1=\sum_j a_j,\qquad M_2=\sum_j a_j^2.
\]
A complete row-$a$ code has the fingerprint
\[
(M_0,M_1,M_2)=(q,qa,qa^2).
\]
If another subset has the same three numbers, then
\[
\sum_j(a_j-a)^2=M_2-2aM_1+a^2M_0=0,
\]
so every selected row label equals $a$.  Thus the zeroth moment fixes the number of selected pair fields,
the first moment fixes their mean, and the second moment forces zero variance.  Lemma~\ref{wta5:lem:code}
implements exactly this fingerprint using base-$Q$ digits.

\wtaFiveLead{Layer II: the inspector makes arbitrary nonneutral behavior rare.}
The common primary budget is
\[
B_*=D+1.
\]
Primary $i$ has a participation field $\mathrm{part}_i$.  Its neutral action puts all $B_*$ troops on $\mathrm{part}_i$; a
represented source action puts one troop on $\mathrm{part}_i$ and the remaining $D$ troops on one code.  Before we
know anything about the shape of a nonneutral allocation, the inspector compares full inspection of $\mathrm{part}_i$
with an outside reserve action.  The first part of the canonicalization argument in
Section~\ref{subsec:wta5-canonical} (Lemma~\ref{wta5:lem:canonical}) then implies
\[
x_R,x_C=O(h).
\]
This preliminary estimate is deliberately one-sided: it rules out too much activation, which is all that is needed
before analyzing the stabilizers.  Consequently even completely malformed primary allocations disturb a local
stabilizer contest only on an $O(h)$-probability event.

\wtaFiveLead{Layer III: two stabilizers convert a linear bid scale into a quadratic completion incentive.}
The construction assigns the stabilizer buffer fields values such that placing one troop on an empty private
buffer yields one unit of payoff relative to the all-zero baseline.  Consequently, moving $j$ troops from these
buffers to an encoding field has opportunity cost exactly $j$.  The stabilizers' battlefield value on the encoding
field is set slightly below $c_f$, and small stabilizer regret forces consecutive bid probabilities to have almost
constant pair-sums.  Consequently, one
stabilizer lies below level $j$ with probability roughly proportional to $j/c_f$.

The second stabilizer creates the required curvature: a primary wins the field only by beating both independent
stabilizer bids.  A partial bid of $j$ troops therefore succeeds with probability roughly
\[
\left(\frac{j}{c_f}\right)^2.
\]
Spending half the completion cost yields only about one quarter of the full winning probability, making partial
completion inefficient.  This local calculation is carried out in the proof of Lemma~\ref{wta5:lem:canonical} and is summarized by
\[
c_fg_f(j)\lesssim \frac{j^2}{c_f},\qquad
g_f(j)\lesssim \frac{j^2}{c_f^2},\qquad
c_f[1-g_f(c_f-a)]\gtrsim a,
\]
where $g_f(j)$ is the expected share won by a primary bid of $j$ troops.

These inequalities are then used only once, in the global rounding argument of
Section~\ref{subsec:wta5-canonical} (Lemma~\ref{wta5:lem:canonical}).  Every encoding bid is rounded to $0$ or $c_f$.
If the total distance from these endpoints is large, the allocation loses a linear amount of payoff.  If the
distance is small, the rounded endpoint set has cost close to $D$, so Lemma~\ref{wta5:lem:code} forces it to be
one unique row or column.  Troop conservation then repairs the remaining small defects.

A tiny shared primary buffer is retained solely as a repair device.  It gives a primary somewhere useful
to move a troop that currently sits on a field of zero value.  Its payoff scale is much smaller than the
code-enforcement scale, so it never participates in the simulation itself.

\wtaFiveLead{Layer IV: one common row--column field reproduces one source payoff entry.}
After canonicalization and calibration (Lemmas~\ref{wta5:lem:canonical} and~\ref{wta5:lem:simulation}), the source strategies
are the primary distributions conditional on activation:
\[
\sigma_a=\frac{\mu_R(\wtaFiveCodeR{a})}{x_R},
\qquad
\tau_b=\frac{\mu_C(\wtaFiveCodeC{b})}{x_C}.
\]

Thus $\sigma=(\sigma_a)_{a\in[q]_0}$ and $\tau=(\tau_b)_{b\in[q]_0}$ are the decoded row and column mixed strategies.
The earlier inspector argument only rules out too much activation; by itself it does not prevent $x_i$
from being much smaller than $h$.  After canonicalization, however, every active primary allocation is a complete
code, and such a code has a guaranteed payoff advantage over neutral play unless the inspector assigns
nonnegligible probability to the corresponding full-inspection action $I_i$.  Thus $I_i$ must receive
nonnegligible probability, while the preliminary rarity bound leaves substantial probability on $I_{\wtaFiveOut}$.
Their payoffs to the inspector are $(1+x_i)/2$ and $(1+h)/2$, respectively, so approximate equilibrium
forces them to be nearly indifferent and yields $x_i=h\pm o(h)$.  This calibration is proved in
Section~\ref{subsec:wta5-simulation}, Lemma~\ref{wta5:lem:simulation}.

If $R$ uses $\wtaFiveCodeR{a}$ and $C$ uses $\wtaFiveCodeC{b}$, they meet only on $f_{ab}$, where both bid
$c_{ab}$ and split the field.  Pair-field values are chosen so that this single tie contributes the
source-dependent term $A_{ab}$ to the row player and $B_{ab}$ to the column player.  The balancing fields
remove the deterministic row/column terms introduced by the arithmetic code.  Lemma~\ref{wta5:lem:simulation}
establishes the resulting payoff identity
\[
U_R(\wtaFiveCodeR{a})=C_R+x_C\eta(A\tau)_a+\text{small error},
\]
where $C_R$ is independent of the row label $a$, and $\eta$ is the source-payoff scale defined later.
There is a symmetric expression for $C$.  Section~\ref{subsec:wta5-decoding} then turns a source best response into
a Blotto deviation and transfers the small Blotto regret to the decoded source profile.

\begin{center}
\renewcommand{\arraystretch}{1.08}
\begin{tabular}{@{}p{0.25\textwidth}p{0.65\textwidth}@{}}
\toprule
Object & Role in the proof \\
\midrule
$\wtaFivePartR,\wtaFivePartC$ & \emph{participation fields}: distinguish neutral from active play before the active allocation has been canonicalized \\
$J$ & \emph{inspector}: makes activation rare and later calibrates its probability to $h$ \\
$S,T$ & \emph{stabilizers}: create the quadratic penalty for partial completion \\
$\mathcal B_S,\mathcal B_T$ & \emph{stabilizer buffer fields}: turn a troop moved to an encoding field into a linear cost of one per troop \\
$\mathcal B_P$ & \emph{primary buffer fields}: remove wrong-field troops without introducing another enforcement argument \\
$\wtaFiveBalR{a},\wtaFiveBalC{b}$ & \emph{balancing fields}: equalize the total cost of all source-action codes \\
$f_{ab}$ & \emph{pair field}: unique common field carrying source payoff entry $(A_{ab},B_{ab})$ \\
\bottomrule
\end{tabular}
\end{center}

\subsubsection{Construction and the moment code}\label{subsec:wta5-construction}

We now instantiate the arithmetic code described above.  To make the base-$Q$ interpretation completely literal, take
\[
Q=q^4+1.
\]
For a row label $a$ and column label $b$, define the row and column tags
\[
r_a=Qa+Q^2a^2,
\qquad
s_b=Q^3b+Q^4b^2.
\]

The unscaled completion cost of pair field $f_{ab}$ is
\begin{equation}\label{wta5:eq:oldpair}
\widehat c_{ab}=2(1+r_a+s_b).
\end{equation}
Choose the common unscaled code cost
\[
\widehat D=2\max\left\{\max_a\sum_b\widehat c_{ab},\;\max_b\sum_a\widehat c_{ab}\right\}+1.
\]
The balancing-field costs are exactly the amounts missing from the corresponding row or column:
\[
\widehat c_{\wtaFiveBalR{a}}=\widehat D-\sum_b\widehat c_{ab},
\qquad
\widehat c_{\wtaFiveBalC{b}}=\widehat D-\sum_a\widehat c_{ab}.
\]

Thus every row and column code has unscaled cost exactly $\widehat D$.

There are $m=q^2+2q$ encoding fields.  Put
\[
H=256(q+1)q,
\qquad
L=2^{20}mH,
\]
and scale the completion bids by
\begin{equation}\label{wta5:eq:completion}
c_f=L\widehat c_f+1,
\qquad
D=L\widehat D+q+1.
\end{equation}
The added $+1$ on every completed field is useful because a code contains exactly $q+1$ fields, while
the large factor $L$ turns exact uniqueness in the unscaled code into a robust additive gap.

\begin{lemma}[Moment-code geometry]\label{wta5:lem:code}
Every row code $\wtaFiveR_a$ and every column code $\wtaFiveC_b$ has completion cost $D$.  If
$E\subseteq\mathcal F_R$ is not one of the row codes $\wtaFiveR_a$, or if
$E\subseteq\mathcal F_C$ is not one of the column codes $\wtaFiveC_b$, then
\begin{equation}\label{wta5:eq:codegap}
\left|\sum_{f\in E}c_f-D\right|\ge L-m.
\end{equation}
Moreover
\begin{equation}\label{wta5:eq:additive}
\wtaFiveR_a\cap\wtaFiveC_b=\{f_{ab}\},
\qquad
\frac{c_{ab}}2=\theta+u_a+v_b,
\end{equation}
where
\[
\theta=L+\frac12,
\qquad
u_a=Lr_a,
\qquad
v_b=Ls_b.
\]
\end{lemma}

\begin{proof}
The equal-cost and intersection statements are immediate from the definitions.  We focus on uniqueness,
because this is where the moment fingerprint appears.

Consider a subset of row-side encoding fields having unscaled cost $\widehat D$.  Every pair-field cost is
even, while $\widehat D$ and every row-balancer cost are odd.  Also every balancer exceeds
$\widehat D/2$.  Hence the subset contains exactly one row balancer, say $\wtaFiveBalR{a}$, together with some
pair fields $E$.  Those pair fields must have total unscaled cost $\sum_b\widehat c_{ab}$.

Divide this equality by two and look at its base-$Q$ digits.  Because $Q=q^4+1$ and at most $q^2$ pair
fields can be selected, none of the quantities
\[
|E|,\qquad \sum_E a',\qquad \sum_E(a')^2,\qquad \sum_E b',\qquad \sum_E(b')^2
\]
reaches $Q$.  Thus there is no carry between the five relevant digits.  Equality of the first three digits
therefore gives exactly
\begin{equation}\label{wta5:eq:moments}
|E|=q,
\qquad
\sum_{(a',b')\in E}a'=qa,
\qquad
\sum_{(a',b')\in E}(a')^2=qa^2.
\end{equation}
These are the zeroth, first, and second moments of the selected row labels.  Substituting
\eqref{wta5:eq:moments} into the variance identity gives
\[
\sum_{(a',b')\in E}(a'-a)^2
=
\sum_E(a')^2-2a\sum_Ea'+a^2|E|
=0.
\]
Every term is nonnegative, hence every selected pair field has row label $a$.  Since there are exactly
$q$ such fields and row $a$ contains exactly $q$ pair fields, $E$ is the complete row.  The column proof
is identical, now using the zeroth, first, and second moments of the column labels.

Finally, for any endpoint set $E$,
\[
\sum_{f\in E}c_f-D
=
L\left(\sum_{f\in E}\widehat c_f-\widehat D\right)+|E|-(q+1).
\]
If $E$ is not a code, the integer in parentheses is nonzero; since $|E|\le m$, the absolute value is at
least $L-m$.  This proves \eqref{wta5:eq:codegap}.  The additive decomposition in \eqref{wta5:eq:additive} follows
from \eqref{wta5:eq:oldpair} and \eqref{wta5:eq:completion}.
\end{proof}

\paragraph{Interpretation.}
The lemma is doing two different jobs.  Equation~\eqref{wta5:eq:codegap} is the \emph{enforcement} property:
any endpoint pattern whose budget is close to $D$ must be one code.  Equation~\eqref{wta5:eq:additive} is the
\emph{simulation} property: the completion bid of the unique common field separates additively into a row term and
a column term, which can later be canceled by balancing-field values.

\subsubsection{The Blotto instance and payoff scales}\label{subsec:wta5-instance}

The primaries and $J$ have budget
\[
B_*=D+1,
\]
and each stabilizer has budget
\[
B_{\rm st}=\sum_{f\in\mathcal F_{\rm enc}}(c_f-1).
\]
Let $T_B=3B_*+2B_{\rm st}$ and $M=T_B+1$.  We use four banks of $M$ ordinary buffer fields: the
shared primary bank $\mathcal B_P$, the private stabilizer banks $\mathcal B_S$ and $\mathcal B_T$, and
one private bank for $J$.  Since the entire game contains only $T_B$ troops, every pure profile leaves at
least one field of each bank globally empty.  We next use this pigeonhole property to eliminate allocations that waste troops on zero-valued fields: whenever a player places a troop on such a field, there is an empty field in the corresponding buffer bank to which that troop can be moved.

Let $\ell$ be the source bit length. Recall that $H=256(q+1)q$ is the polynomial amplification factor introduced in Section~\ref{subsec:wta5-construction}. We define
\begin{equation}\label{wta5:eq:scales}
P=100(1+\ell+M),
\qquad
h=P^{-20},
\qquad
\lambda=P^{-44},
\qquad
\eta=H\lambda,
\qquad
\kappa=P^{-48}.
\end{equation}
Only the ordering of these scales matters:
\[
1\gg h\gg \lambda\gg\kappa,
\]
while $\eta=H\lambda$ is the amplified source-payoff scale.  The inspector operates at the $1$ and $h$
scales, code enforcement at the $\lambda$ scale, and $\mathcal B_P$ at the much smaller
$\kappa$ scale.

The construction assigns battlefield value $5\kappa/4$ to each primary on every field of $\mathcal B_P$,
value $5(1+h)/(8B_*)$ to $J$ on each of its private buffer fields, and value $5/4$ to a stabilizer on each
field of its own private buffer bank.  Relative to the all-zero baseline, placing one troop on an otherwise
empty buffer field therefore gives the corresponding player an incremental payoff of, respectively,
\[
\kappa,
\qquad
\frac{1+h}{2B_*},
\qquad
1.
\]
Thus $\mathcal B_S$ and $\mathcal B_T$ literally turn an unused troop into one unit of opportunity cost.

The construction defines two participation fields, $\wtaFivePartR$ and $\wtaFivePartC$, and assigns battlefield value one on $\wtaFivePartR$ to $R$ and $J$, and value one on $\wtaFivePartC$ to $C$ and $J$.  The neutral action $N_R$ puts all $B_*$ troops on $\wtaFivePartR$, while the canonical active
allocation $\wtaFiveCodeR{a}$ puts one troop on $\wtaFivePartR$ and $c_f$ troops on every $f\in\wtaFiveR_a$; $N_C$ and $\wtaFiveCodeC{b}$
are defined symmetrically.  These are named feasible allocations, not restrictions on the Blotto strategy spaces.

The inspector has three named actions.  In $I_R$ it places all $B_*$ troops on $\wtaFivePartR$, and in $I_C$ it
places all $B_*$ troops on $\wtaFivePartC$.  Its outside action $I_{\wtaFiveOut}$ places one troop on each of $B_*$
fixed distinct fields of its private buffer bank; this is feasible because $M>B_*$.  As shown in the first part of Lemma~\ref{wta5:lem:canonical}, allocations that place troops on fields of zero
value can be repaired away by a profitable-deviation argument.  After this repair no other player bids on
$J$'s private buffer fields.  Relative to the common all-zero baseline on this private bank, the \emph{buffer}
contribution of $I_{\wtaFiveOut}$ is therefore exactly $(1+h)/2$.  Before positive participation of active primaries has
been enforced, $I_{\wtaFiveOut}$ may in addition receive an all-zero share on a participation field, so its total adjusted
payoff is at least $(1+h)/2$ at that stage.  After the positive-participation repair below, those extra shares
vanish and its total adjusted payoff is exactly $(1+h)/2$.

\paragraph{Induced local stabilizer game.}
Fix an encoding field $f$.  The construction assigns battlefield value $c_f-3/8$ to each stabilizer on this
field.  Although the stabilizers choose full Blotto allocations rather than independent fieldwise bids, their
private buffer banks make the marginal comparison on $f$ one-dimensional: increasing a stabilizer's bid on
$f$ by $j$ troops can be financed by removing $j$ troops from empty private buffer fields, and therefore has
opportunity cost $j$.  If we temporarily ignore the rare visits of the primaries to $f$, this induces a two-player
local game between $S$ and $T$ on bids $\{0,\ldots,c_f-1\}$.  Against an opponent bid $k$, a positive bid $j$
receives battlefield value $c_f-3/8$ when $j>k$, half of this value when $j=k$, and zero when $j<k$, minus
the opportunity cost $j$.  If both stabilizers bid zero, then all players bid zero on $f$, so equal splitting gives
each stabilizer one fifth of its battlefield value.  Consider the distribution
\begin{equation}\label{wta5:eq:pi}
\pi_f(0)=\frac{5}{8(c_f-3/8)},
\qquad
\pi_f(j)=\frac1{c_f-3/8}\quad(1\le j\le c_f-1).
\end{equation}
Its probabilities sum to one because
\[
\frac{5/8+c_f-1}{c_f-3/8}=1.
\]
Against this distribution, bid zero gives expected payoff
\[
\frac{c_f-3/8}{5}\,\pi_f(0)=\frac18.
\]
For every $1\le j\le c_f-1$,
\[
\Pr[k<j]+\frac12\Pr[k=j]
=\frac{j+1/8}{c_f-3/8},
\]
so bid $j$ gives expected payoff
\[
(c_f-3/8)\left(\Pr[k<j]+\frac12\Pr[k=j]\right)-j=\frac18.
\]
Thus every feasible bid is a best response, and \eqref{wta5:eq:pi} is a symmetric full-support Nash equilibrium of this
induced local stabilizer game.

\paragraph{All-zero payoff fraction.}
The induced local game also determines the all-zero contribution used below.  If a primary is assigned any
battlefield value on $f$ but does not bid there, it receives a positive fraction of that value only when both
stabilizers bid zero.  Since their mixed strategies are independent, this event has probability $\pi_f(0)^2$;
under the all-zero convention the primary then receives one fifth of its battlefield value.  We therefore define
the corresponding all-zero payoff fraction
\begin{equation}\label{wta5:eq:b0}
\beta_f^0=\frac15\pi_f(0)^2
=\frac15\left(\frac{5}{8(c_f-3/8)}\right)^2
=\frac{5}{64(c_f-3/8)^2}.
\end{equation}

\paragraph{Primary battlefield values.}
We now choose the values of the encoding fields for the two primaries.  On a pair field $f_{ab}$, the row
player's value $w_{ab}^R$ is defined by
\begin{equation}\label{wta5:eq:pairvalue}
w_{ab}^R\left(\frac12-\beta_{ab}^0\right)
=\frac{\lambda c_{ab}}2-\eta A_{ab}.
\end{equation}
The left-hand side is the difference between receiving one half of the field value in a tie and receiving
only the reference all-zero fraction $\beta_{ab}^0$ when the row player does not visit the field.  Thus
\eqref{wta5:eq:pairvalue} separates this difference into a structural term $\lambda c_{ab}/2$ and the
source-dependent correction $-\eta A_{ab}$.

The value of the row balancing field is then chosen so that the total reference gain from completing row
code $a$ has a prescribed form:
\begin{equation}\label{wta5:eq:balvalue}
\sum_{f\in\wtaFiveR_a}w_f^R(1-\beta_f^0)=\lambda D+h\lambda u_a.
\end{equation}
The column values are defined symmetrically, with $B_{ab}$ and $v_b$ in place of $A_{ab}$ and $u_a$,
and each primary has value zero on the other primary's balancing fields.  The reason for the extra
$h\lambda u_a$ and $h\lambda v_b$ terms is the additive decomposition
$c_{ab}/2=\theta+u_a+v_b$ from Lemma~\ref{wta5:lem:code}: after the participation probabilities are calibrated
to $h$, these terms cancel the row- and column-dependent part of the completion bids.  Lemma~\ref{wta5:lem:simulation}
performs this cancellation explicitly and leaves the source terms $A_{ab}$ and $B_{ab}$.

\paragraph{Bounds on the primary values.}
The construction keeps every primary value close to the corresponding scale $\lambda c_f$.  Since every
unscaled completion cost is at least two,
\begin{equation}\label{wta5:eq:cmin}
c_f\ge2L+1=2^{21}mH+1.
\end{equation}
Moreover,
\[
\beta_f^0=\frac{5}{64(c_f-3/8)^2}\le\frac{1}{5c_f^2}.
\]
Rewriting \eqref{wta5:eq:pairvalue} as
\[
w_{ab}^R=\lambda\frac{c_{ab}-2HA_{ab}}{1-2\beta_{ab}^0}
\]
and using $0\le A_{ab}\le1$ gives
\begin{equation}\label{wta5:eq:pairvaluebound}
|w_{ab}^R-\lambda c_{ab}|\le3\lambda H,
\end{equation}
with the same bound for $w_{ab}^C$.

We next derive the corresponding estimate for the row balancing field.  Subtracting
$\lambda\sum_{f\in\wtaFiveR_a}c_f=\lambda D$ from \eqref{wta5:eq:balvalue} gives
\[
\begin{aligned}
&w_{\wtaFiveBalR{a}}^R(1-\beta_{\wtaFiveBalR{a}}^0)-\lambda c_{\wtaFiveBalR{a}}\\
&\qquad =h\lambda u_a
-\sum_b\left(w_{ab}^R(1-\beta_{ab}^0)-\lambda c_{ab}\right).
\end{aligned}
\]
For every pair field, \eqref{wta5:eq:pairvaluebound} and the bound
$\beta_{ab}^0\le1/(5c_{ab}^2)$ imply
\[
\begin{aligned}
\left|w_{ab}^R(1-\beta_{ab}^0)-\lambda c_{ab}\right|
&\le |w_{ab}^R-\lambda c_{ab}|+\beta_{ab}^0w_{ab}^R\\
&\le 4\lambda H.
\end{aligned}
\]
Indeed, \eqref{wta5:eq:cmin} gives $3H<c_{ab}/2$, so
\eqref{wta5:eq:pairvaluebound} yields $w_{ab}^R\le3\lambda c_{ab}/2$, and hence
$\beta_{ab}^0w_{ab}^R\le3\lambda/(10c_{ab})<\lambda H$.
Moreover, $u_a<D/(2q)$ and $D<B_*<P/100$ give
$hu_a<P^{-19}/(200q)<H$.  Therefore
\[
\left|w_{\wtaFiveBalR{a}}^R(1-\beta_{\wtaFiveBalR{a}}^0)-\lambda c_{\wtaFiveBalR{a}}\right|
\le(4q+1)\lambda H\le5q\lambda H.
\]
Using again $\beta_f^0\le1/(5c_f^2)$, we obtain
\[
\begin{aligned}
|w_{\wtaFiveBalR{a}}^R-\lambda c_{\wtaFiveBalR{a}}|
&\le
\frac{5q\lambda H+\lambda c_{\wtaFiveBalR{a}}\beta_{\wtaFiveBalR{a}}^0}
     {1-\beta_{\wtaFiveBalR{a}}^0}\\
&\le 8q\lambda H.
\end{aligned}
\]
The same argument applies to the column balancing fields.

These additive estimates imply the two bounds used later.  Every row-side encoding field is either a pair
field, for which the deviation from $\lambda c_f$ is at most $3\lambda H$, or a row balancing field, for
which it is at most $8q\lambda H$.  Since $m=q^2+2q\ge q$, Equation~\eqref{wta5:eq:cmin} gives
\[
8qH\le8mH<\frac{c_f}{2}.
\]
Hence
\[
|w_f^R-\lambda c_f|\le\frac{\lambda c_f}{2},
\]
which in particular implies
$\lambda c_f/2\le w_f^R\le2\lambda c_f$.  Finally, $\mathcal F_R$ contains $q^2$ pair fields and $q$
row balancing fields, so
\[
\begin{aligned}
\sum_{f\in\mathcal F_R}|w_f^R-\lambda c_f|
&\le 3q^2\lambda H+8q^2\lambda H\\
&=11q^2\lambda H
\le11m\lambda H
<100\lambda mH.
\end{aligned}
\]
Consequently
\begin{equation}\label{wta5:eq:valuebounds}
\frac{\lambda c_f}{2}\le w_f^R\le2\lambda c_f
\quad(f\in\mathcal F_R),
\qquad
\sum_{f\in\mathcal F_R}|w_f^R-\lambda c_f|\le100\lambda mH,
\end{equation}
with the same bounds for $C$.

These are the only quantitative properties of the primary values used in the canonicalization argument.
We next record a uniform bound on the raw payoff that either primary can obtain away from its participation
field.  Such payoff can come only from the primary's encoding fields and from the shared buffer bank
$\mathcal B_P$, so it is bounded by
\begin{equation}\label{wta5:eq:outsidevalue}
\max\left\{\sum_{f\in\mathcal F_R} w_f^R,\;\sum_{f\in\mathcal F_C} w_f^C\right\}
+\frac{5\kappa M}{4}
<3\lambda M<3\lambda P.
\end{equation}
Indeed, Equation~\eqref{wta5:eq:valuebounds} gives
\[
\sum_{f\in\mathcal F_R}w_f^R
\le 2\lambda\sum_{f\in\mathcal F_R}c_f
\le 2\lambda\sum_{f\in\mathcal F_{\rm enc}}c_f
=2\lambda(B_{\rm st}+m),
\]
and the same bound holds for the column primary.  By the definition of $B_{\rm st}$ and
Equation~\eqref{wta5:eq:cmin}, every encoding field contributes at least one unit to
$\sum_{f\in\mathcal F_{\rm enc}}(c_f-1)$, so $m\le B_{\rm st}$.  Since
$M=3B_*+2B_{\rm st}+1$, it follows that $B_{\rm st}+m\le2B_{\rm st}<M$, and hence the first term in
\eqref{wta5:eq:outsidevalue} is smaller than $2\lambda M$.

For the second term, the construction assigns value $5\kappa/4$ to a primary on each of the $M$ fields of
$\mathcal B_P$.  Thus $5\kappa M/4$ bounds its entire raw payoff from that bank, including all-zero shares.
Finally, $\kappa/\lambda=P^{-4}\le10^{-8}$ gives
\[
2\lambda M+\frac{5\kappa M}{4}<3\lambda M,
\]
while $M<P/100<P$ gives the last inequality in \eqref{wta5:eq:outsidevalue}.

\paragraph{Normalization and accuracy.}
The preceding values are the \emph{raw} battlefield values used in the proof.  The total raw battlefield
value accessible to any one player is smaller than $P$: for a primary this follows from
\eqref{wta5:eq:outsidevalue} and its participation-field value one; for a stabilizer and for $J$ it follows from
$B_{\rm st}<M/2$ and $M<P/100$.  Dividing every battlefield value by $P$ therefore normalizes all utilities
to $[0,1]$.  We fix
\[
\eps=P^{-1000}.
\]

Thus the initial raw regret is at most $P\eps=P^{-999}$.  Moreover the explicit instance contains all $M$
buffer fields and the source payoff data, so after harmless padding $n_{\rm inst}\ge300$ and
\[
P=100(1+\ell+M)\le300n_{\rm inst}\le n_{\rm inst}^2.
\]
Consequently $n_{\rm inst}^{-2000}\le P^{-1000}$, so the target accuracy used in the reduction is inverse-polynomial in $n_{\rm inst}$, as required by Theorem~\ref{thm:wta-five-player-hardness}.

\subsubsection{Canonicalization: arbitrary Blotto allocations become codes}\label{subsec:wta5-canonical}

For a Blotto player $i$ and a pure allocation $\alpha_i$, let $U_i(\alpha_i)$ denote its full expected raw
payoff from $\alpha_i$ against the current mixed strategies of the other players.

The next lemma is the structural heart of the proof.  Its job is to start from the full unrestricted Blotto
strategy space and show that, after changing at most $P^{-800}$ total probability mass in the regret-based
repairs (together with a harmless payoff-monotone normalization inside private buffer banks), each primary is
either neutral or is playing one exact source-action code.

\begin{lemma}[Canonicalization]\label{wta5:lem:canonical}
From every normalized $P^{-1000}$-Nash equilibrium one can obtain a profile of raw regret at most
$\eps_{\rm raw}:=P^{-800}$ with the following properties.  In the regret-based repairs below, the total probability
mass that is replaced is at most $P^{-800}$.  In addition, after zero-valued fields have been removed, we apply a
payoff-monotone normalization inside the private buffer banks of $S,T$, and $J$; this auxiliary normalization may
act on arbitrary probability mass, but it leaves every other player's payoff unchanged and does not increase any
regret.
\[
\begin{aligned}
\operatorname{supp}(\mu_R)&\subseteq\{N_R,\wtaFiveCodeR{0},\ldots,\wtaFiveCodeR{q-1}\},\\
\operatorname{supp}(\mu_C)&\subseteq\{N_C,\wtaFiveCodeC{0},\ldots,\wtaFiveCodeC{q-1}\}.
\end{aligned}
\]
Moreover, if $x_i$ is primary $i$'s probability of nonneutral play, then
\[
x_R,x_C\le2h,
\]

$J$ uses only $I_R,I_C,I_{\wtaFiveOut}$, and every stabilizer bid on $f$ lies in $\{0,\ldots,c_f-1\}$.
\end{lemma}

\begin{proof}
The proof first eliminates allocations that waste troops on fields of zero value and derives a preliminary
activation bound.  It then uses one local estimate for the stabilizer contests and a global defect argument
to round every relevant nonneutral allocation to a unique code.  Before any repair, the normalized
$P^{-1000}$-equilibrium has raw regret at most $P^{-999}$.  We will update this bound explicitly after each
repair.  Since every raw utility lies in $[0,P]$, changing total probability mass $\delta$ in the players'
mixed strategies changes any raw regret by at most $2P\delta$.  The private-buffer normalization used below
is different: because it weakly increases the normalized player's payoff and leaves every other player's payoff
unchanged, it cannot increase any player's regret and therefore requires no small-mass estimate.

\wtaFiveLead{Removing troops from zero-valued fields.}
We first eliminate allocations that place a troop on a field of zero value.  Fix such an allocation and move
one zero-value troop to a field chosen uniformly from the player's corresponding buffer bank.  Removing the
troop from its original field causes no payoff loss.  Consider first a primary.  Each field in $\mathcal B_P$
has battlefield value $5\kappa/4$.  If the selected buffer field is globally empty, then before the move the
primary receives only the all-zero share
\[
\frac15\cdot\frac{5\kappa}{4}=\frac{\kappa}{4},
\]
whereas after placing one troop there it is the unique positive bidder and receives $5\kappa/4$.  The gain is
therefore exactly $\kappa$.  If the selected buffer field is already occupied, adding the troop cannot reduce
the primary's payoff on that field: before the move it bids zero and receives zero, while afterwards it either
still loses, ties, or wins.  The analogous gain from an empty buffer field is larger for $J,S,$ and $T$.
Since every pure profile leaves at least one of the $M$ fields in the relevant buffer bank globally empty, the
uniformly random move has expected gain at least $\kappa/M$.

For each player, let $\delta$ be the probability assigned to allocations containing a zero-value troop.
At the initial profile, define a unilateral deviation that, on each such realization, selects one zero-value troop
and moves that troop to a uniformly random field of the corresponding buffer bank, while leaving all other
realizations unchanged.  The deviation depends only on the player's own realized allocation and its own fresh
randomization, so it is a valid mixed-strategy deviation.  By the preceding calculation its expected gain is at
least $\delta\kappa/M$.  Since the raw regret is at most $P^{-999}$ and
$M/\kappa< P^{49}/100$, we obtain
\[
\delta\le \frac{MP^{-999}}{\kappa}<\frac{P^{-950}}{100}.
\]
This deviation is used only to prove that the bad mass is small.  Having obtained the bound on $\delta$, we
replace every pure allocation containing at least one zero-value troop by any fixed feasible allocation for that
player using only fields that it values or its designated buffer bank.  Thus all zero-value troops are removed in
one repair, even if a bad allocation contained several of them.  Repairing all five players therefore changes
total probability mass by less than $P^{-950}$.  The raw regret after this repair is consequently smaller than
\[
P^{-999}+2P\cdot P^{-950}<P^{-948}.
\]
We may henceforth assume that every player uses only fields it values or its corresponding buffer bank.

\wtaFiveLead{Payoff-monotone normalization inside private buffer banks.}
We now make precise a harmless normalization that is used implicitly by the local stabilizer argument and by
the inspector canonicalization.  After the preceding repair, no player other than the owner of a private buffer
bank places a positive bid on that bank.  Consider first a stabilizer.  If one of its pure allocations places at
least two troops on one private-buffer field, one of those troops can be moved to an unused field of the same
private bank.  The original field remains positively occupied, so the stabilizer loses no payoff there; on the
new field the move changes the all-zero share into a unique positive win and increases the stabilizer's raw
payoff by one.  Repeating this operation yields an allocation in which every occupied private-buffer field
contains exactly one troop.  The same operation for $J$ increases its payoff by
$(1+h)/(2B_*)$ whenever a duplicate is split.  Once all occupied private-buffer fields are singletons, fix an
ordering of each private bank and move the $k$ occupied singleton bids to the first $k$ fields of that ordering.
Because all private-buffer fields have the same value and no other player bids there, this relocation does not
change the owner's payoff.

Apply this deterministic normalization to every pure allocation in the mixed strategies of $S,T$, and $J$.
It does not alter any bid on an encoding or participation field.  It also does not alter any other player's
payoff, because the other players assign value zero to these private-buffer fields and, after the preceding
repair, never bid on them.  The normalized player's expected payoff weakly increases while its best-response
value against the unchanged opponents is the same.  Hence its regret weakly decreases, and every other
player's regret is unchanged.  This normalization may be applied to arbitrary probability mass and is therefore
not included in the $P^{-800}$ small-mass accounting in the statement of the lemma.  From this point onward,
every stabilizer and inspector private-buffer bid is either zero or one, and for each possible number $k$ of
private-buffer troops the canonical locations are the first $k$ fields of the fixed ordering.

\wtaFiveLead{Capping stabilizer bids.}
The same reasoning removes stabilizer overbids.  After the preceding repair, a stabilizer's private buffers
are not visited by other players.  If it bids $j\ge c_f$ on encoding field $f$, moving those $j$ troops to
distinct unused private-buffer fields gains at least
\[
j-(c_f-3/8)\ge\frac38.
\]
For an allocation containing one or more overbids, the stabilizer can select one such field and perform the
preceding move; this defines a valid deviation with gain at least $3/8$ on every bad realization.  Hence, since
the current raw regret is below $P^{-948}$, each stabilizer assigns probability at most
$(8/3)P^{-948}$ to allocations containing at least one overbid.  After this probability bound has been
established, we repair every such pure allocation completely, for example by moving all troops involved in
overbids to distinct unused private-buffer fields.  If there are several overbids, the gains from removing them
are nonnegative field by field, and there are enough unused private-buffer fields because $M$ exceeds the total
number of troops in the game.  Repairing both stabilizers changes total mass by less than $P^{-947}$, so the
raw regret afterwards is smaller than
\[
P^{-948}+2P\cdot P^{-947}<P^{-945}.
\]
Every stabilizer bid on $f$ now lies in $\{0,\ldots,c_f-1\}$.  For any capped pure allocation, the number of
private-buffer troops is
\[
B_{\rm st}-\sum_{g\in\mathcal F_{\rm enc}}t_g
=\sum_{g\in\mathcal F_{\rm enc}}(c_g-1-t_g),
\]
so it is at least the slack $c_f-1-t_f$ on any one encoding field $f$.  Because the private-buffer troops are
singletons after the normalization above, increasing the bid on $f$ by one costs exactly one unit of buffer
payoff per transferred troop; decreasing the bid can be financed in the reverse direction using unused private
buffer fields.  Applying this modification to each realized pure allocation defines a valid unilateral mixed
strategy.  Hence the stabilizer's global regret bounds the corresponding local regret on each encoding field.

\wtaFiveLead{A preliminary bound on primary activation.}
For the inspector calculations from this point onward, subtract from every payoff of $J$ the constant
\[
B_J^0:=M\frac{1+h}{8B_*},
\]
which is the sum of the all-zero shares of all fields in $J$'s private buffer bank after the zero-value repair.
This is the same constant for every action of $J$, so the subtraction does not change any best response or
regret.  We continue to write $U_J$ for this adjusted payoff.  In this convention, each singleton troop placed
on an otherwise empty private-buffer field contributes exactly $(1+h)/(2B_*)$ above the baseline.

Let $\iota_i$ be the probability that $J$ fully inspects $\mathrm{part}_i$.  A neutral deviation by primary $i$ gets
participation payoff $1-\iota_i/2$, so its current participation payoff is at least
$1-\iota_i/2-3\lambda P-P^{-945}$.  Since the primary and $J$ together receive at most one unit from $\mathrm{part}_i$, while $J$'s
adjusted reserve payoff is at most $(1+h)/2$ whenever it is not fully inspecting, we obtain
\begin{equation}\label{wta5:eq:Jupper}
U_J\le
\frac{1+h}{2}(1-\iota_R-\iota_C)+\frac{\iota_R+\iota_C}{2}+6\lambda P+2P^{-945}.
\end{equation}
The outside action $I_{\wtaFiveOut}$ guarantees payoff $(1+h)/2$.  Since the current raw regret is at most
$P^{-945}$, the inspector's equilibrium payoff must therefore satisfy
\[
U_J\ge \frac{1+h}{2}-P^{-945}.
\]
Combining this lower bound with \eqref{wta5:eq:Jupper} yields
\[
\frac{1+h}{2}-P^{-945}
\le
\frac{1+h}{2}(1-\iota_R-\iota_C)
+\frac{\iota_R+\iota_C}{2}
+6\lambda P+2P^{-945}.
\]
Expanding the first term on the right and cancelling $(1+h)/2$ from both sides gives
\[
\frac{h}{2}(\iota_R+\iota_C)
\le 6\lambda P+3P^{-945}.
\]
Hence
\[
\iota_R+\iota_C
\le \frac{12\lambda P+6P^{-945}}{h}.
\]
Using $h=P^{-20}$ and $\lambda=P^{-44}$, this becomes
\[
\iota_R+\iota_C
\le 12P^{-23}+6P^{-925}
<\frac18,
\]
where the final inequality follows from $P\ge100$.  Fully inspecting primary $i$ guarantees
$(1+x_i)/2$: it ties the neutral bid and beats every nonneutral participation bid.  Similarly,
\[
x_i\le h+12\lambda P+6P^{-945}
\le h\bigl(1+12P^{-23}+6P^{-925}\bigr)<2h.
\]
Thus
\begin{equation}\label{wta5:eq:early}
\iota_R+\iota_C\le\frac18,
\qquad
x_R,x_C\le2h.
\end{equation}
This is the crucial ordering fact: arbitrary nonneutral primary behavior is already rare before we invoke
any code structure.

\wtaFiveLead{Removing partial inspector allocations.}
We next remove partial inspector allocations.  A \emph{partial inspection} sends a positive but non-full
amount of $J$'s budget to one participation field.  Move one of those troops to the next unused field in the fixed
ordering of $J$'s private buffer bank.  When the corresponding primary is neutral, which occurs with probability at least
$1-2h$ by \eqref{wta5:eq:early}, the partial bid already loses the participation field, so the move loses nothing
there and gains $(1+h)/(2B_*)$ from the buffer.  On the exceptional nonneutral event the loss is at most
one.  Hence the expected gain is at least
\[
(1-2h)\frac{1+h}{2B_*}-2h.
\]
Because $M=3B_*+2B_{\rm st}+1>B_*$ and $P=100(1+\ell+M)$, we have $B_*<P/100$ and $P\ge100$.  Therefore
\[
\frac12-h-2h^2-4hB_*
>\frac12-P^{-20}-2P^{-40}-\frac1{25}P^{-19}>0,
\]
so
\[
(1-2h)\frac{1+h}{2B_*}-2h>\frac1{4B_*}.
\]
Consequently, if $J$ assigns probability $\delta_J$ to partial inspections, the unilateral deviation that
performs the preceding one-troop move on each partial-inspection realization has expected gain at least
$\delta_J/(4B_*)$.  Using the current raw-regret bound $P^{-945}$ and $B_*<P/100$,
\[
\delta_J\le4B_*P^{-945}<P^{-944}.
\]
We now use this probability bound to eliminate all partial inspections at once: every partial-inspection pure
allocation is replaced by $I_{\wtaFiveOut}$.  This changes only $\delta_J$ probability mass, irrespective of how many
participation troops the original allocation contained.  After this repair the raw regret is smaller than
\[
P^{-945}+2P\cdot P^{-944}<P^{-942}.
\]
Any inspector allocation that was not repaired has no partial participation bid.  Since its budget is $B_*$,
it is therefore either $I_R$, $I_C$, or an allocation placing no troop on either participation field.  In the last
case all $B_*$ troops lie in the private buffer bank, and the private-buffer normalization above makes that
allocation exactly $I_{\wtaFiveOut}$.  Hence we may thereafter assume that $J$ uses only $I_R,I_C,I_{\wtaFiveOut}$.

\wtaFiveLead{Positive participation by active primaries.}
Likewise, consider a nonneutral allocation of primary $i\in\{R,C\}$ that places zero troops on its own
participation field $\mathrm{part}_i$.  Fix any canonical code allocation of that primary.  By construction,
the code places one troop on $\mathrm{part}_i$.
Suppose first that $J$ does not inspect $\mathrm{part}_i$.  After the preceding elimination of bids on
zero-valued fields, no other player makes a positive bid on this participation field.  Under the original
allocation all bids on $\mathrm{part}_i$ are therefore zero, so equal splitting gives primary $i$ only the
all-zero share $1/5$ of the field's value one.  Under the code deviation, primary $i$ bids one troop and is
the unique positive bidder, so it receives the full value one.  The gain on the participation field is thus
exactly
\[
1-\frac15=\frac45.
\]
If instead $J$ uses $I_i$, then $J$ bids all $B_*$ troops on $\mathrm{part}_i$.  The primary loses the field
both with its original bid zero and with the code's bid one, so the deviation causes no loss on the
participation field in this event.  Consequently the expected gain from the participation field alone is
\[
(1-\iota_i)\frac45.
\]
By \eqref{wta5:eq:early}, $\iota_i\le1/8$, so this gain is at least $\frac78\cdot\frac45=\frac7{10}$.

Changing to the fixed code also changes the bids on the other battlefields.  In the worst case the primary
could lose all of its payoff outside its participation field; Equation~\eqref{wta5:eq:outsidevalue} bounds this
possible loss by $3\lambda P$.  The deviation therefore improves its expected raw payoff by at least
\[
\frac7{10}-3\lambda P
=\frac7{10}-3P^{-43}
>\frac12,
\]
since $P\ge100$.  If such zero-participation nonneutral allocations have total probability $\delta_i$, replacing only those
realizations by the fixed code improves expected payoff by at least $\delta_i/2$.  The current raw regret is
below $P^{-942}$, so, for each primary,
\begin{equation}\label{wta5:eq:zero-participation-mass}
\delta_i<2P^{-942}.
\end{equation}
By Equation~\eqref{wta5:eq:zero-participation-mass}, repairing both primaries changes total mass by less than
$P^{-941}$ and leaves raw regret below
\[
P^{-942}+2P\cdot P^{-941}<P^{-939}.
\]
After this repair, every relevant nonneutral allocation bids at least one troop on its participation field.

\wtaFiveLead{Local quadratic return.}
First we introduce a local stability estimate, which is used in two places.  In the present section, it yields the
return bounds~\eqref{wta5:eq:shape1}--\eqref{wta5:eq:shape3} that drive the canonicalization argument.  It is used
again in Section~\ref{subsec:wta5-simulation} to control the stabilizers' zero-bid probabilities under the
effective prize induced by primary visits.

\begin{proposition}[Local stabilizer estimate]\label{wta5:clm:local}
Let $c\ge3$ be odd, let $V=c-1+\delta$ with $1/2\le\delta\le5/8$, and consider the induced two-player
stabilizer game on bids $\{0,\ldots,c-1\}$ with opportunity cost one per troop and all-zero payoff $V/5$.
Its symmetric equilibrium is
\begin{equation}\label{wta5:eq:localexact}
q_0^\delta=\frac{\delta}{V},
\qquad
q_j^\delta=
\begin{cases}
(2-8\delta/5)/V,&j\text{ odd},\\
(8\delta/5)/V,&j\ge2\text{ even}.
\end{cases}
\end{equation}
If both marginal bid distributions have local regret at most $r$ and $10^8c^4r\le1$, then each differs
coordinatewise from $q^\delta$ by at most $10^4c^2r$.
\end{proposition}

\begin{proof}
We first identify the exact symmetric equilibrium.  The construction-specific case $\delta=5/8$ was already
verified in Section~\ref{subsec:wta5-instance}: Equation~\eqref{wta5:eq:pi} is a symmetric full-support Nash equilibrium
of the induced local stabilizer game.  The present calculation extends the same full-support indifference argument
to every $\delta\in[1/2,5/8]$.  Let $q=(q_0,\ldots,q_{c-1})$ be the opponent's marginal bid
distribution.  A positive bid $j$ has payoff
\[
V\left(\Pr[W<j]+\frac12\Pr[W=j]\right)-j,
\]
whereas bid zero has payoff $Vq_0/5$.  Since the candidate equilibrium in \eqref{wta5:eq:localexact} has
positive probability on every bid, all pure bids in its support must be indifferent.  Thus, for $j\ge2$,
we compare the utilities of bids $j$ and $j-1$:
\[
V\left(\Pr[W<j]+\frac12\Pr[W=j]\right)-j
=
V\left(\Pr[W<j-1]+\frac12\Pr[W=j-1]\right)-(j-1),
\]
which gives
\[
q_{j-1}+q_j=\frac2V.
\]
Similarly, comparing the utility of bid $1$ with the all-zero payoff of bid $0$ gives
\[
V\left(q_0+\frac12q_1\right)-1=\frac{Vq_0}{5},
\]
and therefore
\[
\frac45q_0+\frac12q_1=\frac1V.
\]
The first relation makes the positive probabilities alternate.  Since $c-1$ is even, normalization then
gives $q_0=\delta/V$, and substitution yields exactly \eqref{wta5:eq:localexact}.  For
$\delta\in[1/2,5/8]$ every coordinate is positive; in particular every coordinate of $q^\delta$ is at
least $1/(2c)$.

We next prove stability without assuming that the approximate marginals have full support.  Let $p$ and
$q$ be the two marginal bid distributions.  Adding the opponent's bid and subtracting $V/2$ from each
player's payoff does not change any unilateral payoff difference.  Away from the all-zero outcome, the
resulting row payoff is the skew-symmetric matrix
\[
S_{zw}=\frac V2\operatorname{sgn}(z-w)-z+w.
\]
The all-zero outcome contributes the additional term $-3V/10$.  If $e_0$ is the unit vector of bid zero,
set
\[
s=\frac{3Vp_0}{10},
\qquad
 t=\frac{3Vq_0}{10}.
\]
For the fixed profile $(p,q)$, both players' regrets are therefore the same as in the zero-sum game with
row-payoff matrix
\[
Z=S-te_0\one^\top+s\one e_0^\top.
\]
Indeed, the extra terms that distinguish $Z$ from the transformed local game are constant across the
player's own pure bids and hence do not affect best-response differences.

The corresponding zero-sum game has explicit full-support equalizers.  Define
\[
\chi_j=
\begin{cases}
\delta/V,&j\text{ even},\\
(2-\delta)/V,&j\text{ odd},
\end{cases}
\qquad
\zeta_j=
\begin{cases}
0,&j=0,\\
2/V,&j>0\text{ even},\\
-2/V,&j\text{ odd}.
\end{cases}
\]
A direct summation gives
\[
S\chi=0,
\qquad
S\zeta=e_0-\frac{\delta}{V}\one,
\qquad
\one^\top\chi=1,
\qquad
\one^\top\zeta=0.
\]
Consequently $\chi+s\zeta$ and $\chi+t\zeta$ equalize the two sides of $Z$.

To use these equalizers in the duality-gap argument below, they must be valid mixed strategies and, more
importantly, must assign a nonnegligible probability to every bid.  Since $\zeta$ has both positive and
negative coordinates, this reduces to showing that the coefficients $s$ and $t$ are bounded away from the
values at which some coordinate of $\chi+s\zeta$ or $\chi+t\zeta$ could vanish.  We show
$s,t<2/3$.  By symmetry it is enough to prove the bound for $t$.

Put $\alpha=\delta/(4V)$.  This threshold separates the case in which bid zero is played often enough
that the regret bound controls its payoff directly from the case in which its probability is already small.
If bid zero is played with probability $p_0\ge\alpha$, then its payoff can be at most $r/p_0$ below a
best response; otherwise the contribution of bid zero alone to the row player's regret would exceed $r$.
We therefore distinguish three cases.

\paragraph{Case 1: $p_0\ge\alpha$.}
Comparing the row player's bids zero and one gives
\[
\frac{4Vq_0}{5}-1\le\frac{r}{p_0}
\le\frac{4Vr}{\delta},
\]
which implies $t=3Vq_0/10<2/3$ under $10^8c^4r\le1$.

\paragraph{Case 2: $p_0<\alpha$ and $q_0<\alpha$.}
Here the desired bound follows directly from the definition of $t$:
\[
t=\frac{3Vq_0}{10}<\frac{3V\alpha}{10}=\frac{3\delta}{40}<\frac23.
\]

\paragraph{Case 3: $p_0<\alpha\le q_0$.}
Now the column player uses bid zero with probability at least $\alpha$.  Since its total regret is at most
$r$, its zero bid is within $r/q_0\le r/\alpha=4Vr/\delta$ of a best response.  Against the row
marginal $p$, the column player's top bid $c-1$ has payoff
\[
\delta-\frac{V}{2}p_{c-1},
\]
whereas bid zero has payoff $Vp_0/5$.  Comparing these two bids therefore gives
\[
\delta-\frac{V}{2}p_{c-1}-\frac{Vp_0}{5}
\le \frac{4Vr}{\delta}.
\]
Because $p_0<\alpha=\delta/(4V)$, we have $Vp_0/5<\delta/20$; moreover
$10^8c^4r\le1$, together with $V<c$ and $\delta\ge1/2$, gives
$4Vr/\delta<\delta/20$.  Hence
\[
\frac{V}{2}p_{c-1}
>\delta-\frac{\delta}{20}-\frac{\delta}{20}
=\frac{9\delta}{10},
\]
so in particular $p_{c-1}\ge\delta/V$.

The row player therefore uses bid $c-1$ with probability at least $\delta/V$.  Its regret bound implies
that this top bid is within $r/p_{c-1}\le Vr/\delta$ of a best response.  Comparing it with bid $1$ gives
\[
V\left(q_0+\frac12q_1\right)-1
-\left(\delta-\frac{V}{2}q_{c-1}\right)
\le\frac{Vr}{\delta}.
\]
Dropping the nonnegative terms involving $q_1$ and $q_{c-1}$ yields
\[
Vq_0\le1+\delta+\frac{Vr}{\delta}.
\]
The assumed regret bound also gives $Vr/\delta<1/100$, and since $\delta\le5/8$,
\[
t=\frac{3Vq_0}{10}
\le\frac3{10}\left(1+\delta+\frac{Vr}{\delta}\right)
<\frac12<\frac23.
\]

These three cases cover all possibilities for $t$.  Interchanging the two players gives the symmetric bound
$s<2/3$.  Hence every coordinate of both $\chi+s\zeta$ and $\chi+t\zeta$ is positive and, in fact, at
least
\[
\frac1{24V}>\frac1{24c}.
\]

The two regrets are at most $r$, so the zero-sum duality gap is at most $2r$.  Let $\nu$ denote the value
of $Z$.  Because the equalizers above put probability at least $1/(24c)$ on every pure bid, the duality-gap
bound controls every pure payoff, not only their average.  Concretely, every row pure-bid payoff against
$q$ lies in
\[
[\nu-48cr,\,\nu+2r],
\]
and the analogous statement holds for the column player.  Thus any two pure bids of the same player have
payoffs differing by at most $96cr$.

We finally translate this payoff closeness into coordinate closeness.  Let
$e_j=q_j-q_j^\delta$.  For $j\ge2$,
\[
\left|\frac V2(q_{j-1}+q_j)-1\right|\le96cr,
\]
while comparison of bids one and zero gives
\[
\left|V\left(\frac45q_0+\frac12q_1\right)-1\right|\le96cr.
\]
Subtracting the corresponding equalities for $q^\delta$ and using $V\ge c/2$ yields
\[
|e_{j-1}+e_j|\le384r\qquad(j\ge2),
\qquad
\left|\frac45e_0+\frac12e_1\right|\le192r.
\]
Since $c-1$ is even and $\sum_j e_j=0$, summing the pairs
$(e_1+e_2),(e_3+e_4),\ldots$ gives $|e_0|\le192cr$.  The second displayed inequality then gives
$|e_1|\le692cr$, and iterating the adjacent-pair recurrence gives
\[
|e_j|\le1076cr<10^4c^2r
\qquad(0\le j<c).
\]
The same argument applies to $p$, proving the claim.  Notice that the stability statement does not
assume that the approximate marginals have full support.  In particular, setting $r=0$ shows that
$q^\delta$ is the unique Nash equilibrium of the induced local game; hence possible non-full-support
equilibria are ruled out by the stability argument itself, rather than by the preceding full-support
construction.
\end{proof}

Return to an encoding field $f$ and write $c=c_f$.  By \eqref{wta5:eq:completion}, $c$ is odd.  Primary visits
occur with probability at most $x_R+x_C\le4h$, so replacing the actual contest by the induced game with
$\delta=5/8$ changes the regret by at most $8hc$.  The current raw regret is below $P^{-939}$ and
$c-1\le B_{\rm st}<M<P/100$, hence the induced local regret is at most
\[
r_f:=P^{-939}+8hc<P^{-19}.
\]
Therefore
\[
10^8c^4r_f
<10^8\left(\frac{P}{100}\right)^4P^{-19}
=P^{-15}<1,
\qquad
10^4c^2r_f\le\frac{1}{10^4c^2}.
\]
We next extract the three consequences of Proposition~\ref{wta5:clm:local} that are used in the rounding argument.
For the present application $\delta=5/8$, so the exact distribution in \eqref{wta5:eq:localexact} has
probability $5/[8(c_f-3/8)]$ on bid zero and probability $1/(c_f-3/8)$ on every positive bid.  The
coordinatewise error above is at most $1/(10^4c_f^2)$.  Hence, for $0\le j<c_f$,
\[
\Pr[\text{stabilizer bid}\le j]
\le
\frac{j+5/8}{c_f-3/8}+\frac{j+1}{10^4c_f^2}
\le\frac{j+1.0001}{c_f},
\]
where the second inequality uses
$(j+5/8)/(c_f-3/8)\le(j+1)/c_f$ for $j\le c_f-1$.  Similarly, every
positive bid $\ell$ satisfies
\[
\Pr[\text{stabilizer bid}=\ell]
\ge\frac1{c_f-3/8}-\frac1{10^4c_f^2}
>\frac1{c_f}.
\]

Let $g_f(j)$ be the expected share of the field won by a primary that bids $j$ against the two independent
stabilizers.  To obtain any positive share, both stabilizers must bid at most $j$, so
\[
g_f(j)
\le
\Pr[S_f\le j]\Pr[T_f\le j]
\le\left(\frac{j+1.0001}{c_f}\right)^2.
\]
Consequently, for $1\le j<c_f$,
\[
c_fg_f(j)
\le\frac{(j+1.0001)^2}{c_f}
\le\frac{j^2}{c_f}+4,
\]
and, if $1\le j\le c_f/4$,
\[
g_f(j)
\le\left(\frac{j+1.0001}{c_f}\right)^2
\le5\frac{j^2}{c_f^2}.
\]
For the lower-tail estimate, put $j=c_f-\Delta$ with $1\le\Delta\le c_f/4$.  If one fixed stabilizer,
say $S$, bids at least $j$, then the primary either loses the field or ties for the highest bid, and hence
loses at least one half of the full-field payoff.  The bids $j,j+1,\ldots,c_f-1$ are all positive and there
are exactly $\Delta$ of them, each with probability at least $1/c_f$.  Therefore
\[
1-g_f(c_f-\Delta)
\ge\frac12\Pr[S_f\ge c_f-\Delta]
\ge\frac{\Delta}{2c_f}.
\]
We have thus proved
\begin{align}
c_fg_f(j)&\le\frac{j^2}{c_f}+4 &&(1\le j<c_f),\label{wta5:eq:shape1}\\
g_f(j)&\le5\frac{j^2}{c_f^2} &&(1\le j\le c_f/4),\label{wta5:eq:shape2}\\
c_f[1-g_f(c_f-\Delta)]&\ge\frac{\Delta}{2}
&& (1\le \Delta\le c_f/4).\label{wta5:eq:shape3}
\end{align}
These three inequalities are exactly the properties of stabilizer play used in the remainder of the
canonicalization argument.  The global rounding proof below does not use the explicit equilibrium formula
~\eqref{wta5:eq:localexact}: Equation~\eqref{wta5:eq:shape1} controls the return from an arbitrary intermediate bid,
Equation~\eqref{wta5:eq:shape2} makes small stray bids quadratically negligible, and
Equation~\eqref{wta5:eq:shape3} gives a linear loss when a bid falls short of its completion level.  Thus
Proposition~\ref{wta5:clm:local} enters the rest of Lemma~\ref{wta5:lem:canonical} only through
~\eqref{wta5:eq:shape1}--\eqref{wta5:eq:shape3}.

\wtaFiveLead{Global rounding to a code.}
We prove the row case.  We first condition on $C$ playing its neutral allocation; under this conditioning,
the argument splits into two cases according to the size of the defect introduced below.  For each positively
valued row-side encoding field $f\in\mathcal F_R$, let $t_f$ be $R$'s bid on $f$.  Let $t_{\rm buf}$ be the
number of troops that $R$ places in $\mathcal B_P$, and let $t_{\rm part}\ge1$ be its bid on $\wtaFivePartR$.
After the preceding repair there are no other troops, so
\begin{equation}\label{wta5:eq:budget}
\sum_{f\in\mathcal F_R}t_f+t_{\rm buf}+(t_{\rm part}-1)=D.
\end{equation}
For every encoding field, round $t_f$ to the nearer endpoint
$\bar t_f\in\{0,c_f\}$, breaking an exact tie arbitrarily.  Define the total defect
\begin{equation}\label{wta5:eq:defect}
\Phi=t_{\rm buf}+(t_{\rm part}-1)+\sum_{f\in\mathcal F_R}|t_f-\bar t_f|.
\end{equation}
This single number measures every way in which the allocation fails to be a complete codeword: troops
left in the tiny buffer, excess participation troops, partial completion, stray bids, and overbids.

Let $U_R^0(t)$ denote $R$'s expected payoff from its encoding fields and $\mathcal B_P$, conditional on
$C$ being neutral, after subtracting from each field the payoff that $R$ would receive there with bid zero.
We now derive the defect bound field by field.  If $0<t_f<c_f$, then the adjusted contribution of field $f$
is at most $w_f^Rg_f(t_f)$.  By \eqref{wta5:eq:shape1},
\[
w_f^Rg_f(t_f)
\le \lambda c_fg_f(t_f)+|w_f^R-\lambda c_f|
\le \lambda\frac{t_f^2}{c_f}+4\lambda+|w_f^R-\lambda c_f|.
\]
Because $\bar t_f$ is the nearer endpoint in $\{0,c_f\}$, the elementary inequality
\[
\frac{t_f^2}{c_f}
\le t_f-\frac12|t_f-\bar t_f|
\qquad(0\le t_f\le c_f)
\]
holds: for $t_f\le c_f/2$ it is $t_f^2/c_f\le t_f/2$, while for
$t_f\ge c_f/2$ it is
$t_f^2/c_f\le 3t_f/2-c_f/2$.  If $t_f=c_f$, the primary wins the field outright, and if $t_f>c_f$ it
still wins outright; in both cases the same upper bound below follows directly from
$w_f^R\le\lambda c_f+|w_f^R-\lambda c_f|$.  Thus for every $t_f\ge0$ the adjusted contribution of field
$f$ is at most
\[
\lambda t_f-\frac{\lambda}{2}|t_f-\bar t_f|
+4\lambda+|w_f^R-\lambda c_f|.
\]
Summing over $f\in\mathcal F_R$, using $|\mathcal F_R|\le m$ and
\eqref{wta5:eq:valuebounds}, gives an encoding-field contribution at most
\[
\lambda\sum_{f\in\mathcal F_R}t_f
-\frac{\lambda}{2}\sum_{f\in\mathcal F_R}|t_f-\bar t_f|
+4\lambda m+100\lambda mH.
\]
Conditional on $C$ being neutral, an occupied field of $\mathcal B_P$ contributes at most $\kappa$ above
its all-zero baseline, so the entire buffer contribution is at most $\kappa t_{\rm buf}$.  Substituting the
budget identity \eqref{wta5:eq:budget} and using $\kappa<\lambda/2$ therefore yields
\begin{align*}
U_R^0(t)
&\le \lambda\bigl(D-t_{\rm buf}-(t_{\rm part}-1)\bigr)
-\frac{\lambda}{2}\sum_{f\in\mathcal F_R}|t_f-\bar t_f|
+\kappa t_{\rm buf}+4\lambda m+100\lambda mH\\
&\le
\lambda D-\frac{\lambda}{2}\Phi+4\lambda m+100\lambda mH.
\end{align*}
Hence
\begin{equation}\label{wta5:eq:defectbound}
U_R^0(t)
\le
\lambda D-\frac{\lambda}{2}\Phi+4\lambda m+100\lambda mH.
\end{equation}

For comparison, a complete code bids $c_f$ on each of its selected fields and therefore wins each such
field outright.  Proposition~\ref{wta5:clm:local} gives zero-bid probability below $1/c_f$ for each stabilizer in the
present $\delta=5/8$ contest, so the bid-zero baseline fraction on field $f$ is below $1/(5c_f^2)$.
Consequently the baseline loss on a selected field is at most
$2\lambda/(5c_f)$ by \eqref{wta5:eq:valuebounds}.  Since the selected completion costs sum to $D$ and the
value deviations over all row-side fields sum to at most $100\lambda mH$, the adjusted payoff of a
complete code is at least
\[
\lambda D-100\lambda mH-\frac{2\lambda m}{5}
\ge \lambda D-101\lambda mH.
\]

\paragraph{Case 1: $\Phi\ge1000mH$.}
By \eqref{wta5:eq:defectbound}, in this case the gain from replacing the allocation by a complete code is at least
\[
\bigl(\lambda D-101\lambda mH\bigr)
-\bigl(\lambda D-400\lambda mH+4\lambda m\bigr)
=\lambda m(299H-4),
\]
Since $H\ge512$ and $m\ge1$, this is larger than $\lambda/10$.

\paragraph{Case 2: $\Phi<1000mH$.}
Let
\[
E=\{f\in\mathcal F_R:\bar t_f=c_f\}
\]
be the completed-field set of the rounded endpoint pattern.  Equations \eqref{wta5:eq:budget}--\eqref{wta5:eq:defect} imply
\[
\left|\sum_{f\in E}c_f-D\right|<1000mH<L-m,
\]
where the last inequality follows from $L=2^{20}mH$.  The moment-code lemma now does all the combinatorial work: $E$ must be one unique row code, say
$\wtaFiveR_a$.

It remains only to eliminate the small residual defects.  Writing $(x)_+=\max\{x,0\}$, troop conservation gives
the exact accounting identity
\begin{equation}\label{wta5:eq:accounting}
\sum_{f\in E}(c_f-t_f)_+
=\sum_{f\in E}(t_f-c_f)_+
 +\sum_{f\in\mathcal F_R\setminus E}t_f
 +t_{\rm buf}+(t_{\rm part}-1).
\end{equation}
The left-hand side is the missing mass on selected code fields; the terms on the right are overbids, stray
bids, buffer troops, and excess participation troops.  By \eqref{wta5:eq:cmin},
$1000mH/\min_fc_f<1000/2^{21}<1/2000$, so every deficit and stray bid is at most $c_f/2000$ on its
field.  Using the weaker consequence
$c_f[1-g_f(c_f-\Delta)]\ge\Delta/4$ of \eqref{wta5:eq:shape3}, filling a deficit of $\Delta$ troops gains at
least $\lambda\Delta/8$.  By \eqref{wta5:eq:shape2} and $w_f^R\le2\lambda c_f$, deleting a stray bid of $j$
troops costs less than $\lambda j/200$.  Overbids cost nothing to remove, a troop removed from
$\mathcal B_P$ loses at most $\kappa<\lambda/200$, and lowering $t_{\rm part}$ to one is payoff-neutral.
If the allocation is not already $\wtaFiveCodeR{a}$, Equation~\eqref{wta5:eq:accounting} has at least one missing troop and the
net improvement is therefore more than
$(1/8-1/200)\lambda>\lambda/10$.

Both cases therefore show that, conditional on $C$ being neutral, every noncanonical nonneutral row allocation
has a code deviation of gain at least $\lambda/10$.

\paragraph{Removing the conditioning on $C$.}
When $C$ is nonneutral, the payoff comparison above may change, but this event has probability at most
$x_C\le2h$ by Equation~\eqref{wta5:eq:early}.  For either fixed row allocation, changing $C$ from neutral to a
nonneutral allocation can only weakly decrease $R$'s payoff on the shared encoding and primary-buffer fields.
Therefore, in the difference ``payoff of the code deviation minus payoff of the current allocation'', the bad
effect of removing the conditioning is at most the payoff loss of the code itself.  Equation~\eqref{wta5:eq:outsidevalue}
bounds that raw payoff away from the participation field by $3\lambda P$.  Multiplying by
$x_C\le2h$, removing the conditioning changes the gain by at most $6h\lambda P$.  Since $h=P^{-20}$ and $P\ge100$,
\[
6h\lambda P=6P^{-19}\lambda<\frac{\lambda}{20}
\qquad\text{because}\qquad 120P^{-19}<1.
\]
Thus every noncanonical nonneutral row allocation has an unconditional code deviation of gain at least
$\lambda/20$.

Let $\delta_R$ and $\delta_C$ denote the probabilities that the two primaries assign to noncanonical
nonneutral allocations before this final repair.  Replacing each such realization by its corresponding
profitable code gives expected gain at least $\delta_i\lambda/20$.  Since the current raw regret is below
$P^{-939}$,
\[
\delta_R,\delta_C
\le \frac{20P^{-939}}{\lambda}
=20P^{-895}<P^{-894}.
\]
Repairing both primaries therefore changes total mass by less than $P^{-893}$ and leaves raw regret below
\[
P^{-939}+2P\cdot P^{-893}<P^{-891}.
\]
Together with the preceding regret-based repairs, the total probability mass replaced is below $P^{-892}$.
As stated in the lemma, this accounting excludes only the payoff-monotone private-buffer normalization, which
does not increase regret.  In particular, both $P^{-800}$ bounds stated in the lemma hold.  For the remainder
of the proof we use the convenient common upper bound
\begin{equation}\label{wta5:eq:rawregret}
\eps_{\rm raw}:=P^{-800}
\end{equation}
on the raw regret of the repaired profile.  Both primaries are now supported only on their neutral
allocation and exact codewords, completing the proof.
\end{proof}

\subsubsection{Calibration, payoff transfer, and decoding}\label{subsec:wta5-simulation}

At the repaired profile every active primary allocation is now an exact codeword, and
\eqref{wta5:eq:rawregret} gives raw regret at most $\eps_{\rm raw}=P^{-800}$.  The following lemma first calibrates
the activation probabilities and then transfers code payoffs to the source game.

\begin{lemma}[Calibration and single-intersection payoff transfer]\label{wta5:lem:simulation}
For each primary $i$,
\begin{equation}\label{wta5:eq:calibration}
|x_i-h|\le\frac{16\eps_{\rm raw}}{\lambda D},
\qquad
x_i\ge\frac h2.
\end{equation}
Define the decoded source strategies by
\[
\sigma_a=\frac{\mu_R(\wtaFiveCodeR{a})}{x_R},
\qquad
\tau_b=\frac{\mu_C(\wtaFiveCodeC{b})}{x_C}.
\]
Then, for every $a\in[q]_0$,
\begin{equation}\label{wta5:eq:rowsim}
U_R(\wtaFiveCodeR{a})=C_R+x_C\eta(A\tau)_a+\operatorname{err}_a^R,
\end{equation}
where $C_R$ is independent of $a$ and
\begin{equation}\label{wta5:eq:error}
|\operatorname{err}_a^R|\le20\lambda h(q+1)+4P\eps_{\rm raw}.
\end{equation}
The symmetric formula holds for every $b\in[q]_0$, with $x_R\eta(\sigma^\top B)_b$.
\end{lemma}

\begin{proof}
\wtaFiveLead{Calibration of activation.}
We first show why full inspection must occur with nonnegligible probability.  Fix a primary $i$ and any one
of its complete codes.  The balancing field of that code has completion bid larger than $D/2$: this follows
from the proof of Lemma~\ref{wta5:lem:code}, where every unscaled balancer exceeds $\widehat D/2$, together with
\eqref{wta5:eq:completion}.  By \eqref{wta5:eq:valuebounds}, its battlefield value is therefore larger than
$\lambda D/4$.  When the primary is neutral it bids zero on this field and receives at most the all-zero
share $1/5$ of that value, whereas the complete code bids the completion amount and, because stabilizer bids
are at most $c_f-1$, wins the field outright.  Thus the balancing field alone gives the code an advantage
larger than
\[
\frac45\cdot\frac{\lambda D}{4}
=\frac{\lambda D}{5}
>\frac{\lambda D}{8}.
\]
The code cannot lose payoff on its other encoding fields by changing its bid from zero to the completion bid.
On the participation field, the neutral allocation and the code receive the same payoff unless $J$ uses the
corresponding full-inspection action $I_i$; in that event neutral play ties the inspector while the code loses,
so the code's expected participation loss relative to neutral play is exactly $\iota_i/2$.  Hence the complete
code improves on the neutral allocation by at least
$\frac{\lambda D}{8}-\frac{\iota_i}{2}$.

By Lemma~\ref{wta5:lem:canonical}, $x_i\le2h<1/2$, so the neutral allocation has probability
$\mu_i(N_i)=1-x_i>1/2$.  If a pure action used with probability at least $1/2$ were more than
$2\eps_{\rm raw}$ below a best response, its contribution alone would exceed the player's raw regret
$\eps_{\rm raw}$.  Therefore no deviation from $N_i$ can improve payoff by more than
$2\eps_{\rm raw}$.  Applying this to the complete-code deviation gives
\[
\frac{\lambda D}{8}-\frac{\iota_i}{2}
\le2\eps_{\rm raw},
\]
and hence
\[
\iota_i\ge\frac{\lambda D}{4}-4\eps_{\rm raw}.
\]
Since $\eps_{\rm raw}=P^{-800}$, $\lambda=P^{-44}$, and $D\ge1$, we have
$4\eps_{\rm raw}<\lambda D/8$.  Consequently
\begin{equation}\label{wta5:eq:qlower}
\iota_i\ge\frac{\lambda D}{8}.
\end{equation}
Lemma~\ref{wta5:lem:canonical} already gives $\iota_R+\iota_C\le1/8$, and after canonicalization
$J$ uses only $I_R,I_C,I_{\wtaFiveOut}$.  Hence $I_{\wtaFiveOut}$ has probability at least $7/8$.  By the positive-participation
repair in Lemma~\ref{wta5:lem:canonical}, every primary now bids positively on its participation field in every
nonneutral realization; neutral play also bids positively there.  Thus $J$ receives no participation-field
all-zero share under $I_{\wtaFiveOut}$.  At the canonical profile, after subtracting the common all-zero private-buffer
baseline, the relevant inspector payoffs are
\[
U_J(I_{\wtaFiveOut})=\frac{1+h}{2},
\qquad
U_J(I_i)=\frac{1+x_i}{2}.
\]
Suppose first that $x_i\ge h$.  Replacing every realization of $I_{\wtaFiveOut}$ by $I_i$ is a feasible unilateral
deviation for the inspector, and on each replaced realization it gains $(x_i-h)/2$.  Since $I_{\wtaFiveOut}$ has
probability at least $7/8$ and the inspector's raw regret is at most $\eps_{\rm raw}$,
\[
\frac78\cdot\frac{x_i-h}{2}\le\eps_{\rm raw}.
\]
Thus $x_i-h\le16\eps_{\rm raw}/7<3\eps_{\rm raw}$.

Suppose instead that $x_i<h$.  Replacing every realization of $I_i$ by $I_{\wtaFiveOut}$ gains
$(h-x_i)/2$ on each replaced realization.  By \eqref{wta5:eq:qlower}, $I_i$ has probability at least
$\lambda D/8$, so the regret bound gives
\[
\frac{\lambda D}{8}\cdot\frac{h-x_i}{2}
\le\iota_i\frac{h-x_i}{2}
\le\eps_{\rm raw}.
\]
Therefore $h-x_i\le16\eps_{\rm raw}/(\lambda D)$.  Combining the two cases yields
$|x_i-h|\le16\eps_{\rm raw}/(\lambda D)$.  Since
$16\eps_{\rm raw}/(\lambda D)\le16P^{-756}<h/2$, we also have $x_i\ge h/2$.  This proves
\eqref{wta5:eq:calibration}.

\wtaFiveLead{Decoded source strategies.}
The definitions of $\sigma$ and $\tau$ in the lemma indeed give probability distributions on $[q]_0$.
By Lemma~\ref{wta5:lem:canonical}, all nonneutral mass of $R$ is carried by the allocations
$\wtaFiveCodeR{a}$, $a\in[q]_0$, and all nonneutral mass of $C$ is carried by the allocations
$\wtaFiveCodeC{b}$, $b\in[q]_0$.  Moreover \eqref{wta5:eq:calibration} gives $x_R,x_C\ge h/2>0$.  Hence
\[
\begin{aligned}
\sum_{a\in[q]_0}\sigma_a
&=\frac{\sum_{a\in[q]_0}\mu_R(\wtaFiveCodeR{a})}{x_R}
=\frac{1-\mu_R(N_R)}{x_R}=1,\\
\sum_{b\in[q]_0}\tau_b
&=\frac{\sum_{b\in[q]_0}\mu_C(\wtaFiveCodeC{b})}{x_C}
=\frac{1-\mu_C(N_C)}{x_C}=1.
\end{aligned}
\]

\wtaFiveLead{Localization of the source payoff to the pair field $f_{ab}$.}
Fix $a\in[q]_0$ and consider the canonical row allocation $\wtaFiveCodeR{a}$.  Let $\rho_f^R,\rho_f^C$ denote the unconditional probabilities that the two primaries
visit encoding field $f$.  Whenever a primary visits $f$, it bids $c_f$ and therefore beats both
stabilizers.  The stabilizers' effective local prize is
\[
(c_f-3/8)(1-\rho_f^R)(1-\rho_f^C)=c_f-1+\delta_f^{\rm eff}.
\]
Expanding the definition of the effective prize gives
\[
\delta_f^{\rm eff}
=\frac58-(c_f-3/8)
\bigl(\rho_f^R+\rho_f^C-\rho_f^R\rho_f^C\bigr).
\]
The subtracted term is nonnegative, so immediately $\delta_f^{\rm eff}\le5/8$.  For the lower bound,
$\rho_f^R+\rho_f^C\le x_R+x_C\le4h$ by Lemma~\ref{wta5:lem:canonical}, and therefore
\[
0\le\frac58-\delta_f^{\rm eff}
\le(c_f-3/8)(\rho_f^R+\rho_f^C)
<4hc_f.
\]
Since $h=P^{-20}$ and $c_f<P/100$, we have
\[
4hc_f<\frac{P^{-19}}{25}<\frac18.
\]
Thus the effective prize remains in the interval required by Proposition~\ref{wta5:clm:local}:
\begin{equation}\label{wta5:eq:deltaeff}
\frac12<\delta_f^{\rm eff}\le\frac58.
\end{equation}

We now verify that Proposition~\ref{wta5:clm:local} applies to the actual stabilizer contest at $f$.  After
canonicalization, a primary that visits $f$ bids $c_f$, while every stabilizer bid is at most $c_f-1$.
Thus a stabilizer can receive battlefield payoff on $f$ only when neither primary visits.  Since the players
randomize independently, multiplying the original battlefield value $c_f-3/8$ by the probability of this
event gives exactly the effective prize $c_f-1+\delta_f^{\rm eff}$ above.  Against the other stabilizer's
marginal, a positive bid $j$ therefore has payoff
\[
(c_f-1+\delta_f^{\rm eff})
\left(\Pr[\text{other bid}<j]+\frac12\Pr[\text{other bid}=j]\right)-j,
\]
while bid zero has payoff
\[
\frac{c_f-1+\delta_f^{\rm eff}}5\Pr[\text{other bid}=0].
\]
These are precisely the payoffs in Proposition~\ref{wta5:clm:local} with $c=c_f$ and
$\delta=\delta_f^{\rm eff}$.  Moreover, after the stabilizer-bid repair in Lemma~\ref{wta5:lem:canonical}, a
stabilizer can change its bid on one encoding field and compensate the troop difference through its private
buffers.  Hence its local regret on this one-field game is at most the global raw regret
$\eps_{\rm raw}$.  Finally, $c_f<P/100$ gives
\[
10^8c_f^4\eps_{\rm raw}
<10^8\left(\frac{P}{100}\right)^4P^{-800}
=P^{-796}<1,
\]
so all hypotheses of Proposition~\ref{wta5:clm:local} hold.  Using \eqref{wta5:eq:deltaeff}, Proposition~\ref{wta5:clm:local}
therefore gives, for each stabilizer,
\[
\left|\Pr[\text{bid }0]-\frac{\delta_f^{\rm eff}}{c_f-1+\delta_f^{\rm eff}}\right|
\le10^4c_f^2\eps_{\rm raw}.
\]

We next compare this zero-bid probability with the reference value used in the construction.  Direct
differentiation gives
\[
\frac{d}{d\delta}\left(\frac{\delta}{c_f-1+\delta}\right)
=\frac{c_f-1}{(c_f-1+\delta)^2}
\le\frac1{c_f-1}
\le\frac2{c_f}
\qquad\left(\frac12\le\delta\le\frac58\right).
\]
The mean-value theorem and the bound
$0\le5/8-\delta_f^{\rm eff}<4hc_f$ therefore imply
\[
\left|
\frac{\delta_f^{\rm eff}}{c_f-1+\delta_f^{\rm eff}}
-\frac{5}{8(c_f-3/8)}
\right|
\le\frac2{c_f}\left|\delta_f^{\rm eff}-\frac58\right|
<8h.
\]
Combining this estimate with Proposition~\ref{wta5:clm:local} by the triangle inequality gives, for each stabilizer,
\[
\left|\Pr[\text{bid }0]-\frac{5}{8(c_f-3/8)}\right|
\le8h+10^4c_f^2\eps_{\rm raw}.
\]

Consequently, if
\[
\beta_f=\frac15\Pr[S_f=0]\Pr[T_f=0]
\]
is the actual all-zero payoff fraction of a nonvisiting primary, then \eqref{wta5:eq:b0} and
$5/[8(c_f-3/8)]\le1/c_f$ give a quantitative comparison with $\beta_f^0$.  The preceding bound and the parameter choices also imply
\[
c_f\left(8h+10^4c_f^2\eps_{\rm raw}\right)
<\frac{2}{25}P^{-19}+\frac{1}{100}P^{-797}<1,
\]
and hence $8h+10^4c_f^2\eps_{\rm raw}<1/c_f$.  Each actual stabilizer zero-bid probability is
therefore at most $2/c_f$.  Hence
\[
\begin{aligned}
|\beta_f-\beta_f^0|
&\le\frac15\left|\Pr[S_f=0]-\frac{5}{8(c_f-3/8)}\right|\Pr[T_f=0]\\
&\quad+\frac{1}{5}\frac{5}{8(c_f-3/8)}
\left|\Pr[T_f=0]-\frac{5}{8(c_f-3/8)}\right|\\
&\le\frac{3}{5c_f}
\left(8h+10^4c_f^2\eps_{\rm raw}\right)
\le\frac{5h}{c_f}+6000c_f\eps_{\rm raw}.
\end{aligned}
\]
Thus
\begin{equation}\label{wta5:eq:bclose}
|\beta_f-\beta_f^0|\le\frac{5h}{c_f}+6000c_f\eps_{\rm raw}.
\end{equation}

We now compute the payoff of the fixed allocation $\wtaFiveCodeR{a}$.  The payoff that $R$ would obtain by
bidding zero on every row-valued encoding field is independent of $a$ and will be absorbed into the
constant $C_R$.  If $f\in\wtaFiveR_a$ is selected by $\wtaFiveCodeR{a}$, then relative to this zero-bid baseline its
incremental contribution is
\[
w_f^R\left[
1-\frac{\rho_f^C}{2}-(1-\rho_f^C)\beta_f
\right].
\]
Replacing $\beta_f$ by $\beta_f^0$ changes this contribution by at most
$w_f^R|\beta_f-\beta_f^0|$.  By \eqref{wta5:eq:valuebounds} and \eqref{wta5:eq:bclose},
\[
w_f^R|\beta_f-\beta_f^0|
\le2\lambda c_f\left(\frac{5h}{c_f}+6000c_f\eps_{\rm raw}\right)
\le10\lambda h+\eps_{\rm raw},
\]
where the last inequality uses
$12000\lambda c_f^2\eps_{\rm raw}<1.2P^{-42}\eps_{\rm raw}<\eps_{\rm raw}$.
Since $q+1\le m<M<P/100$, the set $\wtaFiveR_a$ contains fewer than $P/100$ fields, so replacing
$\beta_f$ by $\beta_f^0$ on all selected fields creates total error at most
\[
10\lambda h(q+1)+\frac{P\eps_{\rm raw}}{100}.
\]
With $\beta_f^0$ in place of $\beta_f$, the incremental contribution of a selected field becomes
\[
w_f^R(1-\beta_f^0)
-\rho_f^C w_f^R\left(\frac12-\beta_f^0\right).
\]
The column primary never visits the row balancing field, whereas on the pair field $f_{ab}$ its visit
probability is
\[
\rho_{ab}^C=\mu_C(\wtaFiveCodeC{b})=x_C\tau_b.
\]
At the repaired profile, $\wtaFiveCodeR{a}$ receives participation payoff $1-\iota_R$ and, because neither
primary uses $\mathcal B_P$ after canonicalization, the all-zero payoff from that bank is $M\kappa/4$.
Moreover, if $R$ bid zero on every field of $\mathcal F_R$, its encoding-field payoff would be
\[
\sum_{f\in\mathcal F_R}w_f^R(1-\rho_f^C)\beta_f,
\]
which is independent of $a$.  Adding the incremental contributions of the selected fields, replacing
$\beta_f$ by $\beta_f^0$ there, and then using \eqref{wta5:eq:balvalue} and \eqref{wta5:eq:pairvalue} therefore gives
\begin{equation}\label{wta5:eq:pretransfer}
\begin{aligned}
\Bigg|U_R(\wtaFiveCodeR{a})
&-\Bigg[1-\iota_R+\frac{M\kappa}{4}
+\sum_{f\in\mathcal F_R}w_f^R(1-\rho_f^C)\beta_f\\
&\qquad\quad+\lambda D+h\lambda u_a
-x_C\sum_b\tau_b\left(\frac{\lambda c_{ab}}2-\eta A_{ab}\right)\Bigg]\Bigg|\\
&\le10\lambda h(q+1)+\frac{P\eps_{\rm raw}}{100}.
\end{aligned}
\end{equation}

It remains to separate the source-dependent term from the deterministic code terms.  By
\eqref{wta5:eq:additive},
\[
\sum_b\tau_b\frac{\lambda c_{ab}}2
=\lambda\theta+\lambda u_a+\lambda\sum_b\tau_bv_b.
\]
Substituting this identity into \eqref{wta5:eq:pretransfer}, define the quantity independent of $a$
\[
C_R=
1-\iota_R+\frac{M\kappa}{4}
+\sum_{f\in\mathcal F_R}w_f^R(1-\rho_f^C)\beta_f
+\lambda D-x_C\lambda\theta-x_C\lambda\sum_b\tau_bv_b.
\]
Equation~\eqref{wta5:eq:pretransfer} then becomes
\[
\left|U_R(\wtaFiveCodeR{a})
-\left[C_R+(h-x_C)\lambda u_a+x_C\eta(A\tau)_a\right]\right|
\le10\lambda h(q+1)+\frac{P\eps_{\rm raw}}{100}.
\]
The remaining structural term satisfies
\[
|(h-x_C)\lambda u_a|
\le\frac{16\eps_{\rm raw}}{\lambda D}\lambda\frac{D}{2q}
\le8\eps_{\rm raw}
\]
by \eqref{wta5:eq:calibration} and $u_a<D/(2q)$.  Therefore, defining $\operatorname{err}_a^R$ through the
payoff identity \eqref{wta5:eq:rowsim}, we obtain
\[
|\operatorname{err}_a^R|\le20\lambda h(q+1)+4P\eps_{\rm raw},
\]
as claimed.  The column calculation is symmetric.
\end{proof}

\subsubsection{Decoding the source equilibrium}\label{subsec:wta5-decoding}

For decoded source strategies $\sigma,\tau$, define the two source-game regrets by
\[
\wtaFiveReg_R(\sigma,\tau)=\max_{a\in[q]_0}(A\tau)_a-\sigma^\top A\tau,
\qquad
\wtaFiveReg_C(\sigma,\tau)=\max_{b\in[q]_0}(\sigma^\top B)_b-\sigma^\top B\tau.
\]
Thus $\wtaFiveReg_R$ and $\wtaFiveReg_C$ are the ordinary additive Nash regrets in the original bimatrix game.

Let $a^*$ be a best response to $\tau$ in the source game.  Player $R$ can leave its neutral mass
unchanged and, whenever it would play an active allocation $\wtaFiveCodeR{a}$, replace that allocation by
$\wtaFiveCodeR{a^*}$.  Conditional on active row play, the current distribution over indices is $\sigma$.
Hence the expected gain of this deviation on the active realizations is
\[
x_R\left(
U_R(\wtaFiveCodeR{a^*})-\sum_a\sigma_aU_R(\wtaFiveCodeR{a})
\right).
\]
Using the payoff identity \eqref{wta5:eq:rowsim}, its source-dependent part is exactly
\[
x_Rx_C\eta\left((A\tau)_{a^*}-\sum_a\sigma_a(A\tau)_a\right)
=x_Rx_C\eta\,\wtaFiveReg_R(\sigma,\tau).
\]
By \eqref{wta5:eq:error}, the difference between the error term of $a^*$ and the $\sigma$-average of the
current error terms has absolute value at most
\[
2\bigl(20\lambda h(q+1)+4P\eps_{\rm raw}\bigr).
\]
Therefore the deviation improves the row player's raw expected payoff by at least
\[
x_Rx_C\eta\,\wtaFiveReg_R(\sigma,\tau)
-2x_R\bigl(20\lambda h(q+1)+4P\eps_{\rm raw}\bigr).
\]
The repaired profile has raw regret at most $\eps_{\rm raw}$, so this gain cannot exceed
$\eps_{\rm raw}$.  Rearranging gives
\[
x_Rx_C\eta\,\wtaFiveReg_R(\sigma,\tau)
\le
2x_R\bigl(20\lambda h(q+1)+4P\eps_{\rm raw}\bigr)+\eps_{\rm raw}.
\]
Now $x_R\le2h$, $x_R,x_C\ge h/2$, and $\eta=H\lambda$.  Dividing by
$x_Rx_C\eta$ gives
\[
\wtaFiveReg_R(\sigma,\tau)
\le
\frac{80(q+1)}{H}
+\frac{16P\eps_{\rm raw}}{h\eta}
+\frac{4\eps_{\rm raw}}{h^2\eta}.
\]
Since $16Ph<1$, the last two terms are at most
$8\eps_{\rm raw}/(h^2\eta)$, and therefore
\[
\wtaFiveReg_R(\sigma,\tau)
\le\frac{80(q+1)}{H}+\frac{8\eps_{\rm raw}}{h^2\eta}
=\frac{5}{16}q^{-1}+\frac{8P^{-716}}{H}.
\]
Since $H=256(q+1)q$, $q\ge2$, and $P\ge100$,
\[
\frac{8P^{-716}}{H}\le\frac1{96}q^{-1}.
\]
Hence
\[
\wtaFiveReg_R(\sigma,\tau)\le\frac{31}{96}q^{-1}<\frac12q^{-1}.
\]
The same calculation gives $\wtaFiveReg_C(\sigma,\tau)<q^{-1}/2$.  Thus $(\sigma,\tau)$ is a
$q^{-1}$-Nash equilibrium of the source game.  This completes the reduction and proves
Theorem~\ref{thm:wta-five-player-hardness}.

%% file: conclusion.tex
\section{Conclusions and discussion}
We studied several computational boundaries for Nash equilibrium in discrete Colonel Blotto. A general flow-based reduction to concave games first establishes \PPAD membership for multiplayer Colonel Blotto with arbitrary local reward functions under standard polynomial representation assumptions. For two-player general-sum games, we established \PPAD-completeness for monotone piecewise-constant battlefield payoffs, including instances with exactly two battlefields and, with more battlefields, single-threshold payoffs that are $2$-piecewise constant in either player's allocation. Complementing these hardness results, we identified intrinsic social interaction rank as a tractability parameter and showed that, for every fixed rank, an $\eps$-Nash equilibrium can be computed in time polynomial in the number of battlefields, the numerical budgets, and $1/\eps$. Together, these results delineate a frontier within two-player Colonel Blotto between structured low-rank interactions and the higher-rank general-sum regime.

We then turned to literal multiplayer winner-takes-all Colonel Blotto and identified a qualitatively different frontier driven by tie resolution: full reward to every co-winner admits an exact polynomial-time pure-equilibrium algorithm, whereas equal splitting among co-winners makes inverse-polynomial approximate Nash computation \PPAD-complete even with only five players. This settles the constant-player hardness question left open in the concurrent work of Bichler and Ghosh~\cite{bichlerghosh2026}.

An immediate open question is the complexity of \emph{discrete two-player general-sum WTA Colonel Blotto}. Non-zero-sum and heterogeneous-value WTA variants have a substantial equilibrium-theory literature: equilibrium characterizations are known for continuous asymmetric-value models and related relaxations \cite{kovenock2021generalizations,thomas2018heterogeneous,robersonkvasov2012,perchet2022algorithmic}. What remains unresolved by these results is the computational complexity of finding an inverse-polynomial approximate \emph{mixed} Nash equilibrium in the general discrete model with player-specific battlefield values. Is such an equilibrium computable in polynomial time, or does the classical WTA local structure already suffice for \PPAD-hardness once constant-sum symmetry is lost? Resolving this question would further sharpen the boundary between the tractable classical two-player WTA model and the general-sum hardness phenomena established in this paper.

%% file: appendix.tex
\appendix

\section{Missing Proofs of Section \ref{sec:low-rank}}

\subsection{Proof of Proposition \ref{prop:s-path}}\label{app: 5.1}

\begin{proof}
A source--sink path uses exactly one edge between consecutive layers. If its labels are $z_1,\ldots,z_k$, then the second coordinate after layer $b$ is $\sum_{j\le b}z_j$; reaching $(k,B)$ is equivalent to $\sum_bz_b=B$. Conversely, an allocation determines all cumulative sums and hence a unique path.

A source--sink path \(P\) is identified with its edge-incidence vector
\(\mathbf{1}_P\in\{0,1\}^{E}\), where
\[
(\mathbf{1}_P)_e=
\begin{cases}
1, & e\in P,\\
0, & e\notin P.
\end{cases}
\]
By the definition of the reduced incidence matrix \(D\), every
source--sink path satisfies
\[
D\mathbf{1}_P=d.
\]
Indeed, exactly one edge of \(P\) leaves the source, so the source
coordinate of \(D\mathbf{1}_P\) is \(1\); at every retained
intermediate node of \(P\), one path edge enters and one path edge
leaves, so the corresponding coordinate is \(0\). The sink equation
is absent because the sink row has been deleted.

Now let \(\pi\) be a probability distribution over source--sink
paths. The induced edge flow is
\[
f=\sum_{P}\pi(P)\mathbf{1}_P.
\]
Clearly \(f\ge 0\), and by linearity,
\[
Df
=
\sum_P \pi(P)D\mathbf{1}_P
=
\sum_P \pi(P)d
=
d.
\]
Hence \(f\in F(k,B)\).

Conversely, let \(f\in F(k,B)\). We show that \(f\) can be written as
a convex combination of source--sink path-incidence vectors.

We use the following elementary observation. Suppose that
\(r\in\mathbb{R}_{\ge 0}^{E}\) satisfies
\[
Dr=\alpha d
\]
for some \(\alpha>0\). 
Then the support of \(r\) contains a
source--sink path. 
Let $\delta^+(v)$ denote the set of edges leaving the node \(v\) and $\delta^-(v)$ denote the set of edges pointing to the node \(v\).
Since the source $s$ has no incoming edges and its
flow-balance equation is
\[
\sum_{e\in\delta^+(s)} r_e=\alpha>0,
\]
there is an outgoing edge from the source with positive \(r\)-flow.
Suppose we have reached an intermediate node \(v\) along a positive-flow
edge. Then \(v\) has positive inflow. Since its right-hand side in
\(Dr=\alpha d\) is zero, flow conservation gives
\[
\sum_{e\in\delta^+(v)}r_e
=
\sum_{e\in\delta^-(v)}r_e
>0.
\]
Thus some outgoing edge of \(v\) also has positive flow. Repeating this
argument produces a directed path consisting entirely of positive-flow
edges. Because \(G(k,B)\) is a finite layered DAG and every edge moves
from layer \(b-1\) to layer \(b\), this process cannot revisit a node
and must eventually reach the sink. Hence the support of \(r\)
contains a source--sink path.

We now construct the decomposition iteratively. Set
\[
r^{(0)}=f,
\qquad
\alpha_0=1.
\]
As long as \(r^{(q)}\neq 0\), choose a source--sink path
\(P_q\) contained in the support of \(r^{(q)}\), and define
\[
\gamma_q
=
\min_{e\in P_q} r^{(q)}_e
>0.
\]
Set
\[
r^{(q+1)}
=
r^{(q)}-\gamma_q\mathbf{1}_{P_q},
\qquad
\alpha_{q+1}
=
\alpha_q-\gamma_q.
\]
By the definition of \(\gamma_q\), we have
\(r^{(q+1)}\ge 0\). Moreover,
\[
Dr^{(q+1)}
=
Dr^{(q)}
-
\gamma_qD\mathbf{1}_{P_q}
=
\alpha_qd-\gamma_qd
=
\alpha_{q+1}d.
\]
Thus the invariant
\[
r^{(q)}\ge 0,
\qquad
Dr^{(q)}=\alpha_q d
\]
is preserved at every iteration.

Furthermore, at least one edge \(e\in P_q\) satisfies
\(r^{(q)}_e=\gamma_q\), and hence
\(r^{(q+1)}_e=0\). Since all residual flows are nonnegative, an edge
that has become zero can never become positive again. Therefore the
number of positive-flow edges decreases by at least one at each
iteration. Consequently, after at most \(|E|\) iterations the procedure
terminates with
\[
r^{(m)}=0
\]
for some \(m\le |E|\).

Unrolling the recursion gives
\[
f
=
\sum_{q=0}^{m-1}\gamma_q\mathbf{1}_{P_q}.
\]
Applying \(D\) to both sides yields
\[
d
=
Df
=
\sum_{q=0}^{m-1}\gamma_qD\mathbf{1}_{P_q}
=
\left(\sum_{q=0}^{m-1}\gamma_q\right)d.
\]
Since \(d\neq 0\), it follows that
\[
\sum_{q=0}^{m-1}\gamma_q=1.
\]
Hence the coefficients \(\gamma_q\) form a probability distribution
over at most \(|E|\) source--sink paths, and \(f\) is the corresponding
convex combination of their incidence vectors.

Finally, if \(f\) is rational, then every residual \(r^{(q)}\) is
rational. Therefore
\[
\gamma_q=\min_{e\in P_q}r^{(q)}_e
\]
is rational at every iteration, so the resulting path decomposition
has rational weights.
\end{proof}

\subsection{Proof of Proposition \ref{prop:s-br}}\label{app: 5.2}

\begin{proof}
Induction on $b$ shows that $D_1(b,h)$ is the maximum payoff obtained by distributing exactly $h$ units among the first $b$ battlefields. The recurrence enumerates the allocation on battlefield $b$ and combines it with an optimal prefix. Summing the $h+1$ choices over all states gives $\mathcal{O}(kB_1^2)$ operations after $\mathcal{O}(kB_1B_2)$ operations to form the expected battlefield payoffs. The argument for player $2$ is symmetric.
\end{proof}

\subsection{Proof of Lemma \ref{lem: feasible_dir}}\label{app: 5.3}

\begin{proof}
We first prove that
\[
\operatorname{lin}(F_p-F_p)=\ker D_p.
\]

For the inclusion from left to right, let \(f,f'\in F_p\). Since both
flows satisfy the same flow-conservation equations,
\[
D_pf=d_p,
\qquad
D_pf'=d_p,
\]
and therefore
\[
D_p(f-f')=0.
\]
Hence every element of \(F_p-F_p\) belongs to \(\ker D_p\). Since
\(\ker D_p\) is a linear subspace, it follows that
\[
\operatorname{lin}(F_p-F_p)\subseteq \ker D_p.
\]

For the reverse inclusion, we first construct a feasible flow that is
strictly positive on every edge. Every edge \(e\in E_p\) lies on at
least one source--sink path; choose one such path \(P_e\), and let
\(\mathbf{1}_{P_e}\in\{0,1\}^{E_p}\) denote its edge-incidence vector.
Define
\[
\bar f
:=
\frac{1}{|E_p|}
\sum_{e\in E_p}\mathbf{1}_{P_e}.
\]
Each \(\mathbf{1}_{P_e}\) is a unit source--sink flow and therefore
satisfies
\[
D_p\mathbf{1}_{P_e}=d_p.
\]
Consequently,
\[
D_p\bar f
=
\frac{1}{|E_p|}
\sum_{e\in E_p}D_p\mathbf{1}_{P_e}
=
d_p,
\]
so \(\bar f\in F_p\). Moreover, since \(e\in P_e\),
\[
\bar f_e\ge \frac{1}{|E_p|}>0
\qquad
\text{for every }e\in E_p.
\]

Now let \(z\in\ker D_p\). Since \(\bar f\) is strictly positive
coordinatewise, there exists \(\varepsilon>0\) sufficiently small such
that
\[
\bar f+\varepsilon z\ge 0.
\]
Furthermore,
\[
D_p(\bar f+\varepsilon z)
=
D_p\bar f+\varepsilon D_pz
=
d_p,
\]
and hence
$\bar f+\varepsilon z\in F_p$. 
Therefore
\[
z
=
\frac{1}{\varepsilon}
\bigl((\bar f+\varepsilon z)-\bar f\bigr)
\in
\operatorname{lin}(F_p-F_p).
\]
This proves
\[
\operatorname{lin}(F_p-F_p)=\ker D_p.
\]

We next pass from flow directions to marginal directions. Recall that
\[
\mathcal M_p=T_p(F_p)
\qquad\text{and}\qquad
\mathcal L_p
=
\operatorname{lin}(\mathcal M_p-\mathcal M_p).
\]
Since \(T_p\) is linear,
\[
\mathcal M_p-\mathcal M_p
=
\{T_pf-T_pf':f,f'\in F_p\}
=
T_p(F_p-F_p).
\]
Hence
\[
\begin{aligned}
\mathcal L_p
&=
\operatorname{lin}\bigl(T_p(F_p-F_p)\bigr)\\
&=
T_p\bigl(\operatorname{lin}(F_p-F_p)\bigr)\\
&=
T_p(\ker D_p).
\end{aligned}
\]

Finally, a rational basis of \(\mathcal L_p\) can be computed in
polynomial time. Using rational Gaussian elimination, compute a matrix
\(Z_p\) whose columns form a basis of \(\ker D_p\). Then the columns of
$T_pZ_p$
span \(T_p(\ker D_p)=\mathcal L_p\). A second rational Gaussian
elimination extracts a maximal linearly independent subset of these
columns, yielding a full-column-rank rational matrix \(L_p\) whose
columns form a basis of \(\mathcal L_p\). All these operations are
polynomial time in the dimensions and binary encoding lengths of
\(D_p\) and \(T_p\).

\end{proof}